\documentclass[noinfoline]{imsart}

\RequirePackage[OT1]{fontenc}
\RequirePackage{amsthm,amsmath}
\RequirePackage[numbers]{natbib}
\RequirePackage[colorlinks,citecolor=blue,urlcolor=blue]{hyperref}

\usepackage{verbatim}

\allowdisplaybreaks

\usepackage{fullpage}

\newcommand {\R}{\mathbb{R}}
\newcommand {\N}{\mathbb{N}}
\renewcommand {\O}{\mathcal{O}}

\newcommand{\E}{\mathbb{E}}

\newcommand {\cov}{\textrm{Cov}}
\newcommand {\topavg}{\textrm{Ta}}
\newcommand {\var}{\textrm{Var}}

\newcommand{\verti}[1]{{\left\vert\kern-0.25ex\left\vert\kern-0.25ex\left\vert #1 
    \right\vert\kern-0.25ex\right\vert\kern-0.25ex\right\vert}}
\newcommand{\ct}{\,\tilde\otimes\,}
\newcommand{\Tr}[1]{\text{Tr}_{#1}}
\newcommand{\T}[1]{\text{T}_{#1}}
\renewcommand{\vec}[1]{\,\mathrm{vec}( #1 )}

\renewcommand{\t}[1]{#1^\top}
\newcommand{\argmin}{\operatornamewithlimits{arg\,min}}

\DeclareMathOperator{\sign}{sign}
\DeclareMathOperator{\diag}{diag}

\usepackage{subfigure}
\usepackage{layout}

\usepackage{dsfont}
\usepackage[mathscr]{eucal}
\usepackage[toc,page]{appendix}
\usepackage{mathrsfs}
\usepackage{color}
\usepackage{pifont}
\usepackage{bm}
\usepackage{latexsym}
\usepackage{amsfonts}
\usepackage{amssymb}
\usepackage{epsfig}
\usepackage{graphicx}
\usepackage{multirow}

\newtheorem{theorem}{Theorem}
\newtheorem{proposition}{Proposition}
\newtheorem{corollary}{Corollary}
\newtheorem{lemma}{Lemma}
\newtheorem{definition}{Definition}
\newtheorem{remark}{Remark}
\newtheorem{example}{Example}

\startlocaldefs
\numberwithin{equation}{section}
\theoremstyle{plain}
\endlocaldefs

\begin{document}

\begin{frontmatter}

\title{{\large Random Surface Covariance Estimation by Shifted Partial Tracing}}

\runtitle{Random Surface Covariance Estimation\\ by Shifted Partial Tracing}

\begin{aug}
\author{\fnms{Tomas} \snm{Masak}\ead[label=e1]{tomas.masak@epfl.ch}} \and
\author{\fnms{Victor M.} \snm{Panaretos}\ead[label=e2]{victor.panaretos@epfl.ch}}

\thankstext{t1}{Research supported by a Swiss National Science Foundation grant.}

\runauthor{T. Masak \& V.M. Panaretos}

\affiliation{Ecole Polytechnique F\'ed\'erale de Lausanne}

\address{Institut de Math\'ematiques\\
Ecole Polytechnique F\'ed\'erale de Lausanne\\
e-mail: \href{mailto:tomas.masak@epfl.ch}{tomas.masak@epfl.ch}, \href{mailto:victor.panaretos@epfl.ch}{victor.panaretos@epfl.ch}
}

\end{aug}

\begin{abstract} 
The problem of covariance estimation for replicated surface-valued processes is examined from the functional data analysis perspective. Considerations of statistical and computational efficiency often compel the use of separability of the covariance, even though the assumption may fail in practice. We consider a setting where the covariance structure may fail to be separable locally -- either due to noise contamination or due to the presence of a~non-separable short-range dependent signal component. That is, the covariance is an additive perturbation of a separable component by a~non-separable but banded component.  We introduce non-parametric estimators hinging on the novel concept of shifted partial tracing, enabling computationally efficient estimation of the model under dense observation. Due to the denoising properties of shifted partial tracing, our methods are shown to yield consistent estimators even under noisy discrete observation, without the need for smoothing. Further to deriving the convergence rates and limit theorems, we also show that the implementation of our estimators, including prediction, comes at no computational overhead relative to a separable model. Finally, we demonstrate empirical performance and computational feasibility of our methods in an extensive simulation study and on a real data set.
\end{abstract}



\end{frontmatter}

\tableofcontents

\newpage
\section{Introduction} \label{sec:intro}

\emph{Functional Data Analysis} (FDA, \cite{ramsay2002,hsing2015}) focusses on the problem of statistical inference on the law of a random process $X(u):[0,1]^D\rightarrow\mathbb{R}$ given multiple realisations thereof. The process realisations are treated as elements of a separable Hilbert space $\mathcal H$ of functions on $[0,1]^D$ (e.g.~$\mathcal{L}^2[0,1]^D$). FDA covers the full gamut of statistical tasks, including regression, classification, and testing, to name a few. In any of these problems, the \emph{covariance operator} $C:{\mathcal H}\rightarrow {\mathcal H}$ of the random function $X(u)$ is elemental. This trace-class integral operator with kernel $c(u_1,u_2)=\mbox{cov}\{X(u_1),X(u_2)\}$, encodes the second-order characteristics of $X(u)$ and its associated spectral decomposition
is at the core of many (or even most) FDA inferential methods.
Consequently, the efficient estimation of the covariance operator $C$ (or equivalently its kernel~$c$) associated with $X$ is a fundamental task in FDA, on which further methodology can be based. This is to be done on the basis of $N$ i.i.d. realisations of the random process $X$, say $\{X_1,\ldots,X_N\}$. One wishes to do so \emph{nonparametrically}, since the availability of replicated realisations should allow so. When $D=1$, it is fair to say that this is entirely feasible and well understood, under a broad range of observation regimes (see  \cite{wang2016} for a comprehensive overview).

Though conceptually similar, things are much less straightforward in the case of random surfaces, i.e.~when $D=2$, which is the case we focus on in this paper. In this case, one faces additional challenging limitations to statistical and computational efficiency when attempting to nonparametrically estimate $c:[0,1]^4\rightarrow \mathbb{R}$ on the basis of $N$ replications (see \cite[\S 1]{aston2017} for a detailed discussion). The number of grid points on which $c$ is measured may even exceed $N$, especially in densely observed functional data scenarios. Worse still, one may not be able to even store the empirical covariance, much less invert it. To appreciate this, assume that each of the $N$ i.i.d. surfaces $\{X_n(s,t)\}$ are measured on a common grid of size $K_1 \times K_2$ over $[0,1]^2$. That is, the data corresponding to a single realization $X_n$ form a matrix $\mathbf X_n \in \R^{K_1 \times K_2}$ and the raw empirical covariance is represented by the tensor $\mathbf{C} \in \R^{K_1 \times K_2 \times K_1 \times K_2}$, which is a discretisation of the empirical covariance kernel. If we assume $K_1 = K_2 =: K$, the covariance tensor $\mathbf C$ requires $\mathcal O(N K^4)$ operations to be estimated and $\mathcal{O} (K^4)$ memory to be stored. This becomes barely feasible on a regular computer with $K$ as small as 100. Moreover, as \citep{aston2017} note, the statistical constraints stemming from the need to accurately estimate $\mathcal{O} (K^4)$ parameters contained in $\mathbf C$ from only $N K^2$ measurements are usually even tighter than the computational constraints. 

This dimensionality challenge is often dealt with by imposing additional structure, for example stationarity or separability \citep{genton2007,gneiting2002,gneiting2006}.
Either assumption reduces the four-dimensional nonparametric estimation problem into a two-dimensional one. In the case of a $K \times K$ grid, this reduces the number of parameters from $\O (K^4)$, to $\O (K^2)$. Moreover, both estimation and subsequent manipulation (for example inversion as required in prediction) of the covariance become computationally much simpler, owing to some explicit formulas in case of separability and to the fast Fourier transform in case of stationarity.

Though such assumptions substantially reduce the dimensionality of the problem, the imposed simplicity and structural restrictions are often quite questionable. Stationarity appears overly restrictive when replicated data are available, and indeed is seldom used for functional data. Separability is imposed much more often, despite having shortcomings of its own. A thorough discussion of the implications that separability entails is provided in \cite{rougier2017}. In summary, separable covariances fail to model any space-time interactions whatsoever. Indeed, in recent years, several tests for separability of space-time functional data have been developed and used to demonstrate that for many data sets previously modeled as separable, the separability assumption is distinctly violated \citep{aston2017,bagchi2017,constantinou2017}.

\subsection{Our Contributions}

We propose a more flexible framework than that offered by separability, which allows for mild (non-parametric) deviations from separability while retaining all the computational and statistical advantages that separability offers. In particular, we consider a framework where the target covariance is an additive perturbation of a separable covariance,
\begin{equation}\label{eq:model_kernels_intro}
c(t,s,t',s') = a(t,s,t',s') + b(t,s,t',s')
\end{equation}
where $a(t,s,t',s')=a_1(t,t')a_2(s,s')$ is separable, and $b(t,s,t',s')$ is banded, i.e.~supported on $\{\max(|t-t'|,|s-s'|) \leq \delta\}$ for some $\delta>0$.
Combining the two components results in a non-parametric family of models, which is much richer than the separable class. In particular, the model represents a strict generalisation of separability, reducing to a separable model when $\delta=0$. Intuitively, it postulates that while the global (long-range) characteristics of the process are expected to be separable, there may also be local (short-range) characteristics of the process that may be non-separable. For some practical problems, separability might possibly fail due to some interactions between time and space, which however do not propagate globally. These may be due to (weakly dependent) noise contamination, which can lead to local violations of separability, perturbing the covariance near its diagonal. It could also, however, be due to the presence of signal components that are non-separable and yet weakly dependent.

Heuristically, if we were able to deconvolve the terms $a$ and $b$, then the term $a$ would be easily estimable on the basis of dense observations, exploiting separability. We demonstrate that it actually \emph{is} possible to access a non-parametric estimator of $a$ -- without needing to manipulate or even store the empirical covariance -- by means of a novel device, which we call \emph{shifted partial tracing}. This linear operation mimicks the \emph{partial trace} \cite{aston2017}, but it is suitably modified to allow us to separate the terms $a$ and $b$ in \eqref{eq:model_kernels_intro}. Exploiting this device, we produce a \emph{linear} estimator of $a$ (linear up to scaling, to be precise) that can be computed efficiently, with no computational overhead relative to assuming separability. It is shown to be consistent, with explicit convergence rates, when the processes are observed discretely on a grid, possibly corrupted with measurement error.

The bandwidth $\delta>0$ is assumed constant and non-decreasing in the sample size $N$ or the grid size $K$. Consequently, even though $b$ is banded, it has the same order of entries as $c$ itself, when observed on a grid. Hence if $b$ is also an estimand of interest, and statistical and computational efficiency is sought, an additional structural assumption on $b$ is needed to prevent $b$ from being much more complicated to handle than $a$. We focus on stationarity as a specific assumption, which seems broadly applicable, is interesting from the compuatational perspective, and yields a form of parsimony complementary to separability. Under this additional assumption, we show in detail that both $a$ and $b$ of model \eqref{eq:model_kernels_intro} can be estimated efficiently, and the estimator can be both applied and inverted (numerically), while the computational costs of these operations does not exceed their respective costs in the separable regime. Specifically, we show that all of these operations, i.e.~estimation, application, and inversion of the covariance, can be performed at the same cost as matrix-matrix multiplication between pairs of the sampled observations.

Our methodology is also capable of estimating a separable model under the presence of heteroscedastic noise. When observed on a grid, this leads to a separable covariance superposed with a diagonal structure, which has again the same order of degrees of freedom as the separable part. A heteroscedastic noise may very well arise from a discretization of a random process, which is weakly dependent and potentially even smooth at a finer resolution. In their seminal book \cite{ramsay2005}, Ramsay and Silverman state: \emph{``the functional variation that we choose to ignore is itself probably smooth at a finer scale of resolution.''} In other words, with increasing grid size $K$, a diagonal structure may become a banded structure.  One can thus view our methodology as being able to estimate a separable model observed under heteroscedastic and/or weakly dependent noise. If the degrees of freedom belonging to the noise do not exceed the degrees of freedom of the separable part, we can utilize the noise structure e.g.~for the purposes of prediction with no computational overhead compared to the separable model.

Regardless of whether one views $b$ as an estimand of interest or as a nuisance, the key point of this paper is that the methodology we advocate, and label shifted partial tracing, can be used to estimate the separable part of model \eqref{eq:model_kernels_intro}, provided data are densely observed.

\subsection{Notation}

A real separable Hilbert space $\mathcal H$ is equipped with an inner product $\langle \cdot, \cdot \rangle$ and the induced norm $\| \cdot \|$. The Banach space of operators on $\mathcal H$ is denoted by $\mathcal{S}_\infty(\mathcal H)$.
Using the tensor product notation, we write the SVD of $F \in \mathcal{S}_\infty(\mathcal H)$ as $F = \sum_{j=1}^\infty \sigma_j e_j \otimes f_j$. For $p \geq 1$, $F \in \mathcal{S}_\infty(\mathcal H)$ belongs to $\mathcal{S}_p(\mathcal H)$ if $\verti{F}_p := \big(\sum_{j=1}^\infty \sigma_j^p\big)^{1/p} < \infty$. When equipped with the norm $\verti{\cdot}_p$, $\mathcal{S}_p(\mathcal H)$ is a Banach space. For $F \in \mathcal{S}_1(\mathcal H)$, we define its trace as $\Tr{}(F) = \sum_{j=1} \langle F e_j, e_j \rangle$, where the choice of $\{ e_j \}_{j=1}^\infty$ is immaterial.

The tensor product of two Banach spaces $\mathcal B_1$ and $\mathcal B_2$, denoted by $\mathcal B_1 \otimes \mathcal B_2$, is the completion of the set $\big\{ \sum_{j=1}^N x_j \otimes y_j \,\big|\, x_j \in \mathcal H_1, y_j \in \mathcal H_2, N \in \N \big\}$ (see \cite{weidmann2012}).
We have the isometric isomorphism $\mathcal{S}_p(\mathcal H_1) \otimes \mathcal{S}_p(\mathcal H_2) \simeq \mathcal{S}_p(\mathcal H_1 \otimes \mathcal H_2)$.

For $A_1 \in \mathcal{S}_p(\mathcal H_1)$ and $A_2 \in \mathcal{S}_p(\mathcal H_2)$, we define $A := A_1 \ct A_2$ as the unique operator on $\mathcal{S}_p(\mathcal H) \otimes \mathcal{S}_p(\mathcal H)$ satisfying
$(A_1 \ct A_2)(x \otimes y) = A_1 x \otimes A_2 y$, $\forall x \in \mathcal H_1, y \in \mathcal H_2 $. By the construction above, we have $\verti{A_1 \ct A_2}_p = \verti{A_1}_p \verti{A_2}_p$. 

For a random element $X$ on $\mathcal H$ with $\E \| X \|^2 < \infty$, we denote the mean $m = \E X$ and the covariance $C=\E[ (X - m) \otimes (X-m) ]$ (see \cite{hsing2015}). Covariances are trace-class, and positive semi-definite, i.e.~$C \in \mathcal{S}_1^+(\mathcal H)$. When $H=\mathcal{L}^2[0,1]$, the covariance operator is related to the covariance kernel $c=c(t,s)$ via $ (C f)(t) = \int_0^1 c(t,s) f(s) d s $.

We use capital letters (e.g.~$C \in \mathcal{S}_1(\mathcal L^2[0,1])$) to denote operators, lower-case letters to denote their kernels (e.g.~$c \in \mathcal L^2[0,1]^2$) and bold-face letters to denote discrete objects such as vectors or matrices (e.g.~$\mathbf C \in \R^{K_1 \times K_2}$ when $c$ is measured discretely on a $K_1 \times K_2$ grid in $[0,1]^2$). When $\mathcal{H} = \mathcal{H}_1 \otimes \mathcal{H}_2$, we often think of the first dimension as time, denoted by the variable $t$, and the second dimension as space, denoted by the variable $s$. See Appendix~\ref{app:A} for a more detailed exposition of the notation and background concepts.

\section{Methodology} \label{sec:metho}

\subsection{Separable-plus-Banded Covariance}

\begin{definition}
For $\mathcal H := \mathcal H_1 \otimes \mathcal H_2$, $A \in \mathcal{S}_p(\mathcal H)$ is called separable if $A = A_1 \ct A_2$ for some $A_1 \in \mathcal{S}_p(\mathcal H_1)$ and $A_2 \in \mathcal{S}_p(\mathcal H_2)$. For $\mathcal H = \mathcal{L}^2[0,1]^2$, $B \in \mathcal{S}_2(\mathcal H)$ with kernel $b=b(t,s,t',s')$ is banded by $\delta \in [0,1)$ if $b(t,s,t',s') = 0$ almost everywhere on the set \mbox{$\big\{(t,s,t',s') \in [0,1]^4 \big| \max(|t-t'|,|s-s'|) \geq \delta \big\}$.}
\end{definition}

If $\mathcal H = \mathcal{L}^2[0,1]^2$, operator $A \in \mathcal{S}_2(\mathcal H)$ is separable if and only if its kernel factorizes as
\mbox{$a(t,s,t',s') = a_1(t,t') a_2(s,s')$} almost everywhere for some marginal kernels $a_1$ and $a_2$.

We postulate the following model for the covariance of a random element $X \in \mathcal{L}^2[0,1]^2$:
\begin{equation}\label{eq:model}
C = A_1 \ct A_2 + B ,
\end{equation}
where $A_1, A_2 \in \mathcal{S}_1^+(\mathcal{L}^2[0,1])$ and $B \in \mathcal{S}_1^+(\mathcal{L}^2[0,1]^2)$ is banded by $\delta \in [0,1)$. On the level of kernels, this implies for almost all $t,s,t',s' \in [0,1]$ the decomposition
\begin{equation}\label{eq:model_kernels}
c(t,s,t',s') = a_1(t,t') a_2(s,s') + b(t,s,t',s')\,.
\end{equation}

The covariance structure \eqref{eq:model} can arise for example when $X$ is a superposition of two uncorrelated processes $Y$ and $W$, i.e.~$X(t,s) = Y(t,s) + W(t,s)$, $t,s \in [0,1]$, such that the covariance of $Y$ is separable and the covariance of $W$ is banded (e.g.~$W$ is a moving average process with compactly supported window-width). Note that by choosing $\delta=0$ (leading to $B \equiv 0$), model \eqref{eq:model} contains separability as a sub-model. 

\subsection{Shifted Partial Tracing}

\begin{definition}\label{def:SPT}
Let $C \in \mathcal{S}_1(\mathcal{L}^2[0,1]^2)$ with a continuous kernel $c=c(t,s,t',s')$. Let $\delta \in [0,1)$. We define the $\delta$-shifted trace of $C$ as
\[
\Tr{}^\delta(C) := \int_0^{1-\delta} \int_0^{1-\delta} c(t,s,t+\delta,s+\delta) d t d s \,.
\]
We also define the $\delta-$shifted partial traces of $C$, denoted $\Tr{1}^\delta\left(C\right)$ and $\Tr{2}^\delta\left(C\right)$, as the integral operators with kernels given respectively by
\begin{equation}\label{eq:shifted_partial_trace_integral}
c_1(t,t') := \int_{0}^{1-\delta} c(t,s,t',s+\delta)d s \qquad \& \qquad c_2(s,s') := \int_{0}^{1-\delta} c(t,s,t+\delta,s')d t \,.
\end{equation}
\end{definition}

In the special case of $\delta=0$ the definition of shifted trace corresponds to the standard (non-shifted) trace of a trace-class operator with a continuous kernel. Also, for $\delta=0$, $\delta$-shifted partial tracing corresponds to partial tracing as defined in \cite{aston2017}.

Let us denote by $\mathds{V}_1$, resp. $\mathds{V}_2$, the vector space of trace class operators on $[0,1]$, resp. $[0,1]^2$, with continuous kernels, equipped with the trace norm. We assume continuity here only for the sake of presentation. More general (but less intuitive) development of the methodology is given in Appendix~\ref{app:B}, which also contains the proofs.

\begin{proposition}\label{prop:SPT}
The shifted trace $\Tr{}^\delta: \mathds{V}_1 \to \R$ is a well-defined bounded linear functional, and the shifted partial traces $\Tr{1}^\delta, \Tr{2}^\delta: \mathds{V}_2 \to \mathds{V}_1$ are well-defined bounded linear operators.
\end{proposition}

Shifted partial tracing has the following properties, which will be useful for estimation of model \eqref{eq:model}.

\begin{proposition}\label{prop:PT_properties}
Let $A_1,A_2 \in \mathcal{S}_1(\mathcal{L}^2[0,1])$ and $F = A_1 \ct A_2$. Then \\
$\Tr{1}^\delta\left(F\right) = \Tr{}^\delta\left(A_2\right) A_1$, $\Tr{}^\delta\left(F\right) = \Tr{}^\delta\left(A_1\right) \Tr{}^\delta\left(A_2\right)$, and $\Tr{}^\delta\left(F\right) F = \Tr{1}^\delta\left(F\right) \ct \Tr{2}^\delta\left(F\right).$
\end{proposition}

\subsection{Estimation}\label{sec:estimation}

We assume throughout the paper the availability of $N$ independent (and w.l.o.g.~zero-mean) surfaces, say $X_1, \ldots, X_N$, with covariance given by \eqref{eq:model}, where $\delta$ is such that $\Tr{}^\delta(A_1)$ and $\Tr{}^\delta(A_2)$ are non-zero. For now, let the surfaces be fully observed; discrete observations are considered in Sections \ref{sec:computational} and \ref{sec:asymptotics}. The following lemma illustrates the importance of shifted partial tracing for the estimation task.

\begin{lemma}\label{lem:banded_SPT}
Let $B \in \mathcal{S}_1(\mathcal{L}^2[0,1]^2)$ be banded by $\delta^\star$. Then for any $\delta > \delta^\star$ we have $\Tr{1}^\delta(B) = \Tr{2}^\delta(B)=0$.
\end{lemma}

Therefore, shifted partial tracing works around the banded part of the process to enable a direct estimation of the separable part of the covariance.

\begin{figure}[!t]
   \centering \includegraphics[width=\textwidth]{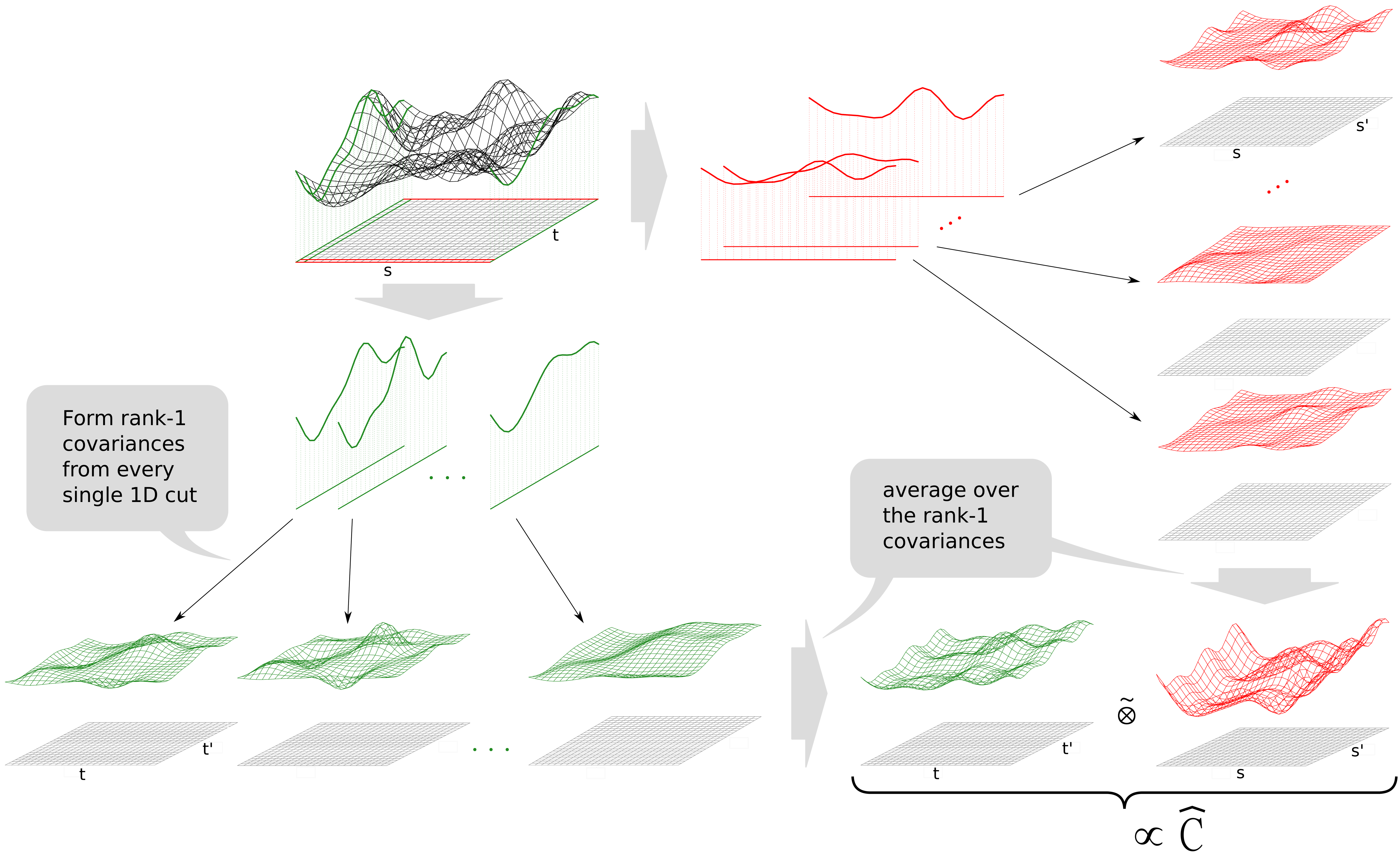}
   \vspace*{-6mm}
   \caption{Visualization of separable model estimation via partial tracing based on a single observation. The observation is cut along the temporal domain to obtain a temporal sample (in green), from which the temporal part of the separable covariance is empirically estimated. Similarly for the spatial part (in red).}
    \label{fig:PT_visualized} 
\end{figure}

\begin{example}\label{ex:PT_visualized}
Assume we have a single continuous observation $X \in \mathcal{L}^2[0,1]^2$ with covariance $C = C_1 \otimes C_2$ with a continuous kernel $c(t,s,t',s') = c_1(t,t') c_2(s,s')$. Assume for simplicity that $\Tr{}(C_1) = \Tr{}(C_2) =1$. Partial tracing (without shifting, i.e.~$\delta=0$) can be used to estimate $C_1$ and $C_2$ in the following way.

The observation $X$ is cut along the temporal axis to form a spatial sample $\{ X^t(s) \}_{t \in [0,1]}$, i.e.~any given time point $t$ provides a single curve $X^t(s)$, $s \in [0,1]$. This spatial sample is used to estimate the spatial covariance $C_2$ in a standard way, i.e.~outer products $X^t \otimes X^t$ are formed and averaged together as
\[
\widehat{C}_2 = \int_0^1 X^t \otimes X^t d t \qquad \text{or equivalently} \qquad \widehat{c}_2(s,s') = \int_0^1 X^t(s) X^t(s') d t \,.
\vspace*{-1mm}
\]
This is a moment estimator in a sense, since $\E(X^t \otimes X^t) = C_2$ for any $t\in [0,1]$.
Similarly for the temporal domain. The process is captured in Figure \ref{fig:PT_visualized}.

When the covariance is instead separable-plus-banded, i.e.~$C = A_1 \otimes A_2 + B$ with $B$ banded by $\delta$, it is no longer true that $\E(X^t \otimes X^t) = A_2$, but it is still true that \mbox{$\E(X^t \otimes X^{t+\delta}) \propto A_2$} for all $t \in [0,1-\delta]$. Hence instead of taking outer products of $X^t$ with itself, we can form outer products $X^t \otimes X^{t+\delta}$ and average over these for $t \in [0,1-\delta]$ to obtain a scaled estimator of $A_2$.
\end{example}

Using the previous lemma together with Proposition \ref{prop:PT_properties}, we obtain the following estimating equation for model \eqref{eq:model}:
\begin{equation}\label{eq:estimating_eq}
\Tr{}^\delta(C) A_1 \ct A_2 = \Tr{1}^\delta(C) \ct \Tr{2}^\delta(C).
\end{equation}

Equation \eqref{eq:estimating_eq} suggests the following estimators for the separable part of the model:
\begin{equation}\label{eq:estimator_A}
\widehat{A}_1 = \Tr{1}^\delta(\widehat{C}_N) \qquad \& \qquad \widehat{A}_2 = \frac{\Tr{2}^\delta(\widehat{C}_N)}{\Tr{}^\delta(\widehat{C}_N)},
\end{equation}
where $\widehat{C}_N = \frac{1}{N} \sum_{n=1}^N (X_n - \bar{X}_N) \otimes (X_n - \bar{X}_N)$ is the empirical estimator of $C$. Of course, we need to assume $\Tr{}^\delta(\widehat{C}_N) \neq 0$. Once the separable part of the model has been estimated, we can define
\begin{equation}\label{eq:estimator_B_non_stationary}
\widehat{B} = \widehat{C}_N - \widehat{A}_1 \ct \widehat{A}_2 \,.
\end{equation}
Optionally, we can set the kernel of $\widehat{B}$ to zero outsize of the band of size $\delta$. Note that none of the estimators defined above is guaranteed to be symmetric or positive semi-definite. However, this is just a technicality, which can be dealt with easily, see Appendix~\ref{app:E}.

If at this point we add the stationarity of $B$ into our assumptions, i.e.~let the kernel $b$ be translation invariant: $b(t,s,t',s') = \varsigma(|t-t'|,|s-s'|)$, $t,t',s,s' \in [0,1]$, where $\varsigma \in \mathcal{L}^2[0,1]^2$ is the \emph{symbol} of $B$. Then we take the following estimator of $B$ instead:
\begin{equation}\label{eq:estimator_B}
\widehat{B} = \topavg(\widehat{C}_N - \widehat{A}_1 \ct \widehat{A}_2) ,
\end{equation}
where $\topavg(\cdot)$ is the ``Toeplitz averaging'' operator, i.e.~the projection onto the stationary operators, defined as follows.

\begin{definition}\label{def:toeplitz_averaging}
For $F \in S_1(\mathcal{L}^2[0,1]^2)$ self-adjoint and $\{ e_j \}_{j=- \infty}^\infty$ the complete orthonormal basis of trigonometric functions in $\mathcal{L}^2[0,1]$, let
\begin{equation}\label{eq:trigonometric_F}
F = \sum_{i,j,k,l \in \mathds{Z}} \gamma_{ijkl} (e_i \otimes e_j) \otimes (e_k \otimes e_l).
\end{equation}
Then we define
\begin{equation}\label{eq:trigonometric_F_projected}
\topavg(F) = \sum_{i,j \in \mathds{Z}} \gamma_{ijij} (e_i \otimes e_j) \otimes (e_i \otimes e_j) .
\end{equation}
\end{definition}

Let us comment on the previous definition. If $\{ e_j \}_{j=-\infty}^\infty$ is the trigonometric basis on $\mathcal{L}^2[0,1]$, then $\{ e_i \otimes e_j \}_{i,j=- \infty}^\infty$ is the trigonometric basis on $\mathcal{L}^2[0,1]^2$, so every compact operator $F$ can be expressed with respect to this basis as in \eqref{eq:trigonometric_F}. For $F$ trace class, the Fourier coefficients $\{ \gamma_{ijkl} \}$ are absolutely summable, leading to $\topavg(F)$ in \eqref{eq:trigonometric_F_projected} being also trace-class. Secondly, a stationary operator has the trigonometric basis as its eigenbasis (see Appendix~\ref{app:C}). Thirdly, $\topavg(\cdot)$ as defined in \eqref{eq:trigonometric_F_projected} is clearly an orthogonal projection. Altogether, $\topavg(\cdot)$ is the orthogonal projection onto the space of stationary operators in $S_1(\mathcal{L}^2[0,1]^2)$, which is itself a Banach space.

The next theorem shows asymptotic behavior of the proposed estimators under the following moment assumption, which ensures that $\sqrt{N}(\widehat{C}_N - C) \stackrel{d}{\longrightarrow} Z$, where $Z$ is a mean-zero Gaussian random element in $\mathcal{S}_1(\mathcal{L}^2[0,1]^2)$, i.e. the convergence is in the trace-norm topology \cite{mas2006}.
\begin{description}
\item[(A1)] Let $\sum_{j=1}^\infty \big( \E \langle X,e_j \rangle^4 \big)^{1/4} < \infty$ for some orthonormal basis $\{ e_j \}_{j=1}^\infty$ in $\mathcal{L}^2[0,1]^2$.
\end{description}

\begin{theorem}\label{thm:asymptotics}
Let $X_1,\ldots,X_N \sim X$ be a (w.l.o.g.~centered) random sample with covariance given by \eqref{eq:model}, where $B$ is stationary and $\delta^\star$-banded. Let $\delta \geq \delta^\star$ such that $\Tr{}^\delta(C) \neq 0$. Let (A1) hold.
Then$\sqrt{N}( \widehat{A}_1 \ct \widehat{A}_2 - A_1 \ct A_2)$ and $\sqrt{N}( \widehat{B} - B )$
converge to mean zero Gaussian random elements of $\mathcal{S}_1(\mathcal{L}^2[0,1]^2)$.
\end{theorem}

It can be seen from the proof of the theorem (in Appendix \ref{app:G1}) that the the asymptotic distribution of $\widehat{A}_1$ and $\widehat{A}_2$ remains valid even without the stationarity assumption placed on $\widehat{B}$.

\subsection{Choice of Bandwidth}\label{sec:delta}

In order to apply our methodology in practice, it remains to provide means to choose the bandwidth $\delta$. In this Section, we write $\widehat C (\delta) = \widehat A_1(\delta) \ct \widehat A_2(\delta) + \widehat B(\delta)$ to make the dependency on $\delta$ explicit. Even though the separable part of the model $A_1 \ct A_2$ does not depend on $\delta$, its estimator from formula \eqref{eq:estimator_A} does, since shifted partial tracing with a fixed $\delta$ is used. If we actually knew true covariance $C$, we would use it in formulas \eqref{eq:estimator_A} and \eqref{eq:estimator_B} instead of the empirical covariance $\widehat{C}_N$ to obtain a separable-plus-banded proxy of $C$, denoted here as $C(\delta) = A_1(\delta) \ct A_2(\delta) + B(\delta)$. Under the separable-plus-banded model, it is $C(\delta)=C$ for any $\delta$ large enough to eliminate $B$ by $\delta$-shifted partial tracing. However, among all such bandwidths, the smaller ones will lead to better empirical performance.

Let $\Delta := \{ \delta_1, \ldots, \delta_m \}$ be the search grid of candidate values. If we knew $C$, the bandwidth value leading to the best performance of our estimation methodology would be given by
\begin{equation}\label{eq:CV_theoretical}
\delta^\star := \argmin_{\delta \in \Delta} \verti{C(\delta) - C}_2^2.
\end{equation}
Here, $\delta^\star$ is a set. In particular, under model \eqref{eq:model}, $\delta^\star$ contain all such bandwidths $\delta$ that $B$ is banded by $\delta$. We identify $\delta^\star$ with the minimum of this set. This arbitrary choice reflects the fact that $\delta$ is a nuisance parameter, not an estimand of interest. And as suggested by Theorem \ref{thm:asymptotics}, there is a range of valid values, which are asymptotically indistinguishable.

Since we do not know $C$ we cannot evaluate the objective in \eqref{eq:CV_theoretical}. Instead, we propose to approximate the objective by one that is fully calculable:
\begin{equation}\label{eq:CV_empirical}
\widehat{\delta} := \argmin_{\delta \in \Delta} \verti{\widehat{C}(\delta)}_2^2 + \frac{2}{N} \sum_{n=1}^N \langle X_n , \widehat{C}_{-n}(\delta) X_n \rangle,
\end{equation}
where $\widehat{C}_{-n}(\delta)$ is our estimator constructed without the $n$-th observation $X_n$. In Appendix~\ref{app:G2}, we show that \eqref{eq:CV_empirical} is root-$n$ consistent for \eqref{eq:CV_theoretical} up to a constant, and the following theorem provides rates of convergence with the adaptive choice of the bandwidth.

\begin{theorem}\label{thm:rates_adaptive}
Let $X_1,\ldots,X_N \sim X$ be a (w.l.o.g.~centered) random sample with covariance given by \eqref{eq:model}, where $B$ is stationary and $\delta^\star$-banded. Let (A1) hold, and, let $\widehat{\delta}$ be chosen as in \eqref{eq:CV_empirical} from $\Delta$ in which there exists $\delta \geq \delta^\star$ such that $\Tr{}^\delta(C) \neq 0$.
Then $\verti{ \widehat{A}_1(\widehat{\delta}) \ct \widehat{A}_2(\widehat{\delta}) - A_1 \ct A_2 }_2^2 = \O_P(N^{-1})$ and $\verti{ \widehat{B}(\widehat{\delta}) - B }_2^2 = \O_P(N^{-1})$.
\end{theorem}

In fact, we show in Appendix~\ref{app:G} that even if model \eqref{eq:model} is not valid, the chosen bandwidth leads to a separable-plus-banded proxy of $C$, denoted $\widehat{C}(\widehat{\delta})$, which is asymptotically optimal in the sense of \eqref{eq:CV_theoretical}. Recall that in this case $\widehat{C}(\widehat{\delta})$ is not equal to $C$ even asymptotically, since it is a biased proxy obtained via the proposed estimation methodology based on shifted partial tracing. We also give in Appendix~\ref{app:G3} a version of Theorem~\ref{thm:asymptotics} with an adaptively chosen bandwidth, providing the limiting law further to the rates of convergence above, which requires a slight (and from the practical point of view unnecessary) modification of the bandwidth selection scheme.

\section{Computational Considerations}\label{sec:computational}

Consider now the practical scenario in which we only have access to a discrete version of $X$ in the form of a random element $\mathbf X \in \R^{K_1 \times K_2}$. The discrete version of the covariance $C$ is the covariance tensor $\mathbf C \in \R^{K_1 \times K_2 \times K_1 \times K_2}$, where $\mathbf C[i,j,k,l] = \cov(\mathbf X[i,j], \mathbf X[k,l])$.  Assuming separability of $\mathbf C$ translates to $\mathbf C = \mathbf C_1 \ct \mathbf C_2$ for some $\mathbf C_1 \in \R^{K_1 \times K_1}$ and $\mathbf C_2 \in \R^{K_2 \times K_2}$, or entry-wise $\mathbf C[i,j,k,l] = \mathbf C_1[i,k] \mathbf C[j,l]$ for $i,k=1,\ldots,K_1$ and $j,l=1,\ldots,K_2$.

Assume that $N$ independent realizations of $\mathbf X \in \R^{K \times K}$ were sampled (let $K_1 = K_2 =: K$ again for simplicity) and denoted as $\mathbf X_1,\ldots,\mathbf X_N$. Firstly, a general covariance tensor $\mathbf C$ has $\O(K^4)$ degrees of freedom, while it only has $\O(K^2)$ degrees of freedom under the separability assumption. In comparison, the observed degrees of freedom are only $N K^2$. Secondly, it takes $\O(N K^4)$ operations to calculate the empirical estimate of the covariance tensor, i.e.~$\widehat{\mathbf C}_N = \frac{1}{N} \sum_{n=1}^N \mathbf X \otimes \mathbf X$, while this will be shown to reduce to $\O(N K^3)$ under separability. We assume throughout the paper that multiplication of two $K \times K$ matrices requires $\O(K^3)$ operations, and we set the cubic order in $K$ as the limit of computational tractability for ourselves, which for example prevents us from ever explicitly calculating the empirical covariance $\widehat{\mathbf{C}}_N$. Also, the degrees of freedom correspond to storage requirements, thus although a general covariance tensor becomes difficult to manipulate on a standard computer for $K$ as low as $100$ (at that point the empirical covariance takes roughly 6 GB of memory), the situation under separability is much more favorable.

The following two properties hold for matrices $\mathbf A$, $\mathbf B$ and $\mathbf X$ of appropriate sizes:
\begin{equation}\label{eq:fast_formulas}
\begin{split}
(\mathbf A \ct \mathbf B) \mathbf{X} &= \mathbf A \mathbf X \mathbf B ,\\
(\mathbf A \ct \mathbf B)^{-1}& = \mathbf A^{-1} \ct \mathbf B^{-1} .
\end{split}
\end{equation}
These two properties are among the core reasons for popularity of the separability assumption in the space-time processes literature \cite{gneiting2006}, because they allow to apply a separable covariance fast ($\O(K^3)$ instead of $\O(K^4)$ operations) and solve an inverse problem involving the covariance fast ($\O(K^3)$ instead of $\O(K^6)$ operations).

\begin{remark}\label{rem:terminology}
The symbol $\otimes$ is commonly overused in the literature. In this paper, we use it as the symbol for the abstract outer product \cite{weidmann2012}. The symbol $\ct$ also denotes a type of abstract outer product, but we emphasize by the tilde that we do not see e.g.~$A \ct B$ as an element of a product Hilbert space $\mathcal{S}_p(\mathcal H_1) \otimes \mathcal{S}_p(\mathcal H_2)$, but rather as an operator acting on a product Hilbert space $\mathcal H = \mathcal H_1 \otimes \mathcal H_2$. The symbol $\otimes$ is used in linear algebra for the Kronecker product, which we denote $\otimes_K$ here. The following relation between the Kronecker product and the abstract outer product holds in the case of finite dimensional spaces: 
\begin{equation}\label{eq:kron}
\vec{(\mathbf A \ct \mathbf B) \mathbf X} = (\mathbf B^\top \otimes_K \mathbf A) \mathbf x,
\end{equation}
where $\mathbf x = \vec{\mathbf X}$ is the vectorization of matrix $\mathbf{X}$, and $\vec{\cdot}$ is the vectorization operator (c.f.~\cite{vanloan1983}). Properties \eqref{eq:fast_formulas} are well known in computational linear algebra, where the Kronecker product is used instead of the abstract outer product. Due to \eqref{eq:kron}, the first formula in \eqref{eq:fast_formulas} can be translated to
$(\mathbf A \otimes_K \mathbf B) \vec{\mathbf{X}}=\vec{\mathbf B^\top \mathbf X \mathbf A}$.
\end{remark}

In summary, separability leads to an increased estimation accuracy, lower storage requirements, and faster computations. We view our separable-plus-banded model as a generalization of separability, and the aim of this section is to show that this generalization \textit{does not} come at the cost of loosing the favorable properties of the separable model described above. In fact, we show in the remainder of this section that model \eqref{eq:model} can be estimated and manipulated under the same computational costs as the separable model.

\subsection{Estimation Complexity}

When working with discrete samples, the shifted partial tracing is defined as before, only with the Lebesque measure replaced by the counting measure (see also Appendix~\ref{app:B}). This means, e.g., that for $\mathbf M \in \R^{K_1 \times K_2 \times K_1 \times K_2}$ and $d \leq \min(K_1, K_2)$, we have 
\begin{equation}\label{eq:discrete_PT_def}
\Tr{1}^d(\mathbf M)[i,k] = \sum_{j=1}^{K_2 - d} \mathbf{M}[i,j,k,j+d] \,, \qquad\qquad i,k=1,\ldots,K_1 \,.
\end{equation}
The situation is more complicated with Toeplitz averaging since, unlike in the continuous case, the discrete Fourier basis is not necessarily the eigenbasis of a stationary operator. However, one can define $\topavg(\mathbf M) \in \R^{K_1 \times K_2 \times K_1 \times K_2}$ directly as the tensor having $\mathbf{S} [h,l] = \frac{1}{K^2} \sum_{i=1}^{K-h} \sum_{j=1}^{K-l} \mathbf M[i,j,i+h-1,j+l-1]$ as its symbol. This justifies the name ``Toeplitz averaging''. The relation to the discrete Fourier basis is discussed in Appendix~\ref{app:C}.

Now, consider the separable-plus-banded model $\mathbf{C} = \mathbf{A}_1 \ct \mathbf{A}_2 + \mathbf B$ with $\mathbf B$ banded by $d$, i.e.~$\mathbf B[i,j,k,l]=0$ whenever $\min(|i-k|,|j-l|) \geq d$. We use $d$ to denote the discrete version of the bandwidth $\delta$; the relation for an equidistant grid of size $K \times K$ is $d = \lceil \delta K \rceil + 1$. It is straightforward to translate Proposition \ref{prop:PT_properties} and Lemma \ref{lem:banded_SPT} to the discrete case to obtain the estimating equation
$\Tr{}^d(\mathbf{C}) \mathbf{A}_1 \ct \mathbf{A}_2 = \Tr{1}^d(\mathbf{C}) \ct \Tr{2}^d(\mathbf{C})$,
suggesting again the plugin estimators
\begin{equation}\label{eq:discrete_separable_estimators}
\widehat{\mathbf{A}}_1 = \Tr{1}^d(\widehat{\mathbf{C}}_N) \quad \text{and} \quad \widehat{\mathbf{A}}_2 = \Tr{2}^d(\widehat{\mathbf{C}}_N)/\Tr{}^d(\widehat{\mathbf{C}}_N) \,.
\end{equation}
It may be useful to revisit Example \ref{ex:PT_visualized} and Figure \ref{fig:PT_visualized} (which is plotted discretely anyway) for intuitive depiction of discrete shifted partial tracing.

Now we are ready to establish the estimation complexity. Firstly, we focus on shifted partial tracing. Due to linearity, $\Tr{1}^d(\widehat{\mathbf C}_N) = \frac{1}{N} \sum_n  \Tr{1}^d(\mathbf X_n \otimes \mathbf X_n)$, and as can be seen from formula \eqref{eq:discrete_PT_def}, only $K^3$ entries of the total of $K^4$ entries of $\mathbf X_n \otimes \mathbf X_n$ are needed to evaluate the shifted partial trace. Moreover, evaluating the shifted partial trace amounts to averaging over one dimension of the relevant 3D object, which does not have to ever be stored, hence the time and memory complexities to estimate the separable part of the model, i.e.~to evaluate \eqref{eq:discrete_separable_estimators}, are $\O(N K^3)$ and $\O(K^2)$, respectively.

To evaluate $\widehat{\mathbf B} = \topavg(\widehat{\mathbf{C}}_N - \widehat{\mathbf{A}}_1 \ct \widehat{\mathbf{A}}_2) = \frac{1}{N} \sum_n \topavg(\mathbf X_n \otimes \mathbf X_n) - \topavg(\widehat{\mathbf A}_1 \ct \widehat{\mathbf A}_2)$, one can utilize the fast Fourier transform (FFT). Every term $\topavg(\mathbf X_n \otimes \mathbf X_n)$ can be evaluated directly on the level of data, without the necessity to form the empirical estimator, in $\O (K^2 \log(K))$. The term $\topavg(\widehat{\mathbf A}_1 \ct \widehat{\mathbf A}_2)$ can be evaluated in $\O(K^3)$ operations, again without explicitly forming the outer product (see Appendix~\ref{app:C}). Hence estimation of the banded part is equally demanding as the estimation of the separable part.

It remains to show that $\widehat{\mathbf C} := \widehat{\mathbf A}_1 \ct \widehat{\mathbf A}_2 + \widehat{\mathbf B}$ can be applied efficiently, that the bandwidth selection strategy is feasible, and that an inverse problem $\widehat{\mathbf C} \mathbf X = \mathbf{Y}$ can be solved efficiently. The application of $\widehat{\mathbf C}$ is simple due to the additive structure: one applies the separable part using the first formula in \eqref{eq:fast_formulas}, the banded part using the FFT, and sums the two, leading to the desired complexities. For bandwidth selection, one needs to evaluate the objective of \eqref{eq:CV_empirical} for all $\delta \in \Delta$. The norm in \eqref{eq:CV_empirical} can be calculated fast using the separable-plus-banded form of the estimator; see Appendix~\ref{app:D}. The inner products take $\O(K^2)$ operations each after a fast application $\widehat{\mathbf C}_{-n}(\delta)$. However, it is wasteful to re-estimate $\widehat{\mathbf C}_{-n}(\delta)$, holding out one observation at a time. Instead, one can split the data into e.g.~10 folds, and hold out each fold as a whole. Evaluation for a single fold then takes $\O(N K^3)$, since estimation is the dominating operation. Hence, overall, bandwidth selection is computationally tractable. Finally, the inverse problem is non-trivial, since it is not possible to express the inverse of a sum of two operators in terms of inverses of the two summands. This problem is dealt with in the following section.

\newpage
\subsection{Inverse Problem}\label{sec:inverse_alg}

We need a fast solver for the linear system coming from a discretization of model \eqref{eq:model}, i.e.
\begin{equation}\label{eq:inverse_problem_tensor}
(\mathbf A_1 \ct \mathbf A_2 + \mathbf B) \mathbf X = \mathbf Y ,
\end{equation}
where $\mathbf B \in \R^{K \times K \times K \times K}$ is stationary. Equation \eqref{eq:inverse_problem_tensor} can be rewritten in the matrix-vector form as 
\begin{equation}\label{eq:inverse_problem_matrix}
(\mathbf A + \mathbf B) \mathbf x = \mathbf y  ,
\end{equation}
where $\mathbf A = \mathbf A_2 \otimes_K \mathbf A_1$ (see Remark \ref{rem:terminology}), $\mathbf x = \vec{\mathbf X}$, $\mathbf y = \vec{\mathbf Y}$, and $\mathbf B \in \R^{K^2 \times K^2}$ is a two-level Toeplitz matrix (i.e.~a Toeplitz block matrix with Toeplitz blocks).

The naive solution to system \eqref{eq:inverse_problem_matrix} would require $\O(K^6)$ operations. Since the estimation of model \eqref{eq:model} takes $\O(N K^3)$, we are looking for a solver for \eqref{eq:inverse_problem_matrix} with a complexity close to $\O(K^3)$. We will develop an Alternating Direction Implicit (ADI, c.f.~\cite{young2014}) solver with the per-iteration cost of $\O(K^3)$ and rapid convergence.

The system \eqref{eq:inverse_problem_matrix} can be transformed into either of the following two systems:
\begin{equation}\label{eq:ADI_preparation}
\begin{split}
(\mathbf A + \rho \mathbf I) \mathbf x &= \mathbf y - \mathbf B \mathbf x + \rho \mathbf x ,\\
(\mathbf B + \rho \mathbf I) \mathbf x &= \mathbf y - \mathbf A \mathbf x + \rho \mathbf x ,
\end{split}
\end{equation}
where $\mathbf I \in \R^{K^2 \times K^2}$ is the identity matrix and $\rho \geq 0$ is arbitrary. The idea of the ADI method is to start from an initial solution $\mathbf x^{(0)}$, and form a sequence $\{\mathbf x^{(k)} \}_{k; 2k \in \N}$ by alternately solving the linearized systems stemming from \eqref{eq:ADI_preparation} until convergence, specifically:
\begin{equation}\label{eq:ADI}
\begin{split}
(\mathbf A + \rho \mathbf I) \mathbf x^{(k+1/2)} &= \mathbf y - \mathbf B \mathbf x^{(k)} + \rho \mathbf x^{(k)} ,\\
(\mathbf B + \rho \mathbf I) \mathbf x^{(k+1)} &= \mathbf y - \mathbf A \mathbf x^{(k+1/2)} + \rho \mathbf x^{(k+1/2)} .
\end{split}
\end{equation}
The acceleration parameter $\rho$ is allowed to vary between iterations. The optimal choice of $\rho$ based on the spectral properties of $\mathbf A$ and $\mathbf B$, guaranteeing a fixed number of iterations, can be made in some model examples (e.g.~when $\mathbf A$ and $\mathbf B$ commute). Interestingly, numerical studies suggest that the ADI method exhibits excellent performance on a large class of linear systems of the type \eqref{eq:inverse_problem_matrix} with the model choice of $\rho$, as long as matrices $\mathbf A$~and~$\mathbf B$ are real with real spectra \citep{young2014}. Hence we also choose $\rho$ as suggested by the model examples and, in order to boost the convergence speed, we gradually decrease its value as
\mbox{$\rho^{(k)} = \min(\rho^{(k-1)}, \frac{\|\mathbf x^{k+1} - \mathbf x^{k} \|_2}{\| \mathbf x^{k} \|_2})$, $k \in \N$, with $\rho^{(0)} = \sqrt{\max(\alpha_{\text{max}} \alpha_{\text{min}}, \beta_{\text{max}} \beta_{\text{min}})}+\epsilon$,}
where $\alpha_{\text{max}}$~and $\alpha_{\text{min}}$ (resp. $\beta_{\text{max}}$ and $\beta_{\text{min}}$) are maximum and minimum eigenvalues of $\mathbf{A}$~(resp.~$\mathbf{B}$), and $\epsilon$ is a small positive constant (by default the desired precision). Recall that $\mathbf A$ and $\mathbf B$ are positive semi-definite.

Now it remains to show how to efficiently solve the linear sub-problems \eqref{eq:ADI}. The first sub-problem has an analytic solution given in the matrix form by
\[
\mathbf X^{(k+1/2)} =  \mathbf V \big[\mathbf G \odot \t{\mathbf U} (\mathbf Y - \mathbf B \mathbf X^{(k)} + \rho \mathbf X^{(k)} ) \mathbf V\big] \t{\mathbf U}\,,
\]
where $\mathbf G$ depends on the eigenvalues of $\mathbf A_1$ and $\mathbf A_2$, and $\odot$ denotes the Hadamard (element-wise) product. This solution is computable in $\mathcal{O}(K^3)$ operations. The second sub-problem in \eqref{eq:ADI} involves a two-level Toeplitz matrix as its left-hand side and can be solved iteratively via preconditioned conjugate gradient, with a single-step complexity of $\O(K^2 \log(K))$. See Appendix~\ref{app:F} for details.

In summary, we devised a doubly iterative algorithm to solve inverse problems in the context of the separable-plus-stationary model. The outer iterative scheme requires solution of two linear systems, one solvable in $\O(K^3)$ iterations, the other in $\O(\eta_{pcg}K^2 \log(K))$, where $\eta_{pcg}$ is the number of the iterations of the inner scheme. In Section \ref{sec:simulation}, we demonstrate empirically that $\eta_{pcg}$ does not increase with increasing $K$, and hence the overall complexity of the algorithm is $\O(\eta_{adi}K^3)$, where $\eta_{adi}$ is the number of outer iterations. As demonstrated again in Section \ref{sec:simulation}, $\eta_{adi}$ also does not depend on $K$, leading to an overall complexity $\O(K^3)$. Hence we have a tractable inversion algorithm for the separable-plus-stationary model.

Note that stationarity of $\mathbf B$ is used at two instances: in the top right-hand side of \eqref{eq:ADI}, $\mathbf B$ needs to be applied fast, and $(\mathbf B + \rho \mathbf I) \mathbf x = \mathbf y$ needs to be solved fast. Both of these are easy if for example $\mathbf B$ is diagonal. Hence we also have inversion algorithm when a separable covariance is observed under heteroscedastic noise.

\section{Asymptotics under Discrete Noisy Measurements}\label{sec:asymptotics}

While Theorem \ref{thm:asymptotics} establishes the asymptotic behavior of our estimators under complete observations, in practice one only observes discrete and potentially noisy samples, which are the topic of this section.
Let $[0,1]^2 = \bigcup_{i=1}^{K} \bigcup_{j=1}^{K} I_{i,j}^K$, where $I_{i,j}^K$ is a Cartesian product of two sub-intervals of $[0,1]$, $I_{i,j}^K \cap I_{i',j'}^K = \emptyset$ for $(i,j)\neq(i',j')$, and $|I_{i,j}^K|=K^{-2}$ for all $i,j=1,\ldots,K$. 
The observations are assumed to be of the form 
\begin{equation}\label{eq:discrete_noisy_observations}\widetilde{\mathbf X}_n^K[i,j] = \mathbf X_n^K[i,j] + \mathbf{E}_n^K[i,j] \,, \qquad i=1,\ldots,K \,, \; j=1,\ldots,K \,,
\end{equation}
where the matrices $\mathbf X_1,\ldots,\mathbf X_N \in \R^{K\times K}$ are discretely measured versions of the fully observed data $X_1,\ldots,X_N \in \mathcal{L}^2[0,1]^2$, and $\mathbf{E}_n^K$ are measurement errors. 
We will consider two types of sampling schemes, which relate the fully observed data $X_1,\ldots,X_N \in \mathcal{L}^2[0,1]^2$ to their discrete versions $\mathbf X_1,\ldots,\mathbf X_N \in \R^{K\times K}$:
\begin{description}
\item[(S1)] $X_n$, $n=1,\ldots,N$, are observed pointwise on a grid, i.e.~there exist $t_1^K,\ldots,t_{K}^K \in [0,1]$ and $s_1^K,\ldots,s_{K_2}^K \in [0,1]$ such that $(t_i^K,s_j^K) \in I_{i,j}^K$
\[
\mathbf X_n^K[i,j] = X_n(t_i^K,s_j^K) \,, \qquad i=1,\ldots,K \,, \; j=1,\ldots,K \,.
\]
Note that to make such point evaluations of $X$ meaningful, we have to assume that realizations of $X$ are continuous (cf. \cite{hsing2015}).
\item[(S2)] The average value of $X_n$ on the pixel $I_{i,j}^K$ is observed for every pixel, i.e.
\[
\mathbf X_n^K[i,j] = \frac{1}{|I_{i,j}^K|} \int_{I_{i,j}^K} X_n(t,s) d t d s \,, \qquad i=1,\ldots,K \,, \; j=1,\ldots,K \,.
\]
\end{description}

As for the measurement error arrays $\big( \mathbf E_n^K[i,j] \big)_{i,j=1}^K$, these are assumed to be  i.i.d. (with respect to $n$) and uncorrelated with $\mathbf X_n$, satisfying the following 4-th order moment conditions for $i,j,k,l,i',j',k',l'=1,\ldots,K$ and $n=1,\ldots,N$:
\[
\begin{split}
\E\big(\mathbf{E}_n^K[i,j]\big) &= 0 \,,  \\
\exists\, \sigma^2 < \infty : \E\big(\mathbf{E}_n^K[i,j] \mathbf{E}_n^K[k,l]\big) &\leq \sigma^2 \mathds{1}_{[i=k,j=l]} \,, \\
\E\big( \mathbf{E}_n^K[i,j] \mathbf{E}_n^K[k,l] \mathbf{X}_n^K[i',j'] \mathbf{X}_n^K[k',l'] \big) &= \E\big( \mathbf{E}_n^K[i,j] \mathbf{E}_n^K[k,l] \big) \, \E \big( \mathbf{X}_n^K[i',j'] \mathbf{X}_n^K[k',l'] \big)  \,.
\end{split}
\]
The previous inequality allows for heteroscedasticity of the noise process: we allow the variance to change with location, but assume it is bounded over the domain by an unknown constant $\sigma^2$. Note that under the sampling scheme (S1) and homoscedasticity (constraining the inequality to equality), equation \eqref{eq:discrete_noisy_observations} corresponds to the commonly adopted errors-in-measurements model \cite{yao2005,zhang2016}.

Let us denote $X^K(t,s) = \sum_{i=1}^{K} \sum_{j=1}^{K} \mathbf{X}^K[i,j] \mathds{1}_{[(t,s) \in I_{i,j}^K]}$, i.e.~$X^K$ is the piecewise constant continuation of $\mathbf X^K$. One can readily verify that pointwise sampling (scheme S1) corresponds to pointwise evaluations of the covariance, i.e.~$\var(X^K) = C^K$, where $C^K$ has kernel
\[
c^K(t,s,t',s') =  \sum_{i,j,k,l=1}^K c(t_i,s_j,t_k,s_l) \mathds{1}_{[(t,s) \in I_{i,j}^K]} \mathds{1}_{[(t',s') \in I_{k,l}^K]} \,,
\]
while pixel-wise sampling (scheme S2) corresponds in turn to pixelization of the covariance. Namely, if we denote $g_{i,j}^K(t,s) = K \mathds{1}_{[(t,s) \in I_{i,j}^K]}$ then we have $\var(X^K) = C^K$ with
\begin{equation}\label{eq:g_functions}
X^K = \sum_{i=1}^K \sum_{j=1}^K \langle X, g_{i,j}^K \rangle g_{i,j}^K
\,, \qquad 
C^K = \sum_{i,j,k,l=1}^K \langle C , g_{i,j}^K \otimes g_{k,l}^K \rangle g_{i,j}^K \otimes g_{k,l}^K 
\end{equation}
In the same spirit, $C^K$ is the piecewise constant continuation of $\mathbf C^K = \E ( \mathbf X^K \otimes \mathbf X^K)$.

If we constrain ourselves to the noiseless multivariate setting and consider the discrete version of the covariance to be the ground truth, it is straightforward to obtain the multivariate version of Theorem \ref{thm:asymptotics}, regardless of the sampling scheme. When both $N$ and $K$ diverge, Theorem \ref{thm:asymptotics} does not apply, but we can still obtain convergence rates. To this aim, we first ought to clarify how bandedness of $B$, $B^K$ and $\mathbf B^K$ are related. It can be seen that if $B$ is banded by $\delta$, then $\mathbf B^K$ is banded by \mbox{$d_K = \lceil \delta K \rceil + 1$}, while $B^K$ is banded by $\delta_K = d_K/K$, which decreases monotonically down to $\delta$ for $K \to \infty$. In the following theorem, $\widehat{A}_1^K$ and $\widehat{A}_2^K$ denote piecewise constant continuations of \mbox{$\widehat{\mathbf A}_1^K = \Tr{1}^{d_K}(\widehat{\mathbf{C}}_N^K)$} and $\widehat{\mathbf A}_2^K = \Tr{2}^{d_K}(\widehat{\mathbf{C}}_N^K)/ \Tr{}^{d_K}(\widehat{\mathbf{C}}_N^K)$, where $\widehat{\mathbf{C}}_N^K = \frac{1}{N} \sum_{n=1}^N \widetilde{\mathbf X}_n^K \otimes \widetilde{\mathbf X}_n^K$ is the empirical covariance based on the observed (noisy) data \eqref{eq:discrete_noisy_observations}. 

\begin{theorem}\label{thm:rates}
Let $X_1, \ldots, X_N$ be i.i.d. copies of $X \in \mathcal{L}^2[0,1]^2$, which has (w.l.o.g.~mean zero and) covariance given by \eqref{eq:model}, where the the separable part $A := A_1 \ct A_2$ has kernel $a(t,s,t',s')$, which is Lipschitz continuous on $[0,1]^4$ with Lipshitz constant $L>0$. Let $\E \| X \|^4 < \infty$ and $\delta \in [0,1)$ be such that $B$ from \eqref{eq:model} is banded by $\delta$ and $\Tr{}^\delta(A) \neq 0$. Let the samples come from \eqref{eq:discrete_noisy_observations} via measurement scheme (S1) or (S2) with $\var(\mathbf{E}_n^K[i,j]) \leq \sigma^2 = \mathcal{O}(\sqrt{K})$. Then we have
\begin{equation}\label{eq:rates}
\verti{\widehat{A}_1^K \ct \widehat{A}_2^K - A_1 \ct A_2}_2^2  = \mathcal{O}_P(N^{-1}) + 2 K^{-2} L^2,
\end{equation}
where the $\mathcal{O}_P(N^{-1})$ term is uniform in $K$, for all $K \geq K_0$ for a certain $K_0 \in \N$. Furthermore, if $\widehat{A}_1^K = \sum_{j \in \N} \widehat{\lambda}^K_j \widehat{e}_j^K \otimes \widehat{e}_j^K$, $\widehat{A}_2^K = \sum_{j \in \N} \widehat{\rho}^K_j \widehat{f}_j^K \otimes \widehat{f}_j^K$, $A_2 = \sum_{j \in \N} \lambda_j e_j \otimes e_j$, and $A_2 = \sum_{j \in \N} \rho_j f_j \otimes f_j$ are eigendecompositions, then $|\widehat{\lambda}_i^K \widehat{\rho}_j^K - \lambda_i \rho_j|^2$ follows the rate given in \eqref{eq:rates}, and if the eigensubspace associated with $e_j$ is one-dimensional, then also $\| \widehat{e}_j^K - \sign(\langle \widehat{e}_j^K, e_j \rangle)e_j\|_2^2$ follows the rate given in \eqref{eq:rates}.
\end{theorem}

\begin{remark}
Since the roles of $A_1$ and $A_2$ are symmetric, one naturally obtains the rates for $\widehat{f}_j^K$ as well. Secondly, under slightly stricter assumptions, one can show that the rate in \eqref{eq:rates} is valid also in the uniform norm, see Theorem \ref{thm:rates_uniform} in Appendix~\ref{app:G5}. Finally, a version of Theorem \ref{thm:rates} with the bandwidth chosen adaptively as in Section \ref{sec:delta} is also valid, see Theorem \ref{thm:rates_adaptive_discrete} in Appendix~\ref{app:G6}.
\end{remark}

The proofs are postponed to Appendix~\ref{app:G}, but we make several comments here.
Firstly, there is a concentration in $K$ due to shifted partial tracing (recall Figure \ref{fig:PT_visualized}), hence the variance of the errors is allowed to grow with $K$ as stated in Theorem \ref{thm:rates}. 
Secondly, the estimators $\widehat{A}_1^K$ and $\widehat{A}_2^K$ are only defined if $\Tr{}^{d_K}(\widehat{\mathbf{C}}_N^K) \neq 0$. Since $\widehat{\mathbf{C}}_N^K \to \mathbf{C}^K$ for $N \to \infty$ entry-wise apart from the diagonal, we have $\Tr{}^{d_K}(\widehat{\mathbf{C}}_N^K) \to \Tr{}^{d_K}(\mathbf{C}^K)$, so we require $\Tr{}^{d_K}(\mathbf{C}^K) \neq 0$. Due to continuity of the kernel $c$ and the fact that $d_K \to \delta$ for $K \to \infty$, the assumption $\Tr{}^\delta(A) \neq 0$ implies $\Tr{}^{d_K}(\mathbf{C}^K) \neq 0$ for a sufficiently large $K$. This is the only reason why we require $K$ larger than a certain $K_0$ in order for the $\mathcal{O}_P(N^{-1})$ term to be uniform in $K$.
Finally, the Lipschitz continuity assumption allows us to bound the bias while the fourth-order moment condition on data allows us to bound the variance. The bulk of the proof has to do with controlling the variance, and doing so uniformly in the grid size. Also, the Lipschitz continuity assumption can be weakened. For example, continuity almost everywhere is sufficient for the bias to converge to zero, though without an explicit rate in~$K$.

In case the banded part of the covariance is also of interest, the same rates can be achieved in the noiseless setting ($\sigma^2=0$) under smoothness assumptions on the banded part. Without the assumption of stationarity on $B$, i.e.~without Toeplitz \mbox{averaging, one has:}
\[
\begin{split}
\verti{\widehat{B}^K - B}_2 &\leq \verti{\widehat{B}^K - B^K}_2 + \verti{B^K - B}_2 \\
&\leq \verti{\widehat{C}_N^K - \widehat{A}_1^K \ct \widehat{A}_2^K - (C^K - A_1^K \ct A_2^K)}_2 + \verti{B^K - B}_2 \\
&\leq \verti{\widehat{C}_N^K - C^K}_2 + \verti{ \widehat{A}_1^K \ct \widehat{A}_2^K - A_1^K \ct A_2^K}_2 + \verti{B^K - B}_2
\end{split}
\]
where the separable term can be treated as before and $\verti{\widehat{C}_N^K - C^K}_2$ can be bounded similarly. When Toeplitz averaging is used, nothing essential changes in the noiseless case.

The noisy case ($\sigma^2>0$) is trickier however, because we cannot estimate the diagonal of $B$. In such a case, one would need to smooth the estimated symbol of $B$ as in \cite{yao2005}. We omit the details here. However, we note that full covariance smoothing is obviously not computationally tractable, hence any smoothing should either be applied on the level of data (pre-smoothing) or on the level of the estimated 2D parts of the covariance (post-smoothing). Nonetheless, as exemplified by the previous theorem, the mere presence of noise does not call for smoothing when the target of inference is the separable component.

\begin{remark}
In the noiseless case, the convergence rates in Theorem \ref{thm:rates} are immediately applicable to the special case of a separable model and standard (non-shifted) partial tracing, as used by \cite{aston2017}. In the noisy case, however, shifted partial tracing (with an arbitrarily small shift) is needed to remove the noise. Due to continuity, a small shift should have a small impact on the quality of the estimator. Hence it might be recommended to always use shifted partial tracing with the minimal possible shift instead of the standard (non-shifted) partial tracing.
\end{remark}

\section{Empirical Demonstration}\label{sec:empirical}

In this section, we demonstrate how our methodology can be used to estimate a covariance from surface data observed on a grid, and how it compares to the empirical covariance estimator and the separable model, estimated via partial tracing \cite{aston2017} or as the nearest Kronecker product \cite{genton2007}. We begin with simulated data in Section \ref{sec:simulation}, where we focus on weakly dependent contamination of separability, and then move on to real data in Section \ref{sec:mortality_rates}, where we find evidence for heteroscedastic white noise contamination.

\subsection{Simulation Study}\label{sec:simulation}

The data generation procedure is as follows. Firstly, we create covariances $\mathbf A_1, \mathbf A_2 \in \R^{K \times K}$ and draw $\mathbf Y_1, \ldots, \mathbf Y_N$ independently from the matrix-variate Gaussian distribution with mean zero and covariance $\mathbf A = \mathbf A_1 \ct \mathbf A_2$. Secondly, we draw enough $\mathcal{N}(0,1)$ entries (independent of everything), arrange them on a grid, and perform space-time averaging using a window of size $d \in \{1,3,\ldots,19\}$ to obtain a sample $\mathbf W_n$ for every $n=1,\ldots,N$. This sample is drawn from a distribution with mean zero and covariance $\mathbf B \in \R^{K \times K \times K \times K}$, which is by construction stationary and banded by $d$. We set the sample size $N=300$ and the grid size $K=100$, so the discrete bandwidth $d$ approximately corresponds to the continuous bandwidth $\delta$ in percentages. Finally, we form our data set $\mathbf X_1, \ldots, \mathbf X_N \in \R^{K \times K}$ as
$\mathbf X_{n} = \sqrt{\tau} \mathbf Y_n + \mathbf W_n$, $n=1,\ldots,N$,
where $\tau \geq 1$. Thus $\mathbf X_1, \ldots, \mathbf X_N \in \R^{K \times K}$ are drawn from a zero-mean distribution with a separable-plus-banded covariance $\mathbf C = \tau \mathbf A_1 \ct \mathbf A_2 + \mathbf B$. Since $\mathbf A_1$, $\mathbf A_2$ and $\mathbf B$ are standardized to have norm one, $\tau$ can be understood as signal-to-noise ratio. The separable constituents $\mathbf A_1$ and $\mathbf A_2$ are chosen both as rank-7 covariances with linearly decaying eigenvalues and shifted Legendre polynomials as the eigenvectors, while $\mathbf B$ corresponds to the covariance of a spatio-temporal moving average process, as per the construction above. See Appendix~\ref{app:H} for a detailed description and for other simulation results in additional setups.

Note that our methodology based on shifted partial tracing first estimates the separable part of the model, and subsequently estimates $\mathbf B$ using the estimates for the separable part. Therefore the signal-to-noise ratio $\tau$ naturally influences difficulty of the estimation problem. The second parameter governing the difficulty of the estimation problem is the bandwidth $d$. However, the effect of $d$ is discontinuous: a small $d$ does not correspond to a nearly separable model; only $d=0$ formally leads to separable model with no contamination.

The following methods were used to estimate $\mathbf C$:
SPT-$d$ -- shifted partial tracing, the proposed methodology of Section \ref{sec:estimation}, provided with the true bandwidth $d$;
SPT-CV -- shifted partial tracing with $\delta$ chosen as in \eqref{eq:CV_empirical};
PT -- partial tracing \cite{aston2017}, an approach assuming separability;
NKP -- nearest Kronecker product \cite{genton2007}, another approach assuming separability;
ECE -- the standard empirical covariance estimator.
For several different settings, we calculate the relative estimation error $\| \mathbf C - \widehat{\mathbf C} \|_F/\| \mathbf C \|_F$,
where $\widehat{\mathbf C}$ is an estimator computed by one of the above-listed methods. The plots also show the \emph{bias} of a separable estimator, calculated as the best separable approximation to the true covariance $\mathbf C$ \cite{vanloan1993}.

\begin{figure}[!t]
   \advance\leftskip-0.3cm
   \begin{tabular}{ccc}
   \includegraphics[width=0.32\textwidth]{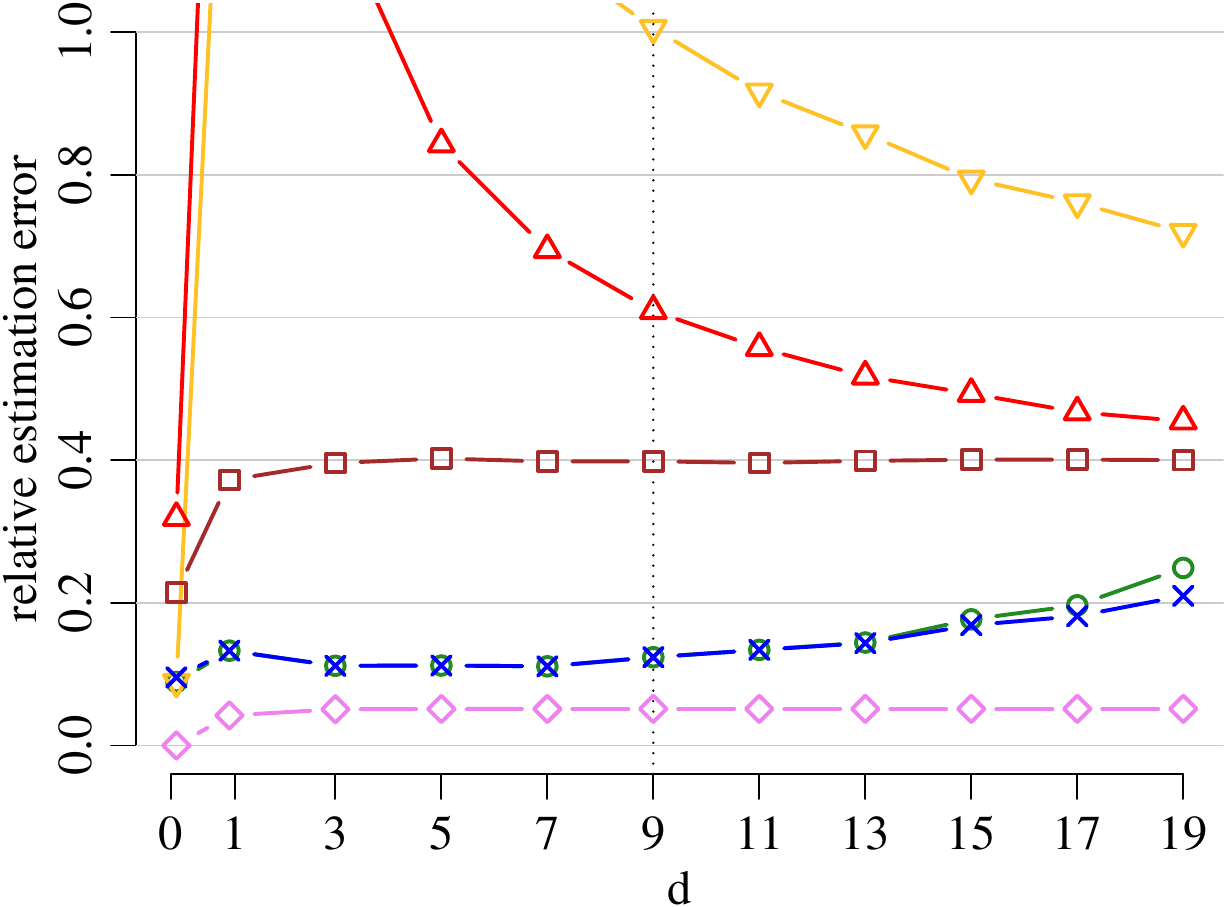} &
   \includegraphics[width=0.32\textwidth]{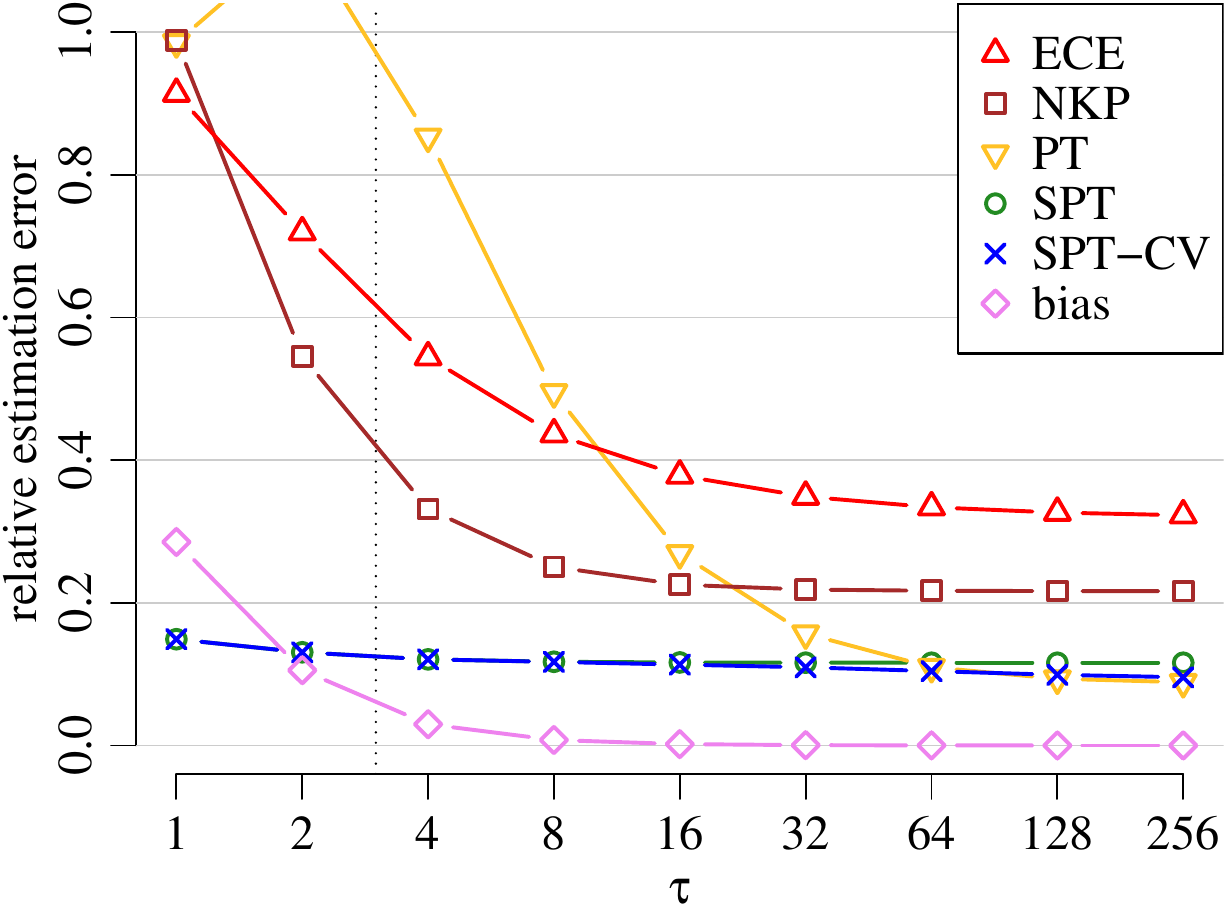} &
   \includegraphics[width=0.32\textwidth]{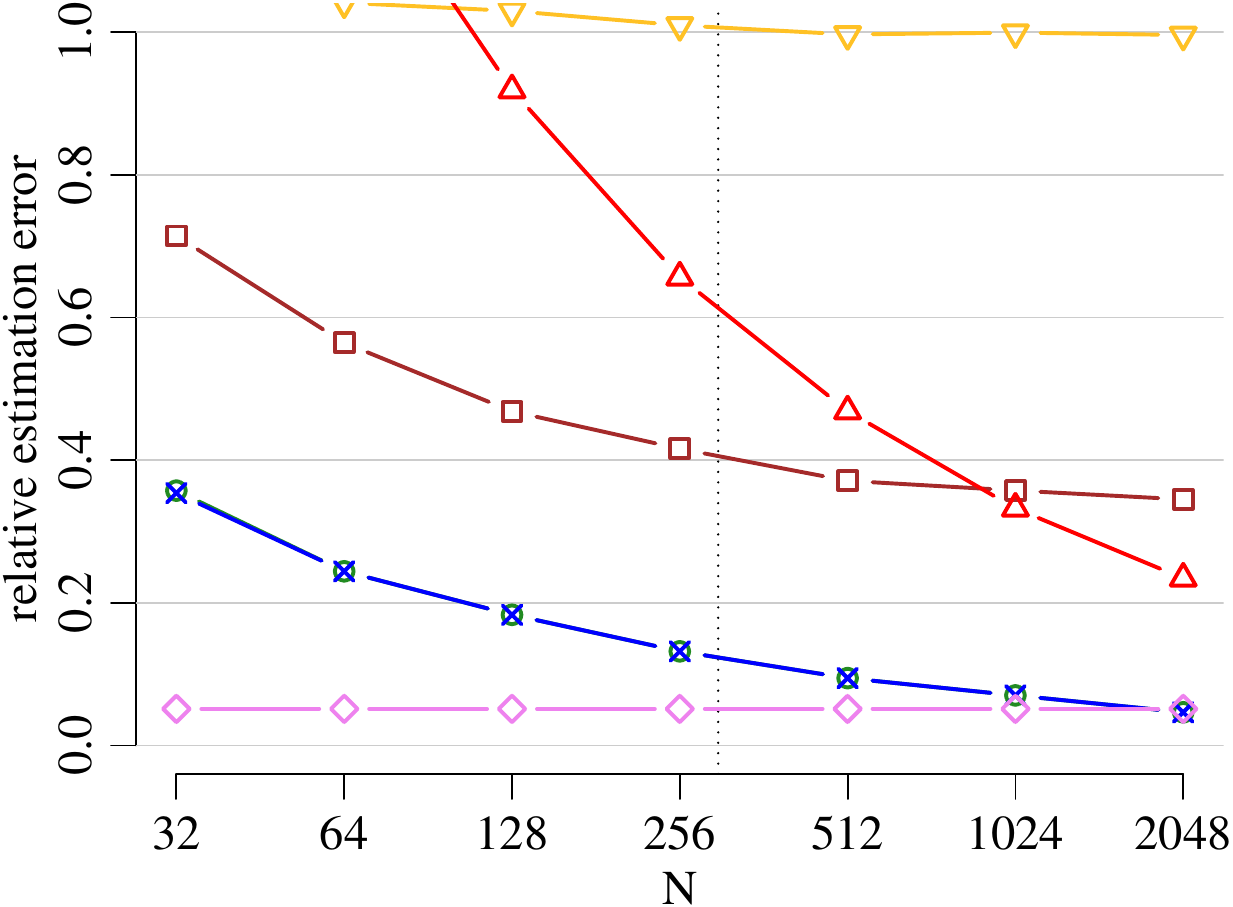} 
   \end{tabular}  
   \caption{Estimation errors for several competing methods with changing bandwidth $d$ (\emph{left}), signal to noise ratio $\tau$ (\emph{middle}), and sample size $N$ (\emph{right}). The vertical dotted lines show where every parameter is fixed for the remaining two plots (e.g. for the left plot, it is $\tau=3$ and $N=300$).}
    \label{fig:buf} 
\end{figure}

Figure \ref{fig:buf} depicts how the estimation error evolves when one of the three difficulty-governing parameters (bandwidth $d$, signal-to-noise ration $\tau$, and sample size $N$) varies, while the remaining two parameters are held fixed at any given plot (at $d=9$, $\tau=3$ or $N=300$). There are several remarks to be made about the results in Figures \ref{fig:buf}:
\begin{enumerate}
\item Shifted partial tracing outperforms both the separable model (estimated either by partial tracing or as the nearest Kronecker product) and the empirical covariance.
\item Bandwidth selection works well, leading to the same or even better performance than with known $\delta$ (see right end of the left and middle plots in Figure \ref{fig:buf}). This is because the banded part $B$ decays away from the diagonal, and sometimes choosing a smaller bandwidth than the true one can lead to a better bias-variance trade-off.
\item When the truth is separable (i.e. $d=0$) or nearly separable (i.e. $\tau$ large), partial tracing leads to the best results. In these cases, the bandwidth selection strategy correctly chooses a very small bandwidth, and hence the performance of SPT-CV matches the one of PT.
\item Note the extreme rise at the beginning of the error curves belonging to the empirical or the separable estimators in Figure \ref{fig:buf} (left). While $d=0$ corresponds to a separable model, $d=1$ is already quite non-separable. Even though the amount of non-separability (c.f.~the \emph{bias} curve) is rather low, it  is enough to substantially deteriorate performance of the separable estimators or the empirical covariance, while performance of the proposed methodology does not suffer too much.
\item The previous point is manifested again for large sample sizes $N$ (see Figure \ref{fig:buf}, right). While the amount of non-separability of $\mathbf C$ is still low and one would expect the performance of the separable estimators to be quite good, this is not the case. Altogether, we can say that presence of noise strikingly obstructs separable estimation.
\end{enumerate}

In the remainder of this section, we examine the functional nature of our problem, behavior of the ADI algorithm of Section \ref{sec:computational}, and the number of iterations needed by the algorithm to converge. We simulate data as described before in the Legendre case, but now we vary the grid size $K \in \{10(2j+1) ; j=1,\ldots,10 \}$, fix $\delta$ at 10 \% (i.e.~$d = K/10$), and we keep $\tau=3$ and $N=300$ for all the grid sizes.

\begin{figure}[!t]
   \centering
   \begin{tabular}{ccc}
   \includegraphics[width=0.35\textwidth]{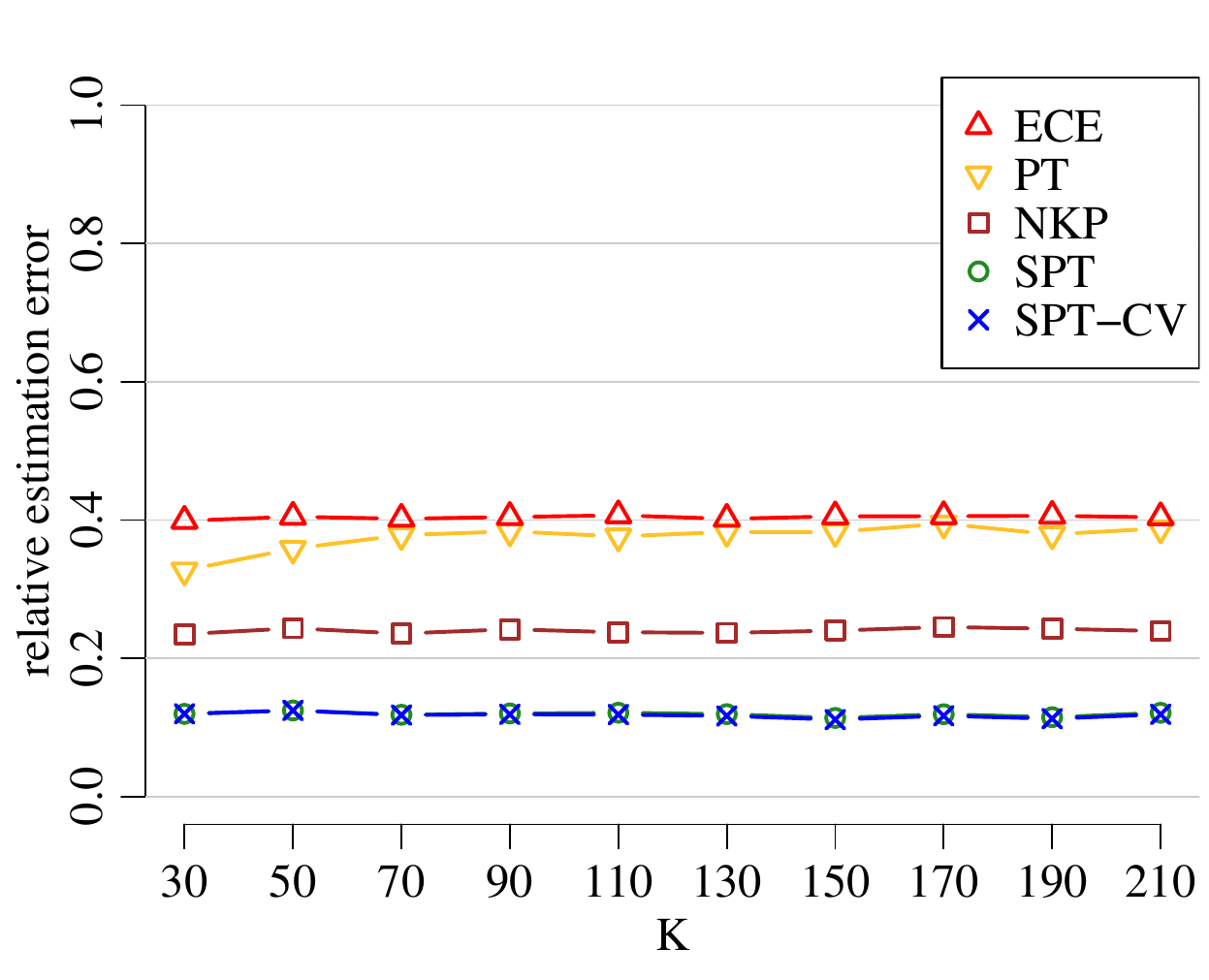} &
   \includegraphics[width=0.35\textwidth]{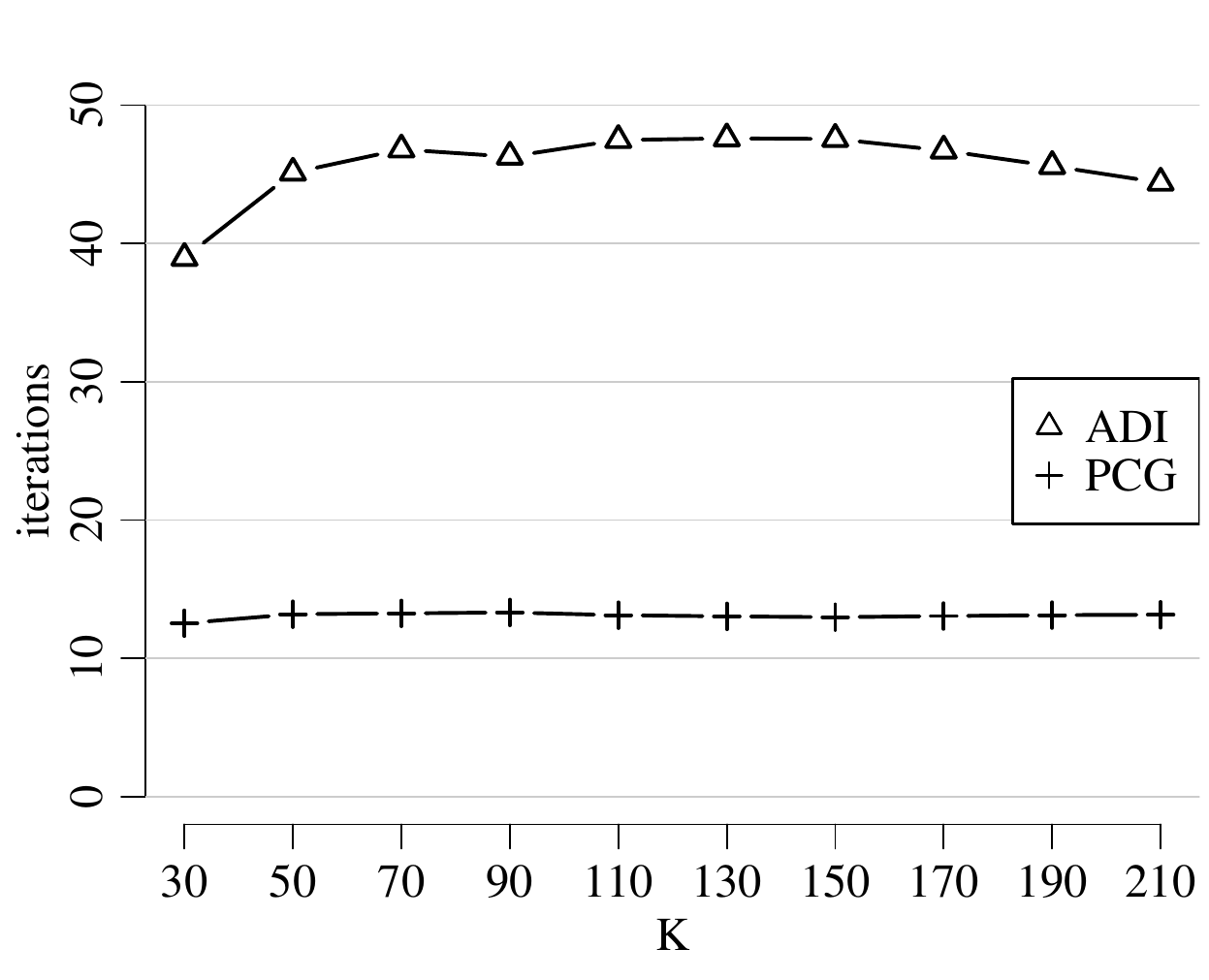}
   \end{tabular}  
   \caption{\emph{Left:} Estimation errors for several competing methods depending on the grid size $K$ with the bandwidth fixed at $d=K/10$. \emph{Right:} Number of iterations needed by the outer iteration scheme (ADI) and the inner iteration scheme (PCG) of the inversion algorithm of Section \ref{sec:inverse_alg}.}
    \label{fig:adi} 
\end{figure}

Let $\widehat{\mathbf C} = \widehat{\mathbf A}_1 \ct \widehat{\mathbf A}_2 + \widehat{\mathbf B}$ denote the estimator obtained by shifted partial tracing. $\widehat{\mathbf A}_1,\widehat{\mathbf A}_2$ and $\widehat{\mathbf B}$ are subsequently projected onto positive semi-definite matrices, as described in Appendix~\ref{app:E}. Also, a ridge regularization of order $10^{-5}$ is added to $\widehat{\mathbf C}$. This is not necessary, because $\widehat{\mathbf B}$ is positive definite, and thus the problem is well defined even without any ridge regularization. However, the performance of the ADI method heavily depends on the condition number of the system matrix, as is the case for any numerical method. Adding the ridge regularization ensures that the condition number stays roughly the same, regardless of $K$. Then, a random $\mathbf X \in \R^{K\times K}$ is generated, and we set $\mathbf Y= \widehat{\mathbf C} \mathbf X$. Subsequently, the ADI algorithm is called on the inverse problem $\widehat{\mathbf C} \mathbf X = \mathbf Y$ with $\widehat{\mathbf C}$ and $\mathbf Y$ given. The desired relative accuracy for the ADI scheme is set to $10^{-6}$. We do not report the relative reconstruction errors of $\mathbf X$, because these varied between $10^{-7}$ and $10^{-11}$ for every single run, leaving no doubt that the ADI scheme always converged to the truth with the desired precision. Instead, we report estimation errors and number of iterations needed by the ADI scheme in Figure \ref{fig:adi}.

\enlargethispage{2mm}
As suggested by our theoretical results, the relative estimation error does not depend on the grid size (see Figure \ref{fig:adi}, left). Additionally, both the number of outer iterations (ADI) and the number of inner iterations (PCG) does not seem to increase with the grid size (see Figure \ref{fig:adi}, right). This suggests super-linear convergence of the algorithm.

\subsection{Real Data}\label{sec:mortality_rates}

We analyze a data set $\mathbf X \in \R^{N \times K_1 \times K_2}$, where $\mathbf X[n,k_1,k_2]$ denotes the mortality rate for the $n$-th country, on the $k_1$-th calendar year and for subjects of age $k_2$. We consider the same set of 32 countries as \cite{chen2012,chen2017}, with $k_1$ ranging in the 50 year span 1964 -- 2014, and we too focus on the mortality rates of older individuals aged between $60 \leq k_2 < 100$. Hence $\mathbf X \in \R^{32 \times 50 \times 40}$. For a single country, we thus have a mortality rate surface of two arguments: the calendar year and the age of subjects in the population. This surface is observed discretely since both the calendar year and age are integers. Figure~\ref{fig:raw_mortalities} shows the raw mortality surfaces for two sample countries. The underlying continuous surfaces for different countries are assumed to be i.i.d. functional observations. The data were obtained from the Human Mortality Database \cite[www.mortality.org, downloaded on 12/4/2019]{wilmoth2007}.

\begin{figure}[!t]
\centering
   \begin{tabular}{cc}
   \includegraphics[width=0.47\textwidth]{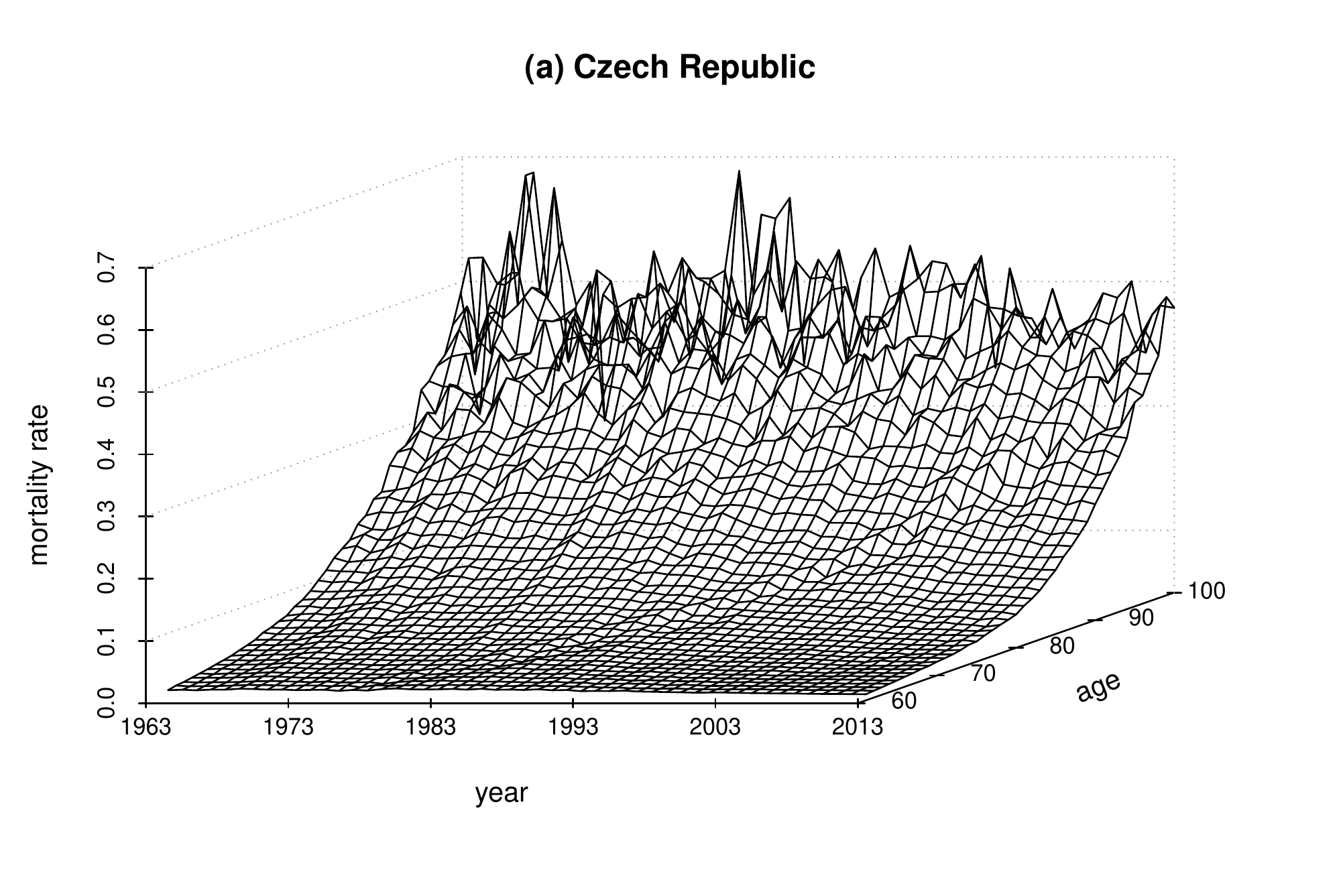} &
   \includegraphics[width=0.47\textwidth]{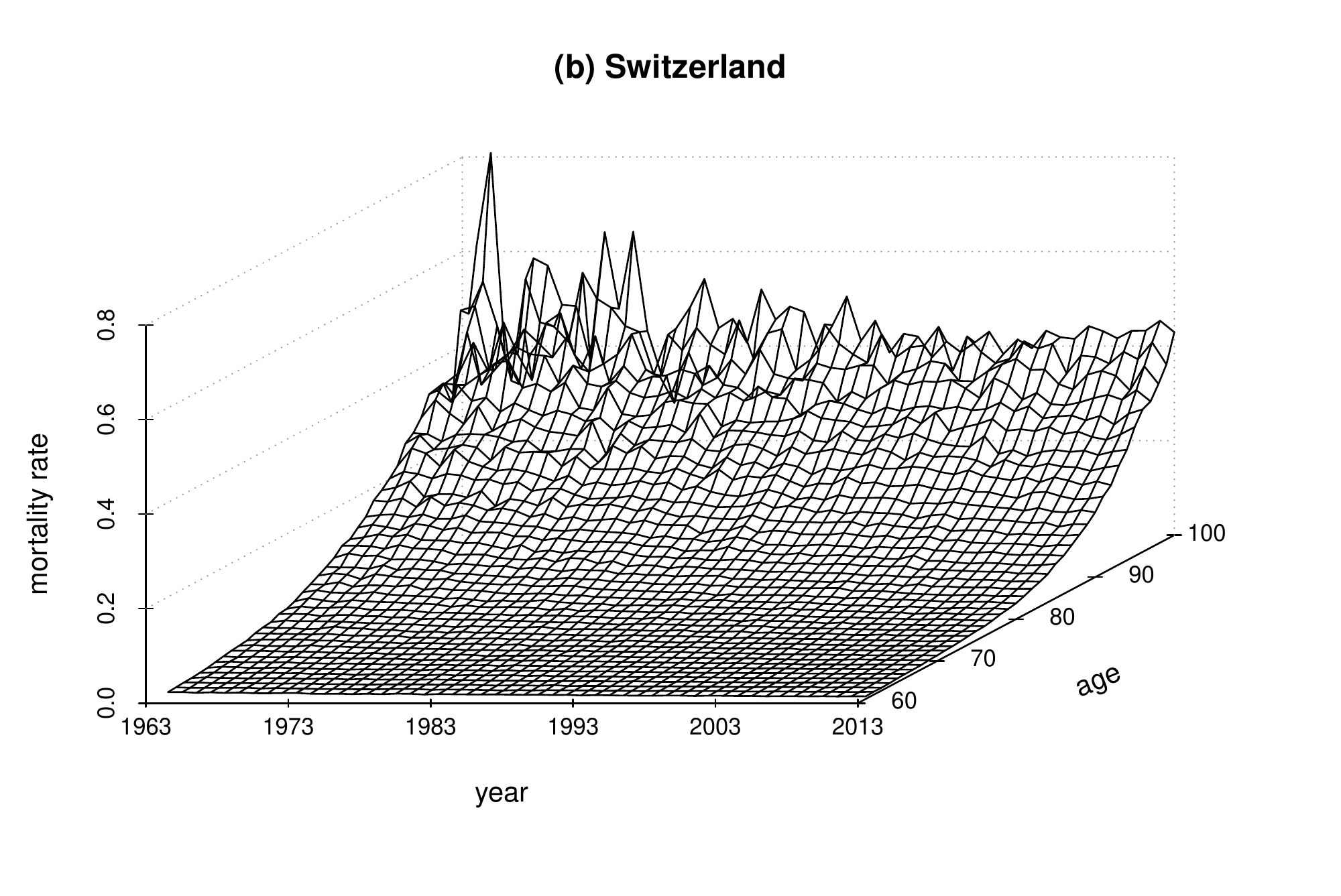}
   \end{tabular}  
   \caption{Raw mortality rate surfaces for the Czech Republic and Switzerland.}  
    \label{fig:raw_mortalities} 
\end{figure}

An in-depth analysis of mortality surfaces was provided in \cite{chen2012}. Mortality surfaces were also considered by the authors of \cite{chen2017}, who -- presumably motivated by \cite{aston2017} and aiming for computational efficiency -- calculated the so-called \textit{marginal kernels} $\Tr{1}(\widehat{\mathbf C}_N)$ and $\Tr{2}(\widehat{\mathbf C}_N)$, found the leading eigenfunctions of these marginal kernels, say $\{\widehat{\boldsymbol \phi}_i\}_{i=1}^I$ and $\{\widehat{\boldsymbol \psi}_j\}_{j=1}^J$, and used the tensor product approximation
$\widehat{\mathbf C}_N \approx \sum_{i=1}^I \sum_{j=1}^J \widehat{\gamma}_{ij} (\widehat{\boldsymbol \phi}_i \otimes \widehat{\boldsymbol \psi}_j) \otimes (\widehat{\boldsymbol \phi}_i \otimes \widehat{\boldsymbol \psi}_j)$,
where $\widehat{\gamma}_{ij} = \langle \widehat{\mathbf C}_N , (\widehat{\boldsymbol \phi}_i \otimes \widehat{\boldsymbol \psi}_j) \otimes (\widehat{\boldsymbol \phi}_i \otimes \widehat{\boldsymbol \psi}_j) \rangle$. Indeed, we highlight that using the marginal eigenfunctions as building blocks for a low-rank approximation of the empirical covariance can be meaningful even if the covariance $C$ is not separable \cite{lynch2018}.

\begin{figure}[!t]
   \centering
   \begin{tabular}{ccc}
   \raisebox{0.1cm}{\includegraphics[width=0.3\textwidth]{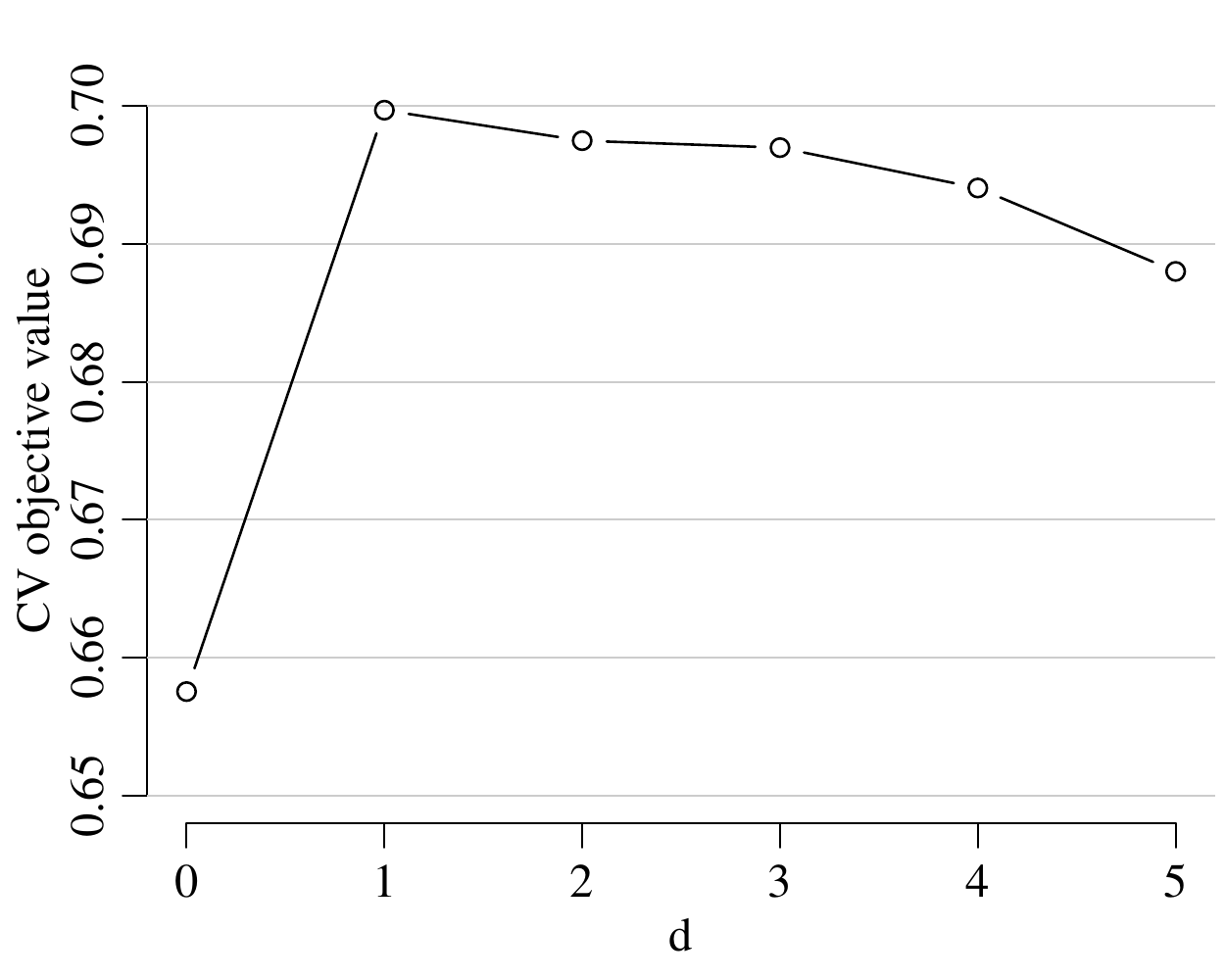}} & &
   \includegraphics[width=0.31\textwidth]{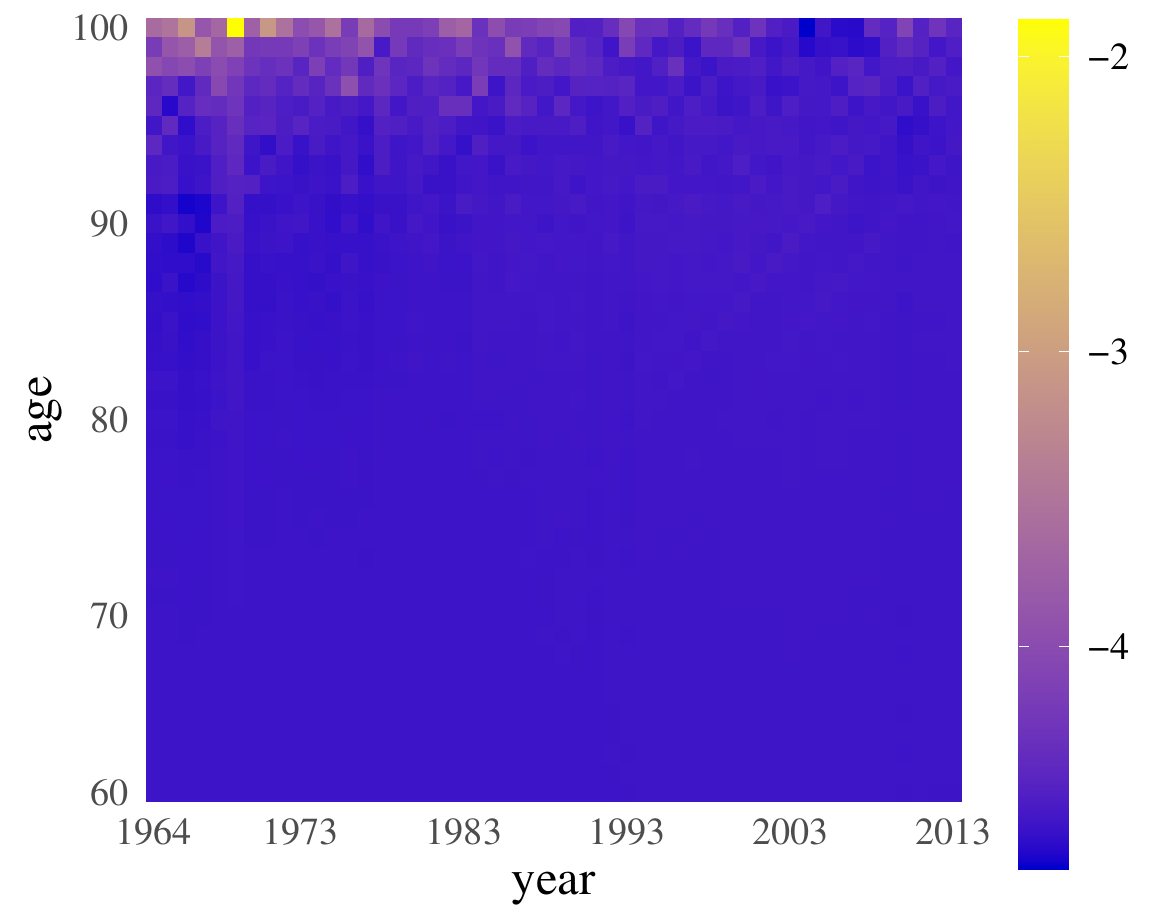}
   \end{tabular}   
   \caption{Cross-validation objective for the mortality data (left) and log-heatmap of the heteroscedastic white noise's variance (right).} 
    \label{fig:cv_B} 
\end{figure}

Compared to \cite{chen2012,chen2017}, we consider the mortality data with a slightly larger span of calendar years (the maximal span in which no data are missing). Our aim here is not to provide a novel analysis of the mortality dataset, but merely to illustrate the usefulness of shifted partial tracing.

Firstly, when investigating the sample curves in Figure \ref{fig:raw_mortalities}, it seems that the discrete observations of the mortality rate surfaces are observed with additional noise, which is likely heteroscedastic with variance increasing with the age of the subjects. This is presumably due to the fact that the size of the population of subjects of a given age decreases fast with increasing age. To probe whether the (most likely heteroscedastic) noise disrupts separability, we can use the bandwidth selection procedure. We do not assume stationarity, and we set the the estimator of the banded part \eqref{eq:estimator_B_non_stationary} to zero outside of the current bandwidth in every step. The objective of \eqref{eq:CV_empirical} is maximized at $\widehat d=1$. We plot the objective curve in Figure \ref{fig:cv_B}, providing a strong evidence for presence of noise. Since $\widehat d=1$, we are in the separable-plus-noise regime, which is computationally feasible even under heteroscedasticity. Figure~\ref{fig:cv_B} also shows a heatmap of the estimated variance (or rather its logarithm, for visualisation purposes) of the noise depending on the location. The heatmap is in alignment with the conjecture that the noise variance is increasing with age.

Secondly, we compare spectra of the \textit{marginal kernels} $\Tr{1}(\widehat{\mathbf C}_N)$ and $\Tr{2}(\widehat{\mathbf C}_N)$ to their shifted counterparts $\Tr{1}^1(\widehat{\mathbf C}_N)$ and $\Tr{2}^1(\widehat{\mathbf C}_N)$. When partial tracing is used to obtain the \textit{marginal kernels}, one has to keep 16 and 4 eigenfunctions, respectively, to capture 90~\% of the marginal variance (c.f.~\cite{lynch2018}) in both dimensions. When shifted partial tracing is used instead, one only needs to retain 4 and 2 eigenfunctions, respectively. Hence shifted partial tracing offers a more parsimonious representation. 

Thirdly, the empirical bootstrap test of \cite{aston2017} with 4 and 2 marginal eigenfunctions (which seems to be the most reasonable choice, also used by \cite{lynch2018}) leads to a borderline $p$-value of 0.06. The test of \cite{aston2017} can be generalized to testing separable-plus-banded model instead, see Appendix~\ref{app:I}. In comparison, the $p$-value for this test is over 0.4, suggesting that the separable-plus-banded model cannot be rejected for this data set.

\begin{figure}[!b]
   \centering
   \begin{tabular}{cccc}
   (a) $\widehat{\boldsymbol \phi}_1$ & (b) $\widehat{\boldsymbol \phi}_2$ & (c) $\widehat{\boldsymbol \psi}_1$ & (d) $\widehat{\boldsymbol \psi}_2$ \\
   \includegraphics[width=0.225\textwidth]{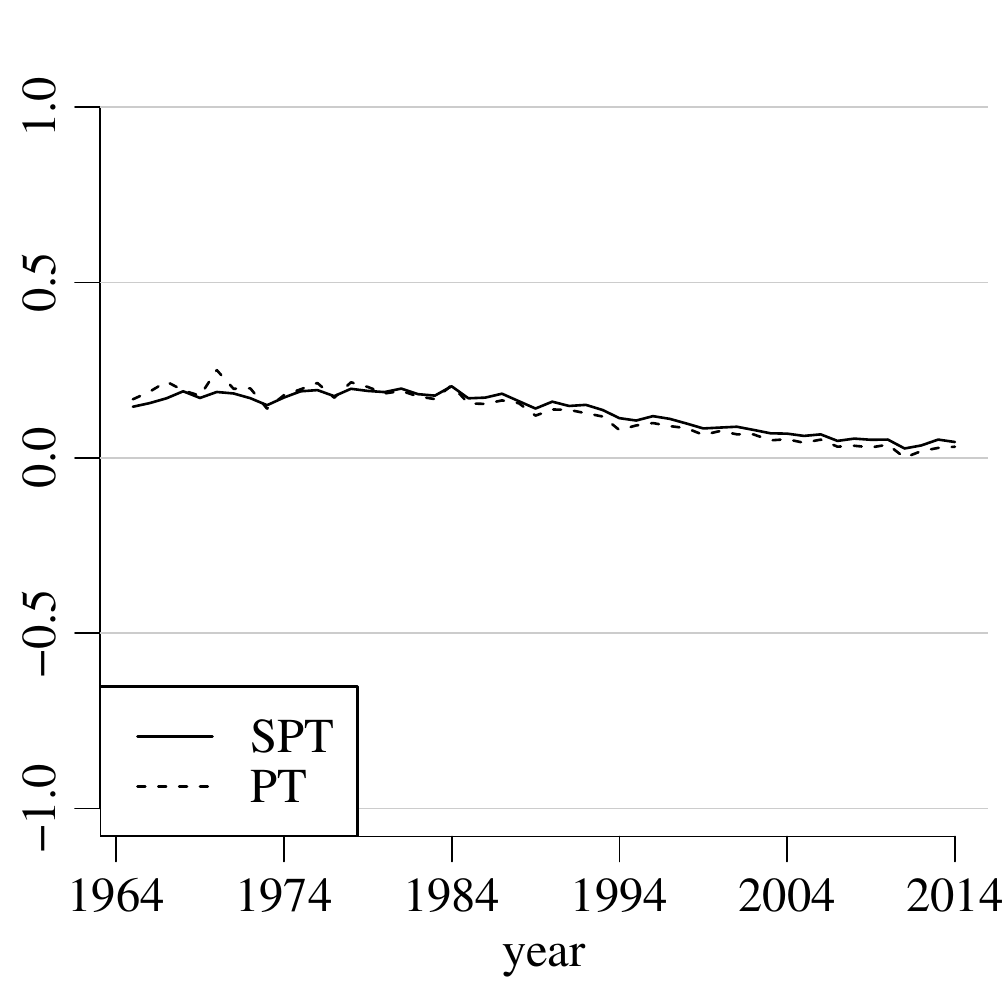} &
   \includegraphics[width=0.225\textwidth]{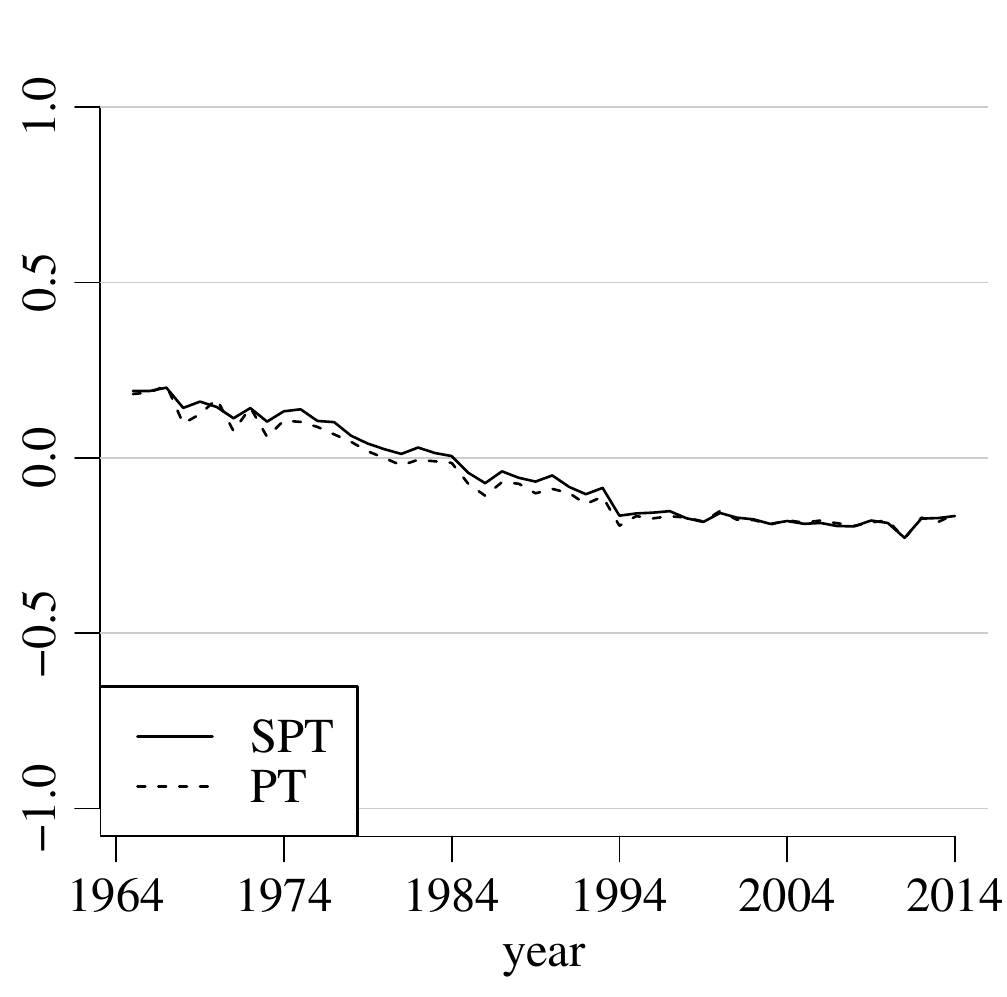} &
   \includegraphics[width=0.225\textwidth]{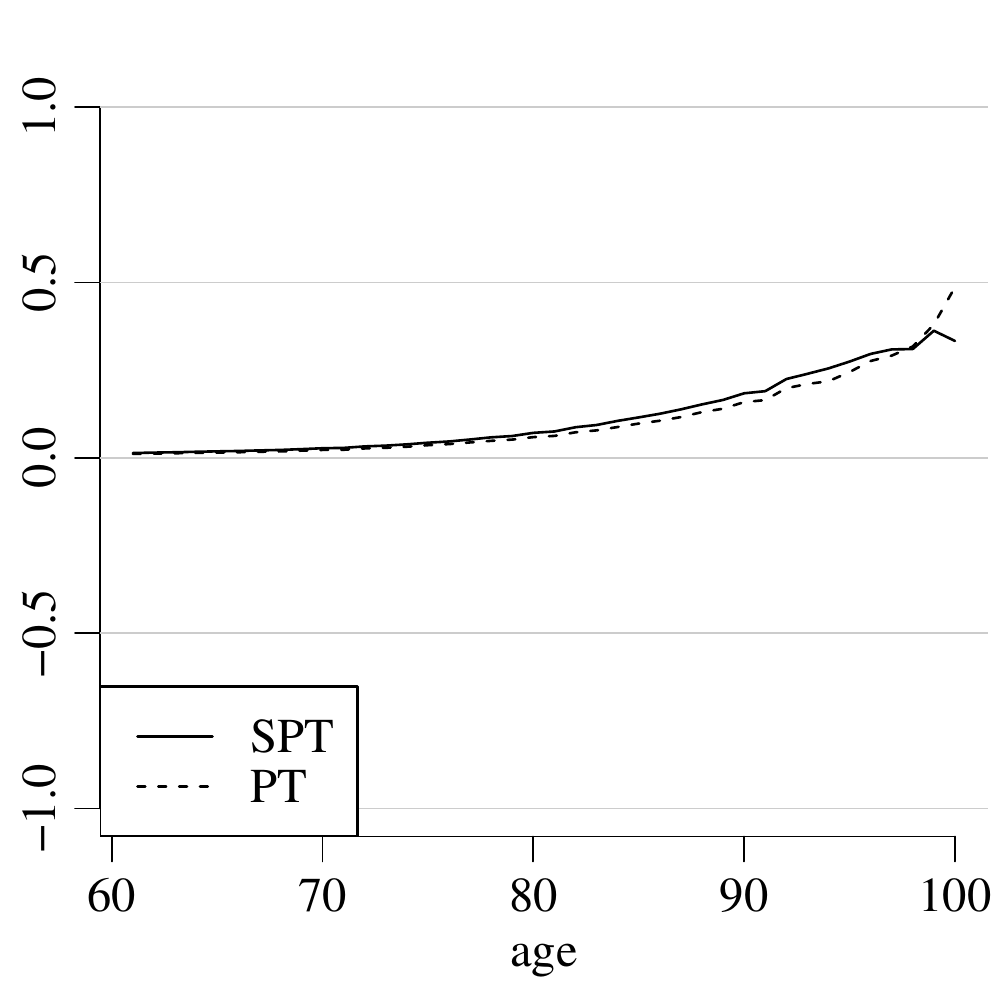} &
   \includegraphics[width=0.225\textwidth]{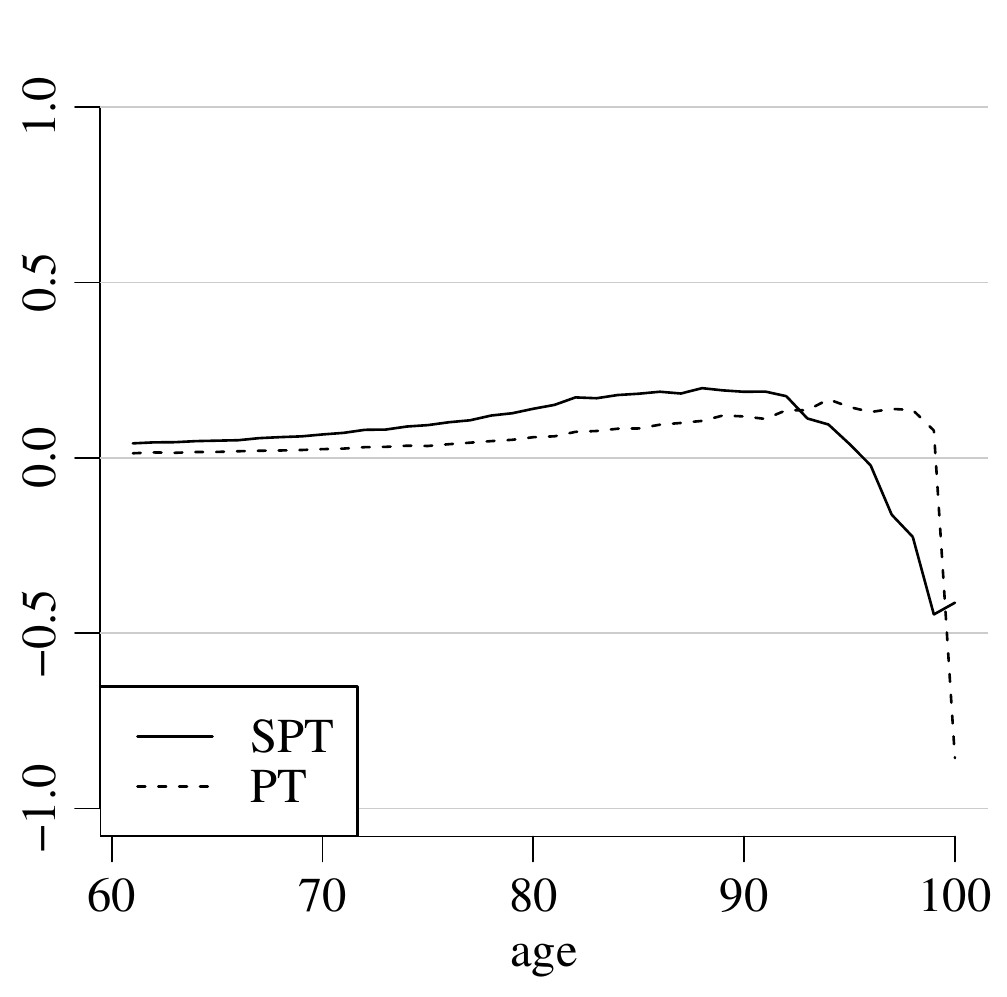}
   \end{tabular}    
   \caption{First two eigenfunctions of the marginal kernels obtained by partial tracing (PT) and shifted partial tracing (SPT).}
    \label{fig:eigenfunctions} 
\end{figure}

Finally, keeping only 2 eigenfunctions in both dimensions (explaining 83~\% and 96~\%  of the variance, respectively) leads to a plausible interpretation, when shifted partial tracing is used. The eigenfunctions are plotted in Figure \ref{fig:eigenfunctions}. The first eigenfunctions in both dimensions capture the overall trend: $\widehat{\boldsymbol \phi}_1$ captures the decreasing variance in calendar years (the first dimension) and $\widehat{\boldsymbol \psi}_1$ the increasing variance in age (the second dimension). The second eigenfunction in the first dimension $\widehat{\boldsymbol \phi}_2$ distinguishes between countries having either a ``U-shape'' (the Czech Republic, for example) or reversed ``U-shape'' in calendar years (Switzerland, for example). This ``U-shape'' is more prominent in older ages, but it is too subtle to be visible by eye in the raw data plotted in Figure \ref{fig:raw_mortalities}. Finally, the second eigenfunction in the second dimension $\widehat{\boldsymbol \psi}_2$ contrasts the old age (around 85) and the oldest age (post 90) mortalities. However, this is only the case if shifted partial tracing is used. The eigenfunction $\widehat{\boldsymbol \psi}_2$ obtained from $\Tr{2}(\widehat{\mathbf C}_N)$ does not have this interpretation; it is in fact not interpretable. Interestingly, the same qualitative conclusions as those drawn here by using shifted partial tracing were drawn in \cite{chen2012} based on different methodology.

\section{Discussion}

The immense popularity of separability stems mainly from the computational advantages it entails. The separable-plus-banded model we propose is an additive generalization of separability. To retain the computational advantages, many natural operations (such as forming the empirical covariance estimator or naively inverting the estimated model) are prohibited. Efficient estimation of the separable-plus-banded model can be achieved with shifted partial tracing -- a novel methodology for estimation of the covariance of surface-valued processes, working on the level of data.

From another point of view, shifted partial tracing can be used to estimate a separable model, even when data are corrupted by heteroscedastic and/or weakly dependent noise. The noise can be incorporated for tasks such as prediction with no computational overhead compared to the simple separable model, whenever it (or equivalently the banded part of the model) does not exceed the separable model in terms of the degrees of freedom it possesses. The latter is true for example when the noise is heteroscedastic and white (as was the case in the real data analysis) or when the noise is weakly dependent but stationary (as was the case in the simulation study). If the noise is both heteroscedastic and weakly dependent, our methodology can still be used. However, one has to pay extra computational costs (compared to just a noiseless separable model), if one wishes to explicitly work with a~noise structure, which has higher complexity than that of the separable model.

Following \cite{aston2017}, partial tracing has become the method of choice for calculation of the marginal kernels, i.e.~for calculating a separable proxy of the covariance. However, given the theoretical development here and the practical evidence found in the mortality dataset, it seems that shifted partial tracing should in general be preferred, due to its denoising properties.

A straightforward extension of our work would be to consider other forms of parsimony than stationarity, which can be imposed on the banded part of the covarince to maintain the computational advantages of our model. In principle, any form of parsimony, which reduces the complexity of storing $\mathbf B$ to $\O(K^2)$ and the number of flops required to apply $\mathbf B$ to $\O(K^3)$, may be considered. The precise form of parsimony capable of complementing separability well in the additive model will depend on the specific application.

\subsubsection*{Acknowledgements}

We thank Prof. John Aston (University of Cambridge), Prof. Daniel Kressner (EPFL), and Dr. Shahin Tavakoli (University of Geneva) for very fruitful discussions.

\newpage

\section{Appendices}
\addtocontents{toc}{\protect\setcounter{tocdepth}{1}}
\renewcommand{\thesubsection}{\Alph{subsection}}
\numberwithin{lemma}{subsection}
\numberwithin{definition}{subsection}
\numberwithin{proposition}{subsection}
\numberwithin{equation}{subsection}
\numberwithin{corollary}{subsection}
\numberwithin{remark}{subsection}
\numberwithin{theorem}{subsection}
\numberwithin{example}{subsection}
\numberwithin{figure}{subsection}

The appendices generalizes the development of shifted partial tracing (Appendix~\ref{app:A}) and provides some computational details (Appendices~\ref{app:B}-\ref{app:E}) and proofs of the asymptotic results (Appendix~\ref{app:G}). Furthermore, additional simulation results are provided (Appendix~\ref{app:H}), and testing for the validity of the separable-plus-banded model is discussed (Appendix~\ref{app:I}). References to the main body of the paper are made in the standard way, while equations, lemmas, etc., that are new to the appendices are labelled and referred to within sub-sections.

\subsection{Background Concepts}\label{app:A}

Let $\mathcal H_1$ and $\mathcal H_2$ be real, complete, separable Hilbert spaces equipped with inner products $\langle \cdot, \cdot \rangle_{\mathcal H_1}$ and $\langle \cdot, \cdot \rangle_{\mathcal H_2}$, and corresponding norms $\| \cdot \|_{\mathcal H_1}$ and $\| \cdot \|_{\mathcal H_2}$, respectively. A linear transformation $F: \mathcal H_1 \to \mathcal H_2$ is \emph{bounded} if 
\[
\verti{F}_\infty := \displaystyle\sup_{\| x\|_{\mathcal H_1}=1} \| F x \|_{\mathcal H_2} < \infty \,.
\]
Bounded linear transformations are called \emph{operators}. The set of operators from $\mathcal H_1$ to $\mathcal H_2$ equipped with the \emph{operator norm} $\verti{\cdot}_\infty$ is a Banach space, denoted by $\mathcal{S}_\infty(\mathcal H_1, \mathcal H_2)$. Operator $F \in \mathcal{S}_\infty(\mathcal H_1, \mathcal H_2)$ is \emph{compact} if its action on an arbitrary $x \in \mathcal H_2$ can be written as
\begin{equation}\label{eq:SVD}
F x = \sum_{j=1}^\infty \sigma_j \langle e_j, x \rangle f_j\,,
\end{equation}
where $\{\sigma_j\}_{j=1}^\infty$ is a non-negative and non-increasing sequence of singular values, and $\{ e_j \}_{j=1}^\infty$ and $\{ f_j \}_{j=1}^\infty$ are orthonormal bases of $\mathcal H_1$ and $\mathcal H_2$, and the series converges in the operator norm. For $p \in [1,\infty)$, the set of all such compact operators $F: \mathcal H_1 \to \mathcal H_2$ such that $\verti{F}_p := (\sum_{j=1}^\infty \sigma_j^p)^{1/p} < \infty$ is denoted by $\mathcal{S}_p(\mathcal H_1,\mathcal H_2)$, and it is a Banach space, when equipped with the \emph{Schatten-$p$ norm} $\verti{\cdot}_p$. It holds $\mathcal{S}_p(\mathcal H_1,\mathcal H_2) \subset \mathcal{S}_q(\mathcal H_1,\mathcal H_2)$ for $p < q$. We further abbreviate $\mathcal{S}_p(\mathcal H_1,\mathcal H_1) =: \mathcal{S}_p(\mathcal H_1)$ and denote $\mathcal{S}_p^+(\mathcal H_1)$ the set of all \emph{positive semi-definite} operators that belong to  $\mathcal{S}_p(\mathcal H_1)$.

We are particularly interested in the case $p=1$. $\mathcal{S}_1(\mathcal H_1, \mathcal H_2)$ is the space of \emph{trace-class} operators, and $\verti{\cdot}_1$ is called \emph{trace norm} or \emph{nuclear norm}. For $F \in \mathcal{S}_1(\mathcal H_1, \mathcal H_2)$ we define its \emph{trace} as $\Tr{}(F) :=  \sum\nolimits_{j=1}^\infty \langle T e_j, e_j \rangle_{\mathcal H_2}$.

The tensor product space of $\mathcal H_1$ and $\mathcal H_2$, denoted by $\mathcal H := \mathcal H_1 \otimes \mathcal H_2$, is defined as the completion of the set of finite linear combinations of abstract tensor products (c.f. \cite{weidmann2012})
\begin{equation}\label{eq:setOfTensors}
\Big\{ \sum_{j=1}^N x_j \otimes y_j \,;\, x_j \in \mathcal H_1, y_j \in \mathcal H_2, N \in \N \Big\}
\end{equation}
under the inner product $\langle x_1 \otimes y_1, x_2 \otimes y_2 \rangle_H := \langle x_1, x_2\rangle_{\mathcal H_1} \langle y_1, y_2\rangle_{\mathcal H_2}$, for all $x_1,x_2 \in \mathcal H_1$ and $y_1,y_2 \in \mathcal H_2$. More precisely, set \eqref{eq:setOfTensors} is a vector space equipped with the inner product $\langle \cdot, \cdot \rangle_H$. Thus its completion $\mathcal H$ is a Hilbert space. If $\{ e_j \}$ and $\{f_j\}$ are orthonormal bases in $\mathcal H_1$ and $\mathcal H_2$, then $\{ e_i \otimes f_j \}_{i,j=1}^\infty$ is an orthonormal basis of $\mathcal H$.
$\mathcal H$ is isometrically isomorphic to $\mathcal{S}_2(\mathcal H_1,\mathcal H_2)$ and also to $\mathcal{S}_2(\mathcal H_2,\mathcal H_1)$.

The previous construction of product Hilbert spaces can be generalized to Banach spaces $\mathcal B_1$ and $\mathcal B_2$. That is, one can define $B := \mathcal B_1 \otimes \mathcal B_2$ in a similar way. The only difference is, that the completion is done under the product norm $\| x \otimes y \|_B := \| x \|_{\mathcal B_1} \| y \|_{\mathcal B_2}$, for $x \in \mathcal B_1$ and $y \in \mathcal B_2$.

Consider now $\mathcal B_1 := \mathcal{S}_p(\mathcal H_1)$ and $\mathcal B_2 := \mathcal{S}_p(\mathcal H_2)$. We construct the tensor product space $\mathcal{S}_p(\mathcal H_1) \otimes \mathcal{S}_p(\mathcal H_2)$ as described above. Now, consider a linear mapping $\Phi: \mathcal{S}_p(\mathcal H_1) \otimes \mathcal{S}_p(\mathcal H_2) \to \mathcal{S}_p(\mathcal H)$ defined on the abstract tensor products as
\[
\Phi(A_1 \otimes A_2) = A_1 \ct A_2 \,, \quad A_1 \in \mathcal{S}_p(\mathcal H_1), A_2 \in \mathcal{S}_p(\mathcal H_2)\,,
\]
where $A_1 \ct A_2: \mathcal H \to \mathcal H$ is the linear operator defined on the abstract tensor products in $\mathcal H$ as
\begin{equation}\label{eq:operatorTensorProduct}
(A_1 \ct A_2)(x \otimes y) = A_1 x \otimes A_2 y\,, \quad x \in \mathcal H_2, y \in \mathcal H_2.
\end{equation}
$\Phi$ is the isomorphism between $\mathcal{S}_p(\mathcal H_1) \otimes \mathcal{S}_p(\mathcal H_2)$ and $\mathcal{S}_p(H)$. Thus we showed that the space of Schatten-$p$ operators on an abstract tensor product space is isometrically isomorphic to the abstract tensor product space of two Schatten-$p$ operator spaces. We prefer the former point of view and, for \mbox{$A_1 \in \mathcal{S}_p(\mathcal H_1)$} and $A_2 \in \mathcal{S}_p(\mathcal H_2)$, $A_1 \ct A_2 \in \mathcal{S}_p(\mathcal H)$ is the unique operator satisfying \eqref{eq:operatorTensorProduct}. By the abstract construction, we also have \mbox{$\verti{A_1 \ct A_2}_p = \verti{A_1}_p \verti{A_2}_p$}.

\begin{lemma}\label{lem:ct}
Let $A_1 \in \mathcal{S}_p(\mathcal H_1)$ and $A_2 \in \mathcal{S}_p(\mathcal H_2)$ are self-adjoint with eigenvalue-eigenvector pairs $\{ (\lambda_j, e_j) \}$ and $\{ (\rho_j, f_j) \}$. Then $A_1 \ct A_2$ is self-adjoint with eigenvalue-eigenvector pairs $\{ (\lambda_i \rho_j, e_i \otimes f_j) \}_{i,j=1}^\infty$. Furthermore, for $p=1$, it holds $\Tr{}(A_1 \ct A_2) = \Tr{}(A_1)\Tr{}(A_2)$.
\end{lemma}
\begin{proof}
Let $x \in \mathcal H_1$ and $y \in \mathcal H_2$, then
\[
\begin{split}
(A_1 \ct A_2)(x \otimes y) &= \Big[\Big( \sum_{j=1}^\infty \lambda_j e_j \otimes e_j \Big) \ct \Big( \sum_{j=1}^\infty \rho_j f_j \otimes f_j \Big)\Big](x \otimes y) \\
&= \Big( \sum_{j=1}^\infty \lambda_j e_j \otimes e_j \Big) x \otimes \Big( \sum_{j=1}^\infty \rho_j f_j \otimes f_j \Big)y \\
&=\Big[\Big( \sum_{j=1}^\infty \lambda_j \langle e_j, x \rangle_{\mathcal H_1} e_j \Big) \otimes \Big( \sum_{j=1}^\infty \rho_j \langle f_j, y \rangle_{\mathcal H_1} f_j \Big)\Big] \,.
\end{split}
\]
For the choice of $x = e_k$ and $y = f_l$ for $k,l \in \N$ we have
\[
(A_1 \ct A_2)(e_k \otimes f_l) = \lambda_k e_k \otimes \rho_l f_l = \lambda_k \rho_l (e_k \otimes f_l)\,,
\] 
which shows that $e_k \otimes f_l$ is an eigenvector of $A_1 \ct A_2$ associated with the eigenvalue $\lambda_k \rho_l$.

The additional part follows from the previous one, since by Fubini's theorem
\[
\sum_{i,j=1}^\infty \lambda_i \rho_j = \Big(\sum_{i=1}^\infty \lambda_i \Big) \Big(\sum_{j=1}^\infty \rho_j \Big) \,.
\]
\end{proof}

The previous lemma can be naturally extended to singular values and singular vectors of operators, that are not self-adjoint, only the notation gets little more complicated. Moreover, we have the following characterization result.

\begin{corollary}\label{cor:PSD}
$A_1 \ct A_2$ is self-adjoint (resp. positive semi-definite, resp. positive definite) if and only if both $A_1$ and $A_2$ are self-adjoint (resp. positive semi-definite, resp. positive definite).
\end{corollary}
\begin{proof}
It is trivial that $A_1 \ct A_2$ must be self-adjoint for $A_1$ and $A_2$ both self-adjoint.

In the other direction, assume that both $A_1$ and $A_2$ are non-zero, otherwise $A_1 \ct A_2$ is zero (because $\langle (A_1 \ct A_2) (u \otimes v), u \otimes v \rangle=0$ for all $u,v$) and the conclusion is trivial. Note that self-adjointness of $A_1 \ct A_2$ gives us
\[
\begin{split}
\langle (A_1 \ct A_2) (x_1 \otimes y_1), x_2 \otimes y_2 \rangle &= \langle x_1 \otimes y_1, (A_1 \ct A_2) (x_2 \otimes y_2) \rangle \\
\Rightarrow \quad 
\langle A_1 x_1, y_1 \rangle \langle A_2 x_2, y_2 \rangle &= \langle x_1, A_1 y_1 \rangle \langle x_2, A_2 y_2 \rangle 
\end{split}
\]
for $x_1,x_2 \in \mathcal H_1$ and $y_1,y_2 \in \mathcal H_2$ arbitrary. Now choose $(x_2,y_2)$ to be the left-right eigenvector pair of $A_2$ associated with a non-negative eigenvalue. We immediately see that $A_1$ must be self-adjoint, and similarly for $A_2$.

The parts about positive semi-definiteness and positive definiteness follow easily from the previous lemma.
\end{proof}

\subsection{General Definition of Shifted Partial Tracing}\label{app:B}

To ease the exposition, we assumed continuity in the definition of shifted partial tracing given in the main paper. But in order to prove the asymptotic results of Section \ref{sec:asymptotics}, it is necessary to generalize the notions of shifted (partial) tracing to general trace-class operators on $\mathcal{L}^2[0,1]^2$, i.e. to covariances of random elements on $\mathcal{L}^2[0,1]^2$, which are not necessarily continuous or have continuous sample paths. We do this by providing alternative definitions of the shifted (partial) traces, which neither require continuity nor positive semi-definiteness. These will be denoted by ``$\T{}$'', replacing ``$\Tr{}$'', to make the distinction. It will be shown subsequently that, under continuity, they coincide with Definition \ref{def:SPT}.

Note that the definition of the shifted partial tracing is not symmetric, meaning that the result of the shifted partial trace is not necessarily self-adjoint. We could define a~symmetrized shifted partial trace instead, but this is (due to linearity of shifted partial tracing and symmetry of the kernel $k$) equivalent to symmetrizing the result. The latter is used in practice for its computational convenience, while the former is hypothetically done in theory, but we avoid it without loss of generality to ease the presentation (see Appendix~\ref{app:E}). Also, note that any trace-class operator on $L^2[0,1]^D$ can be represented as a superposition of two trace-class self-adjoint operators, and any trace-class self-adjoint operator is in turn the difference between two trace-class positive semi-definite operators \cite{beauzamy1988}. Hence we can assume positive semi-definiteness without loss of generality, wherever all operations being performed are linear.

\begin{definition}\label{def:shifting_operator}
We define the shifting operator $\mathrm{S}^\delta: \mathcal{S}_1(\mathcal L^2[0,1]) \to \mathcal{S}_1(\mathcal L^2[0,1])$ by its action on kernels. For $F \in \mathcal{S}_1(\mathcal L^2[0,1])$ with a kernel $k=k(t,s)$, $\mathrm{S}^\delta(F)$ have kernel
\begin{equation}\label{eq:shifted_kernel}
k^\delta(t,s) = \begin{cases}
k(t,s+\delta) , \quad s < 1-\delta , \\
0 , \qquad \qquad \quad \text{otherwise} .
\end{cases}
\end{equation}
\end{definition}

It is straightforward to check that $\mathrm{S}^\delta$ is well-defined linear operator on $\mathcal{S}_1(\mathcal L^2[0,1])$. To check boundedness, let $F = \sum_j \sigma_j g_j \otimes h_j$ be the SVD of $F$. Subsequently we have $k^\delta(t,s) = \sum_j \sigma_j g_j(t) h_j^\delta(s)$, where the equality is understood in the $\mathcal{L}^2$-sense, and
$h_j^\delta(s) = h_j(s +\delta)$ for $s \leq 1-\delta$, and $h_j^\delta(s) = 0$ otherwise. Then
\begin{align*}
\verti{\mathrm{S}^\delta(F)}_1 &= \verti{\sum_{j=1}^\infty \sigma_j \mathrm{S}^\delta(g_j \otimes h_j)}_1 \leq \sum_{j=1}^\infty \sigma_j \verti{g_j \otimes h_j^\delta}_1 \\
&= \sum_{j=1}^\infty \sigma_j \|g_j\| \|h_j^\delta\| \leq \sum_{j=1}^\infty \sigma_j = \verti{F}_1
\end{align*}
where we used the triangle inequality in the first inequality and the fact that $\|h_j^\delta\| \leq \| h_j \|$ in the second inequality.

\begin{samepage}
\begin{definition}\label{def:T}
\begin{enumerate}
\item For $F \in \mathcal{S}_1(\mathcal L^2[0,1])$, we define $\T{}^\delta(F) = \Tr{}(S^\delta F)$.
\item For $F \in \mathcal{S}_1(\mathcal L^2[0,1]^2)$ we define $\T{}^\delta(F) = \Tr{}\big[(S^\delta \otimes S^\delta) F\big]$.
\end{enumerate}
\end{definition}
\end{samepage}

$\T{}^\delta$ is clearly well defined bounded linear functional on $\mathcal{S}_1(\mathcal L^2[0,1])$ and $\mathcal{S}_1(\mathcal L^2[0,1]^2)$, respectively. 

\begin{definition}\label{def:partial_T}
For $F \in \mathcal{S}_1(\mathcal L^2[0,1]^2)$ separable, i.e. of the form $F = A_1 \ct A_2$, we define $\T{1}^\delta(F) = \T{}^\delta(A_2) A_1$.
\end{definition}

The proofs of the following two propositions borrow ideas from \cite{aston2017}.

\begin{proposition}
Let $\delta \geq 0$, then $\T{1}^\delta: \mathcal{S}_1(\mathcal L^2[0,1]^2) \to \mathcal{S}_1(\mathcal L^2[0,1])$ is well defined, linear, and bounded. Moreover, for $F \in \mathcal{S}_1(\mathcal L^2[0,1]^2)$ we have
\begin{equation}\label{eq:trace_characterization}
\Tr{}(G \T{1}^\delta(F)) = \Tr{}([S^\delta \ct G] F), \quad \forall G \in \mathcal{S}_1(\mathcal L^2[0,1]).
\end{equation}
\end{proposition}

\begin{proof}
Let $F = \sum_{r=1}^R A_r \ct B_r$. Then for any $G \in \mathcal{S}_1(\mathcal L^2[0,1])$ we have
\begin{equation}\label{eq:trace_characterization_proof}
\begin{split}
\Tr{}(G \; \T{1}^\delta (F)) &= \sum_{r=1}^R \Tr{}(S^\delta B_r) \Tr{}(G A_r) = \sum_{r=1}^R \Tr{}\big (G A_r) \ct (S^\delta B_r) \big] \\
&= \sum_{r=1}^R \Tr{}\big[ (G \ct S^\delta) (A_r \ct B_r) \big] = \Tr{}\big[(G \ct S^\delta) F \big]
\end{split}
\end{equation}

By Lemma 1.6 of the supplementary material of \citep{aston2017}, the space
\[
\mathcal{X} := \Big\{ \sum_{r=1}^R A_r \ct B_r \;\Big|\; A_r,B_r \in \mathcal{S}_1(\mathcal L^2[0,1]), r \in \N \Big\}
\]
is dense in $ \mathcal{S}_1(\mathcal L^2[0,1]^2)$. Using the following characterization of the trace norm,
\[
\verti{F}_1 = \sup_{\verti{G}_\infty=1} | \Tr{}(G F) | ,
\]
we obtain from \eqref{eq:trace_characterization_proof} that
\begin{equation}\label{eq:shifted_trace_bound}
\begin{split}
\verti{T_{1}^\delta(F)}_1 &= \sup_{\verti{G}_\infty=1} | \Tr{}(G T_{1}^\delta(F)) | = \sup_{\verti{G}_\infty=1} |\Tr{}\big[ (G \ct S^\delta) F \big]| \\
&\leq \sup_{\verti{U}_\infty=1} | \Tr{}(U F)| = \verti{F}_1 .
\end{split}
\end{equation}

Hence $\T{1}$ can be extended continuously to $ \mathcal{S}_1(\mathcal L^2([0,1]^2))$. Equation \eqref{eq:trace_characterization} now follows from \eqref{eq:trace_characterization_proof} also by continuity.
\end{proof}

The following proposition states that the functional specified in Definition \ref{def:T} and the operator specified in Definition \ref{def:partial_T} correspond under the continuity assumption to the shifted trace and the shifted partial trace, respectively.

\begin{proposition}\label{prop:equality_of_definitions}
Let $A \in \mathcal{S}_1(\mathcal L^2[0,1])$ and $F \in \mathcal{S}_1(\mathcal L^2([0,1]^2))$ have continuous kernels $a=a(t,s)$ and $k=k(t,s,t',s')$. Then $\T{}^\delta(A) = \Tr{}^\delta(A)$ and $\T{1}^\delta(F) = \Tr{1}^\delta(F)$.
\end{proposition}

\begin{proof}
We begin by showing the assertion for the shifted trace. We define the continuous version of the shifting operator $\mathrm{S}^\delta$, denoted as $\mathrm{S}^\delta_\tau$. It is defined by Definition \ref{def:shifting_operator} with $\mathrm{S}^\delta$ replaced by $\mathrm{S}^\delta_\tau$ and $k^\delta$ replaced by
\[
k^\delta_\tau(t,s) = \begin{cases}
k(t,s+\delta) , \hfill s < 1-\delta-\tau , \\
(s + \delta + \tau)k(t,1-\delta-\tau) + (s + \delta - \tau)k(t,1-\delta+\tau) , \quad |s-(1-\delta)|\leq \tau, \\
0 , \hfill \text{otherwise} .
\end{cases}
\]
Then by continuity, $\Tr{}(S^\delta_\tau F) \stackrel{\tau \to 0_+}{\longrightarrow} \T{}^\delta(F)$ and at the same time $\Tr{}(S^\delta_\tau F) \stackrel{\tau \to 0_+}{\longrightarrow} \Tr{}^\delta(F)$, implying the equality of the limits.

We now proceed to the shifted partial trace. Note that it follows from the Stone-Weierstrass approximation theorem that for any $\epsilon > 0$ there exist $R \in \N$ and a set of continuous univariate functions on [0,1] $\{ u_r, v_r, x_r, y_r \}_{r=1}^R$ such that $\| k - k_R \|_\infty < \epsilon$ for
\[
k_R(t,s,t',s') = \sum_{r=1}^R u_r(t) x_r(s) v_r(t') y_r(s') .
\]
Grouping together $a_r(t,t') := u_r(t) v_r(t')$ and $b_r(t,t') := x_r(t) y_r(t')$, it follows that for any $\epsilon > 0$ there exists a finite rank operator $F_R = \sum_{r=1}^R A_r \ct B_r$ such that $A_r$ and $B_r$ are rank one operators with continuous kernels ($a_r$ and $b_r$, respectively) and $\verti{k - k_R}_\infty < \epsilon$.

Let us fix $\epsilon > 0$. Then by the triangle inequality we have
\[
\begin{split}
\verti{ \T{1}^\delta(F) - \Tr{1}^\delta(F) }_1 \leq
\verti{ \T{1}^\delta(F) - \T{1}^\delta(F_R) }_1
&+ \verti{ \T{1}^\delta(F_R) - \Tr{1}^\delta(F_R)}_1 \\
&+ \verti{\Tr{1}^\delta(F_R) - \Tr{1}^\delta(F)}_1
\end{split}
\]
The middle term is zero, which follows from linearity of the operators and the first half of this proof. The first and the third terms can be both bounded by 
\[
\begin{split}
\verti{F - F_R}_1 &\leq \verti{F - F_R}_2 = \left( \int_0^1 \int_0^1 \big[k(t,s,t,s) - k_R(t,s,t,s)\big]^2 d t d s \right)^{1/2} \\
&\leq \| k - k_R \|_\infty \leq \epsilon.
\end{split}
\]
Altogether, we have that $\verti{ \T{1}^\delta(F) - \Tr{1}^\delta(F) }_1 < 2\epsilon$. Since $\epsilon$ was arbitrarily small, the proof is complete.

\end{proof}

The development of shifted partial tracing with respect to to the second argument can be done similarly. 
Proposition \ref{prop:SPT} now follows directly from Proposition \ref{prop:equality_of_definitions}. Also, Proposition \ref{prop:PT_properties} holds with the general definitions of the shifted (partial) traces, which can be simply checked using the definitions. It thus remains to show validity of Lemma \ref{lem:banded_SPT}.

\begin{proof}[Proof of Lemma \ref{lem:banded_SPT}]
Firstly, it holds for any operators $A$ and $B$ that
\begin{equation}\label{eq:distributivity}
(A \ct B) = (A \ct Id)(Id \ct B) ,
\end{equation}
which can be verified on the rank one elements:
\[
(A \ct Id)(Id \ct B)(x \otimes y) = (A \ct Id)(x \otimes B y) = A x \otimes B y = (A \ct B) (x \otimes y).
\]

Secondly, we know from equation \eqref{eq:shifted_trace_bound} that
\begin{equation}\label{eq:shifted_trace_bound_alone}
\verti{\Tr{1}^\delta(F)}_1 = \sup_{\verti{G}_\infty=1} \big|\Tr{}\big[ (G \ct S^\delta) F \big]\big|.
\end{equation}

Now, it is enough to show that $\Tr{}\big[ (G \ct S^\delta) B \big]=0$ for any $G \in \mathcal{S}_\infty$. Let $B=\sum_r \sigma_r \widetilde{U}_r \ct V_r$ be the SVD of $B$. Note that $V_r$ in particular is banded by $\delta^\star$. Denoting $U_r = \sigma_r U_r$, we have
\[
\begin{split}
\Tr{}\big[ (G \ct S^\delta) B \big] = \Tr{}\big[ (Id \ct S^\delta) B (G \ct Id) \big]
\end{split}
\]
where we used \eqref{eq:distributivity} and cyclicity of trace. Now, $B (G \ct Id) = \sum_r (G U_r) \ct V_r$, which is still banded in the dimensions corresponding to $V_r$'s. Therefore $(Id \ct S^\delta) B (G \ct Id)$ has a kernel which is 0 along the diagonal and hence its trace is 0 by the limiting argument of \cite{gohberg1978}.
\end{proof}

We have just shown that the conclusions of the paper stand still even without the assumption of continuity.

Shifted partial tracing could still have been defined in slightly greater generality. However, the definition requires the notion of a ``shift'' and hence it requires an explicit set to act on. We could instead of $\mathcal{L}^2([0,1]^2)$ take $\mathcal{L}^2(\Omega)$ with $(\Omega, \mathcal{A}, \mu)$ a measure space with $\Omega$ a linearly ordered metric space and $\mu$ a finite measure. The specific choice of $\Omega = \{1, \ldots, K_1 \} \times \{1, \ldots, K_2 \}$ and $\mu$ being the counting measure would then lead to formula \eqref{eq:discrete_PT_def}. We have not gone down this path since this formalism would not be particularly useful in practice anyway. Note, however, that Definition \ref{def:SPT} and formula \eqref{eq:discrete_PT_def} are compatible in this way, with the difference between them stemming from the change of measure, as depicted in the following lemma.

\begin{lemma}\label{lem:continuous_discrete_traces}
Let $\mathbf M \in \R^{K_1 \times K_2 \times K_1 \times K_2}$. Let $F \in \mathcal{S}_2(L^2[0,1]^2)$ be the pointwise continuation of $\mathbf M$, i.e. the kernel $k$ of $F$ is given by
\[
k(t,s,t',s') = \sum_{i=1}^{K_1} \sum_{j=1}^{K_2} \sum_{k=1}^{K_1} \sum_{l=1}^{K_2} \mathbf M[i,j,k,l] \mathds{1}_{[(t,s) \in I_{i,j}^K]} \mathds{1}_{[(t',s') \in I_{k,l}^K]} \,,
\]
where $I_{j,k} = \left[\frac{i-1}{K_1}, \frac{i}{K_1}\right) \times \left[\frac{i-1}{K_2}, \frac{i}{K_2}\right)$. Then
\begin{enumerate}
\item For $\delta \in [0,1)$ such that $\delta K_2 \in \N_0$ we have $\verti{\Tr{1}^\delta(F) }_2 = K_1^{-1} K_2^{-1} \left\| \Tr{1}^\delta(\mathbf M) \right\|_F$.
\item For $\delta \in [0,1)$ such that $\delta K_1 \in \N_0$ we have $\verti{\Tr{2}^\delta(F) }_2 = K_1^{-1} K_2^{-1} \left\| \Tr{2}^\delta(\mathbf M) \right\|_F$.
\item For $\delta \in [0,1)$ such that $\delta K_1 \in \N_0$ and $\delta K_2 \in \N_0$ we have $\Tr{}^\delta(F) = K_1^{-1} K_2^{-1} \Tr{}^\delta(\mathbf M)$.
\end{enumerate}
\end{lemma}

\begin{proof}
We only show the first part, since the other two parts are similar.

Let $g_{i,j}(t,s) = \sqrt{K_1 K_2} \mathds{1}_{[(t,s) \in I_{i,j}^K]}$. Since
\[
k(t,s,t',s') = K_1^{-1} K_2^{-1} \sum_{i=1}^{K_1} \sum_{j=1}^{K_2} \sum_{k=1}^{K_1} \sum_{l=1}^{K_2} \mathbf M[i,j,k,l] g_{i,j}(t,s)  g_{k,l}(t',s')\,,
\]
we can express $F$ as
\[
F = K_1^{-1} K_2^{-1} \sum_{i=1}^{K_1} \sum_{j=1}^{K_2} \sum_{k=1}^{K_1} \sum_{l=1}^{K_2} \mathbf M[i,j,k,l] g_{i,j} \otimes g_{k,l} \,.
\]
It now follows from linearity of shifted partial tracing that
\begin{equation}\label{eq:blb}
\Tr{1}^\delta(F) = K_1^{-1} K_2^{-1} \sum_{i=1}^{K_1} \sum_{j=1}^{K_2} \sum_{k=1}^{K_1} \sum_{l=1}^{K_2} \mathbf M[i,j,k,l] \Tr{1}^\delta( g_{i,j} \otimes g_{k,l} ) \,.
\end{equation}

Since $I_{i,j}^K$ is a cartesian product of two intervals, we can write $I_{i,j}^K = I_i^K \times I_j^K$. Then $g_{i,j} = g^{(1)}_i \otimes g^{(2)}_j$ with $g^{(1)}_i(t) = \sqrt{K_1} \mathds{1}_{[t \in I_{i}^K]}$ and $g^{(2)}_j(s) = \sqrt{K_1} \mathds{1}_{[s \in I_{j}^K]}$. Furthermore,
\[
g_{i,j} g_{k,l} = g^{(1)}_i \otimes g^{(2)}_i \otimes g^{(1)}_k \otimes g^{(2)}_l = (g^{(1)}_i \otimes g^{(1)}_k) \ct (g^{(2)}_j \otimes g^{(2)}_l) \,
\]
and hence by Definition \ref{def:partial_T} we have $\Tr{1}^\delta( g_{i,j} \otimes g_{k,l} ) = \Tr{}^\delta(g^{(2)}_j \otimes g^{(2)}_l) g^{(1)}_i \otimes g^{(1)}_k $. Note that $\Tr{}^\delta(g^{(1)}_j \otimes g^{(1)}_l) = \mathds{1}_{[j=l+\delta K_1]}$, hence from \eqref{eq:blb} we have
\begin{equation}\label{eq:blbb}
\Tr{1}^\delta(F) = K_1^{-1} K_2^{-1} \sum_{i=1}^{K_1} \sum_{k=1}^{K_1} \left( \sum_{j=1}^{(1-\delta)K_2} \mathbf M[i,j,k,j + \delta K] \right) g^{(2)}_i \otimes g^{(2)}_k \,.
\end{equation}
Thus it is $\verti{\Tr{1}^\delta(F)}_2 = K_1^{-1} K_2^{-1} \left[ \sum_{i=1}^{K_1} \sum_{k=1}^{K_1} \left( \sum_{j=1}^{(1-\delta)K_2} \mathbf M[i,j,k,l] \right)^2\right]^{1/2}$,
while in the discrete case we have
$\left\| \Tr{1}^\delta(\mathbf M) \right\|_F = \left[ \sum_{i=1}^{K_1} \sum_{k=1}^{K_1} \left( \sum_{j=1}^{(1-\delta)K_2} \mathbf M[i,j,k,l] \right)^2\right]^{1/2}$ from by formula \eqref{eq:discrete_PT_def}.
\end{proof}

Note that we actually proved something more general. We can write from \eqref{eq:blbb} that the kernel of $\Tr{1}^\delta(F)$ is
\[
k_1(t,t') = \sum_{i=1}^{K_1} \sum_{k=1}^{K_1} \left( \frac{1}{K_2} \sum_{j=1}^{(1-\delta)K_2} \mathbf M[i,j,k,j+\delta K] \right) \mathds{1}_{t \in I_i} \otimes \mathds{1}_{t \in I_k} \,,
\]
where the term inside the parentheses is almost the $(i,k)$-th element of discrete partial tracing, but instead of summing in the discrete case we have to average in the continuous case. This corresponds to the difference between the Lebesque measure on piecewise constant function on $[0,1]$ with at most $K$ jumps and the counting measure on the set $\{ 1, \ldots, K \}$.

\subsection{Toeplitz Averaging, Circulant Matrices and Fourier Transform}\label{app:C}

We begin this section by showing that a self-adjoint stationary integral operator on $\mathcal L^2[0,1]$ has the Fourier basis as its eigenbasis. We work with $\mathcal L^2[0,1]$ for simplicity, the argument translates easily to higher dimensions.

Let $F$ be a stationary integral operator on $\mathcal L^2[0,1]$ with kernel $k=k(t,s)$, i.e. $k(t,s) = h(t-s)$, $t,s \in [0,1]$, for a symmetric function $h: [-1,1] \to \R$. We expand $h$ into its Fourier series as $h(x) = \sum_{j \in \mathds{Z}} \phi_j e^{-2 \pi i j x}$. 
Thus we have
\[
k(t,s) = \sum_{j \in \mathds{Z}} \phi_j e^{-2 \pi i j t} e^{2 \pi i j s} .
\]
To see that the previous expansion is in fact an eigen-decomposition, note that we have
\[
\int_0^1 k(t,s) e^{-2 \pi i l s} d s = \sum_{j \in \mathds{Z}} \phi_j e^{-2 \pi i j t} \int_0^1  e^{-2 \pi i (l-j) s} d s = \theta_l e^{-2 \pi i l t},
\]
for \mbox{$l = 0,1, \ldots$} and similarly for $-l \in \N$ due to self-adjointness.

The previous justifies the definition of the Toeplitz averaging operator in the continuous case. In the discrete case, there is also a relation between stationary operators and the Fourier transform. It is a well known fact in the time series literature that the periodogram is both the real part of the DFT of the autocovariance function, i.e. of the first row of the (Toeplitz) covariance matrix, and the squared DFT of the data \cite{brockwell1991}. This is a consequence of the Wiener-Khinchin theorem, and it allows one to compute the autocovariance function fast using the FFT. It is straightforward to show that the previous generalizes to the case of 2D data, which is done next for completeness.

Note that in the case of a 1D time series, the 2D covariance operator is captured by the 1D autocovariance. In the case of a 2D datum $\mathbf X \in \R^{K_1 \times K_2}$, the 4D covariance operator $\mathbf C \in \R^{K_1 \times K_2 \times K_1 \times K_2}$ will be captured by the 2D \emph{symbol} $\boldsymbol \Gamma \in \R^{K_1 \times K_2}$. The latter is defined as
\[
\boldsymbol \Gamma[h_1,h_2] = \frac{1}{K_1 K_2} \sum_{k_1=1}^{K_1} \sum_{k_2=1}^{K_2} \mathbf{X}[k_1, k_2] \mathbf{X}^*[k_1 + h_1, k_2 +h_2] .
\]
The DFT of $\mathbf{X}$, denoted as $\mathbf{Z}$, is defined by
\[
\mathbf{X}[k_1,k_2] = \frac{1}{\sqrt{K_1 K_2}} \sum_{a=1}^{K_1} \sum_{b=1}^{K_2} \mathbf{Z}[a,b] e^{-i \omega k_1 a} e^{-i \theta k_2 b},
\]
where $\omega = 2 \pi/K_1$ and $\theta = 2 \pi/K_2$. Thus plugging the DFT of $\mathbf X$ into $\boldsymbol \Gamma$, we obtain
\[
\begin{split}
\boldsymbol \Gamma[h_1,h_2] &= \frac{1}{(K_1 K_2)^2} \sum_{k_1=1}^{K_1} \sum_{k_2=1}^{K_2}
\sum_{a=1}^{K_1} \sum_{b=1}^{K_2} \sum_{t=1}^{K_1} \sum_{s=1}^{K_2}
 \mathbf{Z}[a,b] \mathbf{Z}^*[t,s] e^{-i \omega k_1 a} e^{-i \theta k_2 b} e^{i \omega (k_1+h_1) t} e^{i \theta (k_2+h_2) s}\\
&= \frac{1}{(K_1 K_2)^2} \sum_{a=1}^{K_1} \sum_{b=1}^{K_2} \sum_{t=1}^{K_1} \sum_{s=1}^{K_2}
\mathbf{Z}[a,b] \mathbf{Z}^*[t,s] e^{i \omega h_1 t} e^{i \theta h_2 s}
\underbrace{\Big[ \sum_{k_1=1}^{K_1} e^{-i \omega k_1 (a-t)} \Big]}_{= K_1 \mathds{1}_{[a=t]}}
\underbrace{\Big[ \sum_{k_2=1}^{K_2} e^{-i \theta k_2 (b-s)} \Big]}_{= K_2 \mathds{1}_{[b=s]}} \\
&= \frac{1}{K_1 K_2} \sum_{a=1}^{K_1} \sum_{b=1}^{K_2} \mathbf{Z}[a,b] \mathbf{Z}^*[a,b] e^{i \omega h_1 a} e^{i \theta h_2 b} = \mathbf{W}[h_1,h_2] ,
\end{split}
\]
where $\mathbf{W}$ is the inverse DFT applied to the DFT of $\mathbf{X}$ squared element-wise. Symbolically $\mathbf{W} = \mathrm{ifft}\big(|\mathrm{fft}(\mathbf{X})|^2\big)$, where $| \cdot |^2$ is applied element-wise. This shows that $\topavg(\mathbf X \otimes \mathbf X)$ can be calculated fast using the FFT.

The operator $\topavg(\cdot)$ is linear. Hence, with regards to the tractability of the estimator \eqref{eq:estimator_B}, it remains now to show that $\topavg(\mathbf F \ct \mathbf G)$ can be evaluated efficiently for $\mathbf F \in \R^{K_1 \times K_1}$ and $\mathbf G \in \R^{K_2 \times K_2}$. This is straightforward. For example, $\topavg(\mathbf F \ct \mathbf G)[1,1]$ is the average of the diagonal elements of $\mathbf F \ct \mathbf G$, which can be calculated as a product of the average diagonal element of $\mathbf F$ and average diagonal element of $\mathbf G$. Also, only the symbol of $\topavg(\mathbf F \ct \mathbf G) \in \R^{K_1 \times K_2 \times K_1 \times K_2}$ as an element of $\R^{K_1 \times K_2}$. Altogether, the memory complexity and the number of operations needed for computing the estimator \eqref{eq:estimator_B} in the case of $K_1 = K_2 = K$ is $\O(K^2)$ and $\O(N K^2 \log K + K^3)$, respectively.

The remainder of the section is devoted to showing that a matrix-vector product involving a Toeplitz matrix can be calculated efficiently. For this, we need \textit{circulant matrices} \cite{davis2013}. Recall that matrix $\mathbf Q \in \R^{m \times n}$ is circulant if $\mathbf Q=(q_{ij}) = (q_{j-i+1 \mod n})$, where $\mathbf q \in \R^n$ is the \textit{symbol} of the matrix, i.e. $\mathbf q^\top$ is the first row of $\mathbf Q$. Every circulant matrix is obviously a Toeplitz matrix. Contrarily, every Toeplitz matrix can be embedded into a larger circulant matrix (note that this embedding is not unique). For example, a symmetric Toeplitz matrix $\mathbf{T} \in \R^{n \times n}$ with symbol $\mathbf{t} \in \R^n$ can be embedded into a symmetric circulant matrix $\mathbf{Q} \in \R^{(2n -1) \times (2n -1)}$ with symbol $\mathbf{q} = (t_{1}, \ldots, t_{n},t_{n},\ldots,t_{2})$. In the case of $n=3$, we have
\[
\mathbf Q = \left( \begin{array}{ccc|cc}
t_{1} & t_{2} & t_{3} & t_{3}& t_{2} \\
t_{2} & t_{1} & t_{2} & t_{3}& t_{3} \\
t_{3} & t_{2} & t_{1} & t_{2}& t_{3} \\
\hline
t_{3} & t_{3} & t_{2} & t_{1}& t_{2} \\
t_{2} & t_{3} & t_{3} & t_{2}& t_{1} \\
\end{array}\right)
= \left( \begin{array}{c|c}
\mathbf{T} & \;\cdot \\
\hline
\cdot & \;\cdot
\end{array}\right) \,.
\]
This embedding is useful due to the well known fact that circulant matrices are diagonalizable by the DFT, hence $\mathbf Q = \mathbf E^* \diag(\boldsymbol \lambda) \mathbf E$, where $\mathbf E$ is matrix with the discrete Fourier basis in its columns, i.e. \mbox{$\mathbf E[j,k] = \frac{1}{\sqrt{n}} e^{2 \pi i j k / n}$}. Hence the eigenvalues of $\mathbf{Q}$ can be calculated as the FFT of the symbol $\mathbf{q}$, namely $\boldsymbol \lambda = \mathrm{fft}(\mathbf{q})$. This implies that a matrix-vector product involving a circulant matrix can be calculated in $\O(n \log n)$ as
\begin{equation}\label{eq:fft_multiplication}
\mathbf Q \mathbf v = \mathbf E^* \diag(\boldsymbol \lambda) \mathbf E \mathbf v = \mathrm{ifft}\Big( \boldsymbol \lambda \odot \mathbf E \mathbf v \Big) = \mathrm{ifft}\Big( \mathrm{fft}(\mathbf{q}) \odot \mathrm{fft} (\mathbf v) \Big),
\end{equation}
where $\mathrm{ifft}(\cdot)$ is the inverse FFT and $\odot$ denotes the Hadamard (element-wise) product.
Thus using the circulant embedding, the product of a Toeplitz matrix $\mathbf{T} \in \R^{n \times n}$ with a vector $\mathbf v \in \R^n$ can also be calculated in $\O(n \log n)$:
\begin{equation} \label{eq:fast_multiplication}
\mathbf{Q} \begin{pmatrix}
\mathbf{v} \\
\mathbf 0
\end{pmatrix}
 = \left( \begin{array}{c|c}
\mathbf{T} & \cdot \\
\hline
\cdot & \cdot
\end{array}\right)
\left( \begin{array}{c}
\mathbf{v} \\
\hline
\mathbf 0
\end{array}\right) = \left( \begin{array}{c}
\mathbf T \mathbf{v} \\
\hline
\cdot
\end{array}\right).
\end{equation}

The previous machinery can be naturally extended to higher dimensions, using two-level Toeplitz (resp. circulant) matrices, i.e. Toeplitz (resp. circulant) block matrices with Toeplitz (resp. circulant) blocks. For example, the tensor-matrix product $\widehat{\mathbf B} \mathbf X$ can be written as $\widehat{\mathbf B}_{\mathrm{mat}} \vec{ \mathbf X}$, where $\widehat{\mathbf B}_{\mathrm{mat}}$ is the matricization of $\widehat{\mathbf B}$, which is a two-level Toeplitz matrix. This product can be calculated by embedding $\widehat{\mathbf B}_{\mathrm{mat}}$ into a two-level circulant matrix $\mathbf Q_{\mathrm{mat}}$ and using analogs of \eqref{eq:fft_multiplication} and \eqref{eq:fast_multiplication}. Notably, equation \eqref{eq:fft_multiplication} becomes
\[
\mathbf Q_{\mathrm{mat}} \mathbf X = \mathrm{i2Dfft}\Big( \mathrm{2Dfft}(\boldsymbol \Gamma) \odot \mathrm{2Dfft} (\mathbf X) \Big),
\]
where $\mathrm{2Dfft}$ is the 2D DFT, $\mathrm{i2Dfft}$ is its inverse counterpart, and $\boldsymbol \Gamma \in \R^{(2K_1-1) \times (2K_2-1)}$ is the symbol of $\mathbf Q$, which is the tensorization of $\mathbf Q_{\mathrm{mat}}$. Note that the $K_1 \times K_2$ top-left sub-matrix of $\boldsymbol \Gamma$ is the symbol of $\widehat{\mathbf B}$.

\subsection{Fast Norm Calculation}\label{app:D}

The bandwidth selection strategy discussed in Section 3.4 of the main paper requires calculations of norms of separable-plus-stationary covariances. More generally, norms of the following form need to be calculated:
\begin{equation}\label{eq:norm}
\verti{ \mathbf A_1 \ct \mathbf A_2 + \mathbf B - \mathbf C_1 \ct \mathbf C_2 - \mathbf D}_2 .
\end{equation}
Assuming we work on a $K \times K$ grid, we have $\mathbf A_1, \mathbf A_2,\mathbf C_1,\mathbf C_2 \in \R^{K \times K}$ and $\mathbf B,\mathbf D \in \R^{K \times K \times K \times K}$ (being stationary) in the previous formula. A naive calculation of the norm then requires $\O(K^4)$ flops. In this section, we show that the special structure can be used to reduce the complexity to $\O(K^3)$.

One only needs to realize, that both separable tensors and stationary tensors of size $K \times K \times K \times K$ can be re-arranged into a matrix of size $K^2 \times K^2$ with $K \times K$ blocks such that every block is a rank-one matrix. For example, the diagonal entries of $\mathbf A_1 \ct \mathbf A_2 + \mathbf B - \mathbf C_1 \ct \mathbf C_2 - \mathbf D$ are also entries of
\begin{equation}\label{eq:diagonal}
\diag(\mathbf A_1) \diag(\mathbf A_2)^\top + \mathbf B[1,1] \cdot \mathbf{1} \mathbf{1}^\top - \diag(\mathbf C_1) \diag(\mathbf C_2)^\top - \mathbf D[1,1] \cdot \mathbf{1} \mathbf{1}^\top ,
\end{equation}
where $\mathbf{1} \in \R^K$ is vector of ones. The matrix \eqref{eq:diagonal} is of size $K \times K$, and it is rank-3. The squared Frobenius norm of this rank-3 matrix can be calculated using Gram-Schmidt orthogonalization in only $\O(K)$ flops. Summing together the total of $K^2$ of these blocks, we can calculate the square of \eqref{eq:norm} in $\O(K^3)$ flops. Therefore the norm calculation is within our computational limits.

\subsection{Ensuring Symmetry and Positive Semi-definiteness}\label{app:E}

Among other things, the assumption of separability induces extra symmetry. Every covariance $C$ is symmetric in the sense that $c(t,s,t',s')=c(t',s',t,s)$ for any $t,s,t',s' \in [0,1]$. If $c(t,s,t',s')=c_1(t,t')c_2(s,s')$, it is easy to see that it must be
\[
c(t,s,t',s')= c(t',s,t,s') = c(t,s',t',s) =c(t',s',t,s), \qquad t,s,t',s' \in [0,1].
\]
When we wish to ensure that results of shifted partial tracing are symmetric, we have several options:
\begin{enumerate}
\item symmetrizing the results of shifted partial tracing, for example setting
\[
\widehat{A}_1 = \frac{1}{2} \big[\Tr{1}^\delta(\widehat{ C}_N) + (\Tr{1}^\delta(\widehat{ C}_N))^*\big],
\]
\item inducing the extra symmetry of the covariance, for example $\widehat{A}_1 = \Tr{1}^\delta(\widetilde{C}_N)$ with $\widetilde{c}_N(t,s,t',s') = \frac{1}{2}[ \widehat{c}_N(t,s,t',s') + \widehat{c}_N(t,s',t',s)]$,
\item defining shifted partial tracing in a symmetric manner by replacing \eqref{eq:shifted_kernel} with
\[
k^\delta(t,s) = \begin{cases}
\frac{1}{2} \big[k(t,s+\delta) + k(t+\delta,s)\big], \quad s < 1-\delta , \\
0 , \qquad \qquad \qquad \qquad \qquad \qquad \text{otherwise},
\end{cases}
\]
and developing shifted partial tracing from there, which would ultimately lead to the first formula in \eqref{eq:shifted_partial_trace_integral} replaced by
\[
k_1(t,t') = \int_0^{1-\delta} \frac{1}{2}\big[ k(t,s,t',s+\delta) + k(t,s+\delta,t',s) \big] d s
\]
\end{enumerate}
These options are equivalent due to symmetry of $\widehat{C}_N$ and the fact that adjoining commutes with any linear operator, hence also with shifted partial tracing.

Developing our theory as suggested by option 3 above is straightforward, merely lengthening all the calculations. In practice, option 1 is preferable for computational reasons.

Shifted partial tracing (even the symmetrized one) applied to a positive semi-definite (PSD) operator does not necessarily lead to a PSD operator. In the case of the original operator $C$ being separable, it is easy to see that either $\Tr{1}^\delta{C} \succeq 0$ or $-\Tr{1}^\delta{C} \succeq 0$, so a potential sign flip is enough to ensure PSD. However, $\widehat{C}_N$ is usually not separable even when the original covariance $C$ is. Nonetheless, $\widehat{C}_N$ is still a natural estimator of $C$ and, from our experience, the potential sign flip usually solves the problem. If need be, the eigendecomposition can be calculated and negative eigenvalues set to zero. In the discrete case, this requires $\O(K^3)$ operations and thus it is computationally feasible.

Let us now focus on Toeplitz averaging. Since the argument in \eqref{eq:estimator_B} is symmetric, and since the symmetry is obviously preserved, we only have to discuss positive semi-definiteness. Unfortunately, the argument $\widehat{C}_N - \widehat{A}_1 \ct \widehat{A}_2$ is not necessarily PSD and thus $\widehat{B}$ may also not be. However, using Bochner's theorem the same way as in \cite{hall1994}, the positive semi-definite projection of $\widehat{B}$ can be found. In the discrete case, the matricization of $\widehat{\mathbf B}$ can be embedded into a two-level circulant matrix with symbol $\boldsymbol \Gamma$ (see Appendix~\ref{app:C}). Subsequently, the DFT is applied to $\boldsymbol \Gamma$ to obtain the eigenvalues, negative eigenvalues are set to zero, and the result is transformed back via the inverse DFT, giving the positive part of $\widehat{\mathbf B}$. This procedure requires $\O(K^2 \log K)$ operations when the FFT is used.

If we do not assume that $B$ is stationary and use estimator \eqref{eq:estimator_B_non_stationary}, we can make it positive directly via eigendecomposition.

\subsection{Sub-Problems of the Inverse Algorithm}\label{app:F}

Here we describe how to solve the two linear sub-problems appearing in the ADI scheme. The first sub-problem has the form
\[
(\mathbf A + \rho \mathbf I) \mathbf x = \mathbf y\,,
\]
where $\mathbf A = \mathbf A_2 \otimes_K \mathbf A_1$. Even though $\mathbf A$ has the Kronecker structure, $\mathbf A + \rho \mathbf I$ generally does not. Still, the system can be rewritten in the matrix form as
\begin{equation}\label{eq:Stein}
\mathbf A_2 \mathbf X \mathbf A_1 + \rho \mathbf X = \mathbf Y \,,
\end{equation}
which is the well-known discrete Stein's equation. Although there exist specialized solvers for this particular equation (see \cite{simoncini2016} for an overview), they are not suitable here due to the fact that $\rho$ is usually very small. Instead of using these specialized solvers, we show that, in our case of $\mathbf A_1$ and $\mathbf A_2$ being positive semi-definite, equation \eqref{eq:Stein} has in fact an analytic solution computable in $\mathcal{O}(K^3)$ operations.

We compute the eigendecompositions $\mathbf A_1 = \mathbf U \diag(\boldsymbol \phi) \t{\mathbf U}$ and $\mathbf A_2 = \mathbf V \diag(\boldsymbol \psi) \t{\mathbf{V}}$. Then, using the knowledge of the spectra of Kronecker products (cf. \cite{vanloan1993}), system \eqref{eq:Stein} can be vectorized as
\[
(\mathbf U \otimes_K \mathbf V) \diag\Big[\vec{\boldsymbol \phi \t{\boldsymbol \psi}} \Big] \t{(\mathbf U \otimes_K \mathbf V)} \mathbf x + \rho \mathbf x = \mathbf y \,,
\]
where $\boldsymbol \phi \t{\boldsymbol \psi}$ is a matrix corresponding to the vector of eigenvalues of $\mathbf A$, which is subsequently rearranged into a large diagonal matrix by the $\diag[\cdot]$ operator. Secondly, utilizing the fact that $\mathbf U \otimes_K \mathbf V$ is an orthonormal basis, we can write
\[
(\mathbf U \otimes_K \mathbf V) \diag\Big[\vec{\mathbf H} \Big] \t{(\mathbf U \otimes_K \mathbf V)} \mathbf x = \mathbf y \,, 
\]
where we denote $\mathbf H := \boldsymbol \phi \t{\boldsymbol \psi} + \rho \mathbf{1}$, with $\mathbf{1}$ being a matrix with all entries equal to 1. Finally, one can express the solution as
\[
x = (\mathbf U \otimes_K \mathbf V) \diag\Big[\vec{\mathbf H} \Big]^{-1} \t{(\mathbf U \otimes_K \mathbf V)} \mathbf y \,.
\]
Using property \eqref{eq:kron}, this can be matricized back to
\[
\mathbf X =  \mathbf V(\mathbf G \odot \t{\mathbf U} \mathbf Y \mathbf V) \t{\mathbf U}\,,
\]
where $\mathbf G$ is the element-wise inverse of $\mathbf H$ and $\odot$ denotes the Hadamard (element-wise) product. Hence we found a solution, which is computable in $\mathcal{O}(K^3)$ operations.

The second sub-problem has the following form:
\[
(\mathbf B + \rho \mathbf I) \mathbf x = \mathbf y \,,
\]
where $\mathbf{B}$ is a two-level Toeplitz matrix, and this structure is preserved when a diagonal matrix is added to $\mathbf{B}$. Hence we only need to devise a solver for $\mathbf{B} \mathbf x = \mathbf y$, where $\mathbf B$ is positive definite. Again, even though specialized solvers for this structured linear system exist, provably providing a solution in $\O(K^2 \log^2(K))$, they are not easily accessible, and they are focused on cases when $\mathbf B$ is not symmetric. The latter is likely the case because Preconditioned Conjugate Gradient (PCG) is the method of choice, when positive definiteness is granted.

We do not describe the PCG method here, as it is a classical optimization method. Notably, \citep{shewchuk1994} provides both rigorous proofs and informal geometrical arguments for the fact that CG converges faster if the eigenvalues of $B$ are clustered, which can be ensured by preconditioning. One CG step takes $\O(K^2 \log(K))$ operations, and this complexity is retained if a suitable preconditioning is used. Moreover, under mild assumptions and with a convenient preconditioner, the convergence rate of the PCG is super-linear, which means only a constant number of iterations is needed to attain a prescribed accuracy \citep{chan2007}. Even though we cannot guarantee these mild assumptions, the second choice of preconditioning described in Chapter 5 of \citep{chan2007} was shown to ensure the fixed number of iterations for problems structurally very similar to ours. Hence we use this preconditioning, and show empirically in Section \ref{sec:simulation} that convergence is truly rapid.

\subsection{Proofs of Asymptotic Results}\label{app:G}

\subsubsection{Proof of Theorem \ref{thm:asymptotics}}\label{app:G1}

We begin with an auxiliary result that will be used to prove Theorem \ref{thm:asymptotics}.

\begin{lemma}\label{lem:gaussianity}
\begin{enumerate}
\item Let $Z \in \mathcal{S}_1(\mathcal{L}^2[0,1]^2)$ be a Gaussian random element. Then $\Tr{1}^\delta(Z)$ and $\Tr{2}^\delta(Z)$ are Gaussian random elements of $\mathcal{S}_1(\mathcal{L}^2[0,1]^2)$.
\item Let $Z \in \mathcal{S}_1(\mathcal{L}^2[0,1]^2)$ be a Gaussian random element. Then $\topavg(Z)$ is a Gaussian random elements of $\mathcal{S}_1(\mathcal{L}^2[0,1]^2)$.
\item Let $Z \in \mathcal{S}_1(\mathcal{L}^2[0,1])$ be a Gaussian random element and $F \in \mathcal{S}_1(\mathcal{L}^2[0,1])$. Then $Z \ct F$ and $F \ct Z$ are Gaussian random elements in $\mathcal{S}_1(\mathcal{L}^2[0,1]^2)$.
\end{enumerate}
\end{lemma}

\begin{proof}
Firstly, note that a random element $Z \in \mathcal{S}_1(\mathcal{L}^2[0,1]^2)$ is Gaussian if, for any $G \in \mathcal{S}_\infty(\mathcal{L}^2[0,1]^2)$, $\Tr{}(GZ)$ is Gaussian \cite{bosq2012}.

Secondly, for an operator $F: \mathcal B_1 \to \mathcal B_2$, its adjoint $F^\star: \mathcal B_2^\star \to \mathcal B_1^\star$ is defined so for any $G \in \mathcal B_2^\star$ we have $F^\star G = G F$.
 
\begin{enumerate}

\item This follows immediately from the above and formula \eqref{eq:trace_characterization}.

\item For $\topavg: \mathcal{S}_1(\mathcal{L}^2[0,1]^2) \to \mathcal{S}_1(\mathcal{L}^2[0,1]^2)$, the adjoint $\topavg^\star: \mathcal{S}_\infty(\mathcal{L}^2[0,1]^2) \to \mathcal{S}_\infty(\mathcal{L}^2[0,1]^2)$ satisfies $\topavg^\star(G) = G \, \topavg\;$ for any $G \in S_\infty(\mathcal{L}^2[0,1]^2)$. Hence we have $\Tr{}(G \, \topavg(Z)) = \Tr{}(\topavg^\star(G) Z)$, where $\topavg^\star(G) \in S_\infty(\mathcal{L}^2[0,1]^2)$.

\item This is Proposition 1.2 in \citep{aston2017supplement}. A proof can be found there.
\end{enumerate}
\end{proof}

\begin{proof}[Proof of Theorem \ref{thm:asymptotics}]
Recall that in our model it holds $A_1 \ct A_2 = \frac{\Tr{1}^\delta(C) \ct \Tr{2}^\delta(C)}{\Tr{}^\delta(C)}$. Hence we have
\begin{equation}\label{eq:linearization_distr}
\begin{split}
\sqrt{N}( \widehat{A}_1 \ct \widehat{A}_2 - A_1 \ct A_2 ) &= \sqrt{N} \Bigg( \frac{\Tr{1}^\delta(\widehat{C}_N) \ct \Tr{2}^\delta(\widehat{C}_N)}{\Tr{}^\delta(\widehat{C}_N)} - \frac{\Tr{1}^\delta(C) \ct \Tr{2}^\delta(C)}{\Tr{}^\delta(C)} \\
&\qquad\qquad \pm \frac{\Tr{1}^\delta(\widehat{C}_N) \ct \Tr{2}^\delta(C)}{\Tr{}^\delta(\widehat{C}_N)} \pm \frac{\Tr{1}^\delta(C) \ct \Tr{2}^\delta(C)}{\Tr{}^\delta(\widehat{C}_N)} \Bigg) \\
&= \frac{1}{\Tr{}^\delta(\widehat{C}_N)} \Big( \Tr{1}^\delta(\widehat{C}_N) \ct \Tr{2}^\delta\big[\sqrt{N}\big(\widehat{C}_N - C \big) \big] \\
&\qquad\qquad\quad\;+ \Tr{1}^\delta\big[\sqrt{N}\big(\widehat{C}_N - C \big)\big] \ct \Tr{2}^\delta(C) \\
&\qquad\qquad\quad\;- \Tr{}^\delta\big[\sqrt{N}\big(\widehat{C}_N - C \big) \big] \big( A_1 \ct A_2 \big) \Big).
\end{split}
\end{equation}
Since $\sqrt{N}\big(\widehat{C}_N - C \big)$ converges in distribution to a Gaussian random element $Z$ by the central limit theorem of \cite{mas2006}, we obtain from the continuous mapping theorem in metric spaces \cite{billingsley2013} and Slutsky's theorem that the limiting distribution of $\sqrt{N}( \widehat{A}_1 \ct \widehat{A}_2 - A_1 \ct A_2 )$ is the law of
\[
\frac{1}{\Tr{}^\delta(C)} \Big( \Tr{1}^\delta(C) \ct \Tr{2}^\delta(Z) + \Tr{1}^\delta(Z) \ct \Tr{2}^\delta(C) - \Tr{}^\delta(Z) \big( A_1 \ct A_2 \big) \Big),
\]
which is Gaussian by Lemma \ref{lem:gaussianity}.

Finally, we turn our attention to $\widehat{B}$:
\begin{equation}\label{eq:asymptotics_B}
\begin{split}
\sqrt{N}(\widehat{B} - B) &= \sqrt{N}\left( \topavg(\widehat{C}_N - \widehat{A}_1 \ct \widehat{A}_2) - \topavg(C - A_1 \ct A_2) \right) \\
&= \topavg\left( \sqrt{N}(\widehat{C}_N - C) - \sqrt{N}( \widehat{A}_1 \ct \widehat{A}_2 - A_1 \ct A_2 ) \right) .
\end{split}
\end{equation}
Plugging formula \eqref{eq:linearization_distr} into \eqref{eq:asymptotics_B} and using the CMT again, we obtain
\[
\sqrt{N}(\widehat{B} - B) \stackrel{d}{\longrightarrow} \topavg \left( Z - \frac{\Tr{1}^\delta(C) \ct \Tr{2}^\delta(Z)}{\Tr{}^\delta(C)} - \frac{\Tr{1}^\delta(Z) \ct \Tr{2}^\delta(C)}{\Tr{}^\delta(C)} + \frac{\Tr{}^\delta(Z)}{\Tr{}^\delta(C)} \big( A_1 \ct A_2 \big) \right).
\]
The right-hand side before Toeplitz averaging is Gaussian again due to the reasons above. And by the previous lemma it remains Gaussian after Toeplitz averaging.
\end{proof}

\subsubsection{Proof of Theorem \ref{thm:rates_adaptive}}\label{app:G2}

Next, we provide the rates of convergence for the adaptive bandwidth choice of Section \ref{sec:delta}. For that, we first need to study the behavior of the empirical objective function in \eqref{eq:CV_empirical}.

Let us denote the empirical objective as
\[
\widehat{\Xi}(\delta) = \verti{\widehat{C}(\delta)}_2^2 - \frac{2}{N} \sum_{n=1}^N \langle X_n, \widehat{C}_{-n}(\delta) X_n \rangle
\]
and the theoretical objective as
\[
\Xi(\delta) = \verti{C(\delta) - C}_2^2.
\]
Recall that $\widehat{C}(\delta) = \widehat{A}(\delta) + \widehat{B}(\delta)$ is our separable-plus-banded estimator, while $C(\delta)$ is its limit version with infinite number of samples, i.e. $C(\delta) = A(\delta) + B(\delta)$ with
\[
A(\delta)= \frac{\Tr{1}^\delta(C) \ct \Tr{2}^\delta(C)}{\Tr{}^\delta(C)}, \quad B(\delta) = \topavg\big(C - A_1(\delta) \ct A_2(\delta)\big).
\]

The following linearization of our estimators will often allow us to develop suitable bounds:
\begin{equation}\label{eq:linearization}
\begin{split}
\widehat{A}(\delta_1) - A(\delta_2) &= \frac{\Tr{1}^{\delta_1}(\widehat{C}_N)}{\Tr{}^{\delta_1}(\widehat{C}_N)} \ct \left[ \Tr{2}^{\delta_1}(\widehat{C}_N) - \Tr{2}^{\delta_2}(C) \right] + \left[ \Tr{1}^{\delta_1}(\widehat{C}_N) - \Tr{1}^{\delta_2}(C) \right] \ct \frac{\Tr{2}^{\delta_2}(C)}{\Tr{}^{\delta_1}(\widehat{C}_N)} \\
&\hspace*{4.5cm}+ \frac{\Tr{1}^{\delta_2}(C) \ct \Tr{2}^{\delta_2}(C)}{\Tr{}^{\delta_2}(C) \Tr{}^{\delta_1}(\widehat{C}_N)} \left[ \Tr{}^{\delta_2}(C) - \Tr{}^{\delta_1}(\widehat{C}_N) \right].
\end{split}
\end{equation}

\begin{proposition}\label{prop:consistency}
Let $\delta$ be such that $\Tr{}^\delta(C) \neq 0$ and let assumption (A1) hold, then $\verti{\widehat{C}(\delta) - C(\delta)}_2^2 = \O_P(N^{-1})$.
\end{proposition}

\begin{proof}
Firstly, note that
\[
\verti{\widehat{C}(\delta) - C(\delta)}_2 \leq \verti{\widehat{A}(\delta) - A(\delta)}_2 + \verti{\topavg\big(\widehat{A}(\delta) - A(\delta)\big)}_2 + \verti{\widehat{C}_N - C}_2.
\]
Note that $\verti{\topavg\big(\widehat{A}(\delta) - A(\delta)\big)}_2 \leq \verti{\widehat{A}(\delta) - A(\delta)}_2$
due to Toeplitz averaging being a linear projection. Hence we only need to bound the difference between the separable parts, for which we use the linearization formula \eqref{eq:linearization}:
\begin{equation}\label{eq:proof11}
\begin{split}
\verti{\widehat{A}(\delta) - A(\delta)}_2 &\leq \verti{\widehat{A}(\delta) - A(\delta)}_1 \\
&\leq \frac{\verti{\Tr{1}^{\delta}(\widehat{C}_N)}_1}{|\Tr{}^{\delta}(\widehat{C}_N)|} \verti{ \Tr{2}^{\delta}(\widehat{C}_N - C)}_1 + \verti{\Tr{1}^{\delta}(\widehat{C}_N - C)}_1 \frac{\verti{\Tr{2}^{\delta}(C)}_1}{|\Tr{}^{\delta}(\widehat{C}_N)|} \\
&\hspace*{4.5cm}+ \frac{\verti{\Tr{1}^{\delta}(C)}_1 \verti{\Tr{2}^{\delta}(C)}_1}{|\Tr{}^{\delta}(C) \Tr{}^{\delta}(\widehat{C}_N)|} \left| \Tr{}^{\delta}(\widehat{C}_N - C) \right|.
\end{split}
\end{equation}
From \eqref{eq:shifted_trace_bound}, we have
\[
\verti{ \Tr{2}^{\delta}(\widehat{C}_N - C)}_1 \leq \verti{\widehat{C}_N - C}_1 = \O_P(N^{-1/2})
\]
since the CLT for $\widehat{C}_N$ holds \cite{mas2006}, and similarly for $\verti{\Tr{1}^{\delta}(\widehat{C}_N - C)}_1$ and $\big|\Tr{}^{\delta}(\widehat{C}_N - C)\big|$. The statement then follows upon noticing that the numerators on the right hand side of \eqref{eq:proof11} are obviously bounded while the denominators are bounded away from zero for $N$ large enough.
\end{proof}

\begin{proposition}\label{prop:obj_consistency}
Let $\delta$ be such that $\Tr{}^\delta(C) \neq 0$ and let assumption (A1) hold, then $\widehat{\Xi}(\delta) = \Xi(\delta) - \verti{C}_2^2 + \O_P(N^{-1/2})$.
\end{proposition}

\begin{proof}
Instead of the empirical objective, we will first work with a slightly modified, biased version of it:
\[
\widetilde{\Xi}(\delta) = \verti{\widehat{C}(\delta)}_2^2 - \frac{2}{N} \sum_{n=1}^N \langle X_n, \widehat{C}(\delta) X_n \rangle
\]
By adding and subtracting $\verti{\widehat{C}(\delta) - C}_2^2 = \verti{\widehat{C}(\delta)}_2^2 - 2 \langle \widehat{C}(\delta), C \rangle - \verti{C}_2^2$, we obtain
\[
\begin{split}
\left| \widetilde{\Xi}(\delta) + \verti{C}_2^2 - \Xi(\delta) \right| = \Bigg| \frac{2}{N}\sum_{n=1}^N \langle \widehat{C}(\delta), X_n \otimes X_n \rangle &- 2 \langle \widehat{C}(\delta), C \rangle \\&+ \verti{\widehat{C}(\delta) - C}_2^2 - \verti{C(\delta) - C}_2^2 \Bigg|,
\end{split}
\]
and from the triangle inequality we now have
\[
\begin{split}
\left| \widetilde{\Xi}(\delta) + \| C \|_2 - \Xi(\delta) \right| \leq &\left| \frac{2}{N}\sum_{n=1}^N \langle \widehat{C}(\delta), X_n \otimes X_n \rangle - 2 \langle \widehat{C}(\delta), C \rangle \right| \\
&+\left|\verti{\widehat{C}(\delta) - C}_2^2 - \verti{C(\delta) - C}_2^2 \right| =: (I) + (II).
\end{split}
\]
The first term can be bounded by Cauchy-Schwartz inequality
\[
(I) = 2 \left| \langle \widehat{C}(\delta), \widehat{C}_N - C \rangle \right| \leq 2 \verti{\widehat{C}(\delta)}_2 \verti{\widehat{C}_N - C}_2 = \O_P(N^{-1/2}),
\]
whereas the second term can be bounded similarly after using the mean value theorem:
\[
(II) = 2 \langle \Gamma - C, \widehat{C}(\delta) - C(\delta) \rangle \leq 2 \verti{\Gamma - C}_2 \verti{\widehat{C}(\delta) - C(\delta)},
\]
where $\Gamma$ is between $\widehat{C}(\delta)$ and $C(\delta)$. Hence the term $(II)$ is also $\O_P(N^{-1/2})$ according to the previous proposition.

Now it remains to show that the bias introduced by working with $\widetilde{\Xi}(\delta)$ instead of $\widehat{\Xi}(\delta)$ is asymptotically negligible. For that, it suffices to show that
\begin{equation}\label{eq:obj_bias}
\left| \frac{1}{N} \sum_{n=1}^N \langle \widehat{C}(\delta), X_n \otimes X_n \rangle - \frac{1}{N} \sum_{n=1}^N \langle \widehat{C}_{-n}(\delta), X_n \otimes X_n \rangle \right| = \O_P(N^{-1}).
\end{equation}
The previous expression can be bounded as
\[
\begin{split}
\left| \frac{1}{N} \sum_{n=1}^N \langle \widehat{C}(\delta) - \widehat{C}_{-n}(\delta), X_n \otimes X_n \rangle \right| &\leq \frac{1}{N} \sum_{n=1}^N \left| \langle \widehat{C}(\delta) - \widehat{C}_{-n}(\delta), X_n \otimes X_n \rangle \right| \\
&\leq \frac{1}{N} \sum_{n=1}^N \verti{\widehat{C}(\delta) - \widehat{C}_{-n}(\delta)}_2 \| X_n \|_2^2.
\end{split}
\]
Using the linearization argument \eqref{eq:linearization} again like in the proof of the previous proposition, we obtain
\[
\verti{\widehat{C}(\delta) - \widehat{C}_{-n}(\delta)}_2 \leq \big[const + \O_P(N^{-1/2})\big] \verti{\widehat{C}_N - \widehat{C}_{-n}}_1
\]
where $\widehat{C}_{-n}$ is the empirical covariance estimator without the $n$-th observation. Since
\[
\widehat{C}_N - \widehat{C}_{-n} = \frac{1}{N} X_n \otimes X_n + \frac{1}{N(N-1)} \sum_{j \neq n} X_j \otimes X_j
\]
for any $n=1,\ldots,N$, we have from the triangle inequality that
\[
\verti{\widehat{C}_N - \widehat{C}_{-n}}_1 \leq \frac{1}{N} \| X_n \|_2^2 + \frac{1}{N(N-1)} \sum_{j \neq n} \|X_j\|_2^2.
\]
Overall, we have that the left-hand size of \eqref{eq:obj_bias} is bounded by
$\big[ const + \O_P(N^{-1/2}) \big] \frac{1}{N} D_N$,
where
\[
D_N = \frac{1}{N} \sum_{n=1}^N \|X_n\|_2^4 + \frac{1}{N(N-1)} \sum_{n=1}^N \sum_{j\neq N} \|X_n\|_2^2 \|X_j\|_2^2
\]
which is $o_P(1)$ from the law of large numbers. Hence we get \eqref{eq:obj_bias} and the proof is complete.
\end{proof}

According to the previous proposition, the empirical objective is consistent for the theoretical objective up to a constant. And the constant, though unknown, does not affect the bandwidth choice. This leads to the rates of convergence of our estimators with adaptively chosen bandwidth as stated in Theorem \ref{thm:rates_adaptive} of the main paper.

\begin{proof}[Proof of Theorem \ref{thm:rates_adaptive}]
We have from the triangle inequality
\[
\verti{\widehat{A}(\widehat{\delta}) - A}_2 \leq \verti{\widehat{A}(\widehat{\delta}) - A(\widehat{\delta})}_2 + \verti{A(\widehat{\delta}) - A}_2.
\]

Due to our assumptions the separable-plus-banded model holds with a certain $\delta^\star$ and there exists at least one $\delta \in \Delta$ such that the separable-plus-banded model holds with $\delta$. On the other hand, for any $\delta < \delta^\star$ the separable-plus-banded model does not hold and hence $\Xi(\delta) > \Xi(\delta^\star)$. Therefore, due to Proposition \ref{prop:obj_consistency}, there exists $N_0$ such that for all $N \geq N_0$ we have $\widehat{\delta} \geq \delta^\star$. Thus $\verti{A(\widehat{\delta}) - A}_2 = 0$ for all $N \geq N_0$.

Secondly, we observe from the proof of Proposition \ref{prop:consistency} that $\verti{\widehat{A}(\delta) - A(\delta)}_2 = \O_p(N^{-1/2})$ for any $\delta \in \Delta$ such that $\delta \geq \delta^\star$, hence also for $\widehat{\delta}$, and the proof is complete.

The assertion for the banded part follows easily using the previous part of the proof and triangle inequalities:
\[
\verti{\widehat{B}(\widehat{\delta}) - B}_2 \leq \verti{\widehat{B}(\widehat{\delta}) - B(\widehat{\delta})}_2 + \verti{B(\widehat{\delta}) - B}_2,
\]
where 
\[
\begin{split}
\verti{\widehat{B}(\widehat{\delta}) - B(\widehat{\delta})}_2 &= \verti{\topavg\big(\widehat{C}_N - \widehat{A}(\widehat{\delta}) - C + A(\widehat{\delta})\big)}_2 \leq \verti{\widehat{C}_N - \widehat{A}(\widehat{\delta}) - C + A(\widehat{\delta})}_2 \\
&\leq \verti{ \widehat{C}_N - C}_2 + \verti{\widehat{A}(\widehat{\delta}) - A(\widehat{\delta})}_2 = \O_P(N^{-1/2}),
\end{split}
\]
and similarly
\[
\verti{B(\widehat{\delta}) - B}_2 \leq \verti{\widehat{C}_N - C}_2 + \verti{A(\widehat{\delta}) - A}_2 = \O_P(N^{-1/2}).
\]
\end{proof}

In case that the separable-plus-banded model does not hold, i.e. $C$ does not posses the separable-plus-banded structure for any $\delta$, there still exists $\delta_0 \in \Delta$ such that
\[
\Xi(\delta_0) = \min_{\delta \in \Delta} \Xi(\delta),
\]
and the same argument as the one in the previous proof yields that $\widehat{C}(\widehat{\delta})$ is root-$n$ consistent for $C(\delta_0)$. In this instance, $C(\delta_0) \neq C$, but $C(\delta_0)$ is the best separable-plus-banded proxy to $C$, which can be obtained by the proposed estimation methodology based on shifted partial tracing.

\subsubsection{Asymptotic Distribution under Adaptive Bandwidth Choice}\label{app:G3}

It is clear from the previous section, that the adaptively chosen bandwidth itself is not consistent. This is because there is nothing to be consistent for: under the separable-plus-banded model, there is a whole range of valid bandwidths, which are asymptotically indistinguishable. Still, all those bandwidths lead asymptotically to the same estimator, and hence we are able to show consistency of $\widehat{C}(\widehat{\delta})$ for $C(\delta^\star)$ even if $\widehat{\delta}$ itself is not consistent for $\delta^\star$. This reflects that $\delta$ is merely a nuisance parameter.

However, the development in the previous section only provides rates of convergence, not the asymptotic distribution. The latter can be obtained under a slight modification of our bandwidth selection scheme. For $\tau \geq 0$, we define
$\Xi_\tau(\delta) := \Xi(\delta) + \tau\delta$, $\widehat{\Xi}_\tau(\delta) := \widehat{\Xi}(\delta) + \tau\delta$ and
\begin{equation}\label{eq:CV_tau}
\widetilde{\delta} := \argmin_{\delta \in \Delta} \widehat{\Xi}_\tau(\delta).
\end{equation}
A new parameter $\tau$ has been introduced into the objective to discriminate between equally good choices of $\delta$. The theorem below shows that under the modified scheme, $\widehat{C}(\widetilde{\delta})$ is asymptotically Gaussian, when $\tau>0$ is small enough.

On one hand, $\tau$ is a new nuisance parameter that needs to be chosen instead of $\delta$. On the other hand, it is easier to choose it (it just needs to be small enough). Moreover, the previous section shows that choosing $\tau=0$ provides the correct rates of convergence.

\begin{theorem}\label{thm:asymptotics_adaptive}
Let $X_1,\ldots,X_N \sim X$ be a (w.l.o.g.~centered) random sample with covariance given by \eqref{eq:model}, where $B$ is stationary and $\delta^\star$-banded. Let (A1) hold for some orthonormal basis $\{ e_j \}_{j=1}^\infty$ in $\mathcal{L}^2[0,1]^2$. Let $\Delta = \{ \delta_1,\ldots, \delta_m\}$ be such that $\Tr{}^\delta(C) \neq 0$ for any $\delta \in \Delta$ of which at least one is larger that $\delta^\star$. Finally, let $\widehat{\delta}$ be chosen in \eqref{eq:CV_tau} with $\tau < \min_{\delta \in \Delta, \delta < \delta^\star} |\Xi(\delta) - \Xi(\delta^\star) |$.
Then $\sqrt{N}( \widehat{A}_1(\widetilde{\delta}) \ct \widehat{A}_2(\widetilde{\delta}) - A_1 \ct A_2)$ and $\sqrt{N}( \widetilde{B}(\widehat{\delta}) - B )$
converge to mean zero Gaussian random elements.
\end{theorem}

\begin{proof}[Proof of Theorem \ref{thm:asymptotics_adaptive}]
We begin with the asymptotic distribution for $\widehat{A}(\widehat{\delta})$, modifying the proof of Theorem \ref{thm:asymptotics}.

Let us denote by $\delta_m^\star$ the smallest of such bandwidths in $\Delta$ which is larger than $\delta^\star$. By Theorem \ref{thm:asymptotics}, we know that $\sqrt{N} \big( \widehat{A}(\delta_m^\star) - A \big)$ is asymptotically Gaussian and mean-zero. Since
\[
\sqrt{N}\big(\widehat{A}(\widehat{\delta}) - A\big) = \sqrt{N}\big(\widehat{A}(\widehat{\delta}) - \widehat{A}(\delta_m^\star)\big) + \sqrt{N}\big(\widehat{A}(\delta_m^\star) - A\big),
\]
we only need to show that $\sqrt{N}\big(\widehat{A}(\widehat{\delta}) - \widehat{A}(\delta_m^\star)\big)$ converges to zero in probability.

Let us denote $\widetilde{\Delta} := \{\delta \in \Delta; \delta < \delta_m^\star\}$. Since the separable-plus-banded model holds, and $\delta_m^\star$ is the smallest bandwidth in $\Delta$ such that $B$ is banded by this bandwidth, it must be $\verti{C(\delta) - C}_2 > 0$ for all $\delta \in \widetilde{\Delta}$.

Now, fix any $\epsilon > 0$, and observe that
\begin{equation}\label{eq:probability}
P\left( \left| \sqrt{N}\big(\widehat{A}(\widehat{\delta}) - \widehat{A}(\delta_m^\star)\big) \right| > \epsilon \right) \leq P\left(\widehat{\delta} \neq \delta_m^\star\right) = P\left(\argmin_{\delta \in \Delta} \widehat{\Xi}_\tau(\delta) \neq \argmin_{\delta \in \Delta} \Xi_\tau(\delta)\right).
\end{equation}
Let $\alpha >0$ be arbitrary. For any $j$ such that $\delta_j \neq \delta_m^\star$ there exists $N_j \in \N$ such that for all $N \geq N_j$ we have
\[
\left| \widehat{\Xi}_\tau(\delta_j) + \verti{C}_2^2 - \Xi_\tau(\delta_j) \right| < \alpha.
\]
Taking $N_0 := \max N_j$ and $\alpha := \tau$ we obtain that the probability in \eqref{eq:probability} is equal to zero for any $N \geq N_0$ and the proof is thus complete.

The proof for $\widehat{B}(\widehat{\delta})$ is an equivalent modification of the proof of Theorem \ref{thm:asymptotics}.
\end{proof}

\subsubsection{Proof of Theorem \ref{thm:rates}}\label{app:G4}

Now we move our attention to Theorem \ref{thm:rates}, providing rates of convergence under discrete and noisy observations. The following auxiliary result will be needed.

\begin{lemma}\label{lem:trivial_inequality}
Let $Z_1,\ldots,Z_K$ be i.i.d. random variables with finite second moments. Then
\[
\E\Big( \sum_{k=1}^K Z_k \Big)^2 \leq K \sum_{k=1}^K \E Z_k^2 \,.
\]
\end{lemma}

\begin{proof}
The claim follows from the Cauchy-Schwartz inequality followed by the arithmetic-geometric mean inequality:
\[
\begin{split}
\E\Big( \sum_{k=1}^K Z_k \Big)^2 &= \sum_{k=1}^K \sum_{l=1}^K \E Z_k Z_l \leq \sum_{k=1}^K \sum_{l=1}^K \sqrt{\E Z_k^2} \sqrt{\E Z_l^2} \\
&\leq \sum_{k=1}^K \sum_{l=1}^K \frac{\E Z_k^2 + \E Z_l^2}{2} = K \sum_{k=1}^K \E Z_k^2 \,.
\end{split}
\]
\end{proof}

\begin{proof}[Proof of Theorem \ref{thm:rates}, pointwise sampling scheme S1.]
We begin with the bias-variance decomposition
\[
\verti{\widehat{A}_2^K \ct \widehat{A}_1^K - A_1 \ct A_2}_2^2 = 2 \verti{\widehat{A}_1^K \ct \widehat{A}_2^K - A_1^K \ct A_2^K}_2^2 + 2 \verti{A_1^K \ct A_2^K - A_1 \ct A_2}_2^2 \,.
\]
For the bias term, we first distribute the norm calculation over the grid:
\[
\begin{split}
\verti{A_1^K \ct A_2^K - A_1 \ct A_2}_2^2  &= \sum_{i,j,k,l=1}^K \int_{I_{i,j}^K \times I_{k,l}^K} \Big[a_1^K(t,t')a_2^K(s,s') - a_1(t,t')a_2(s,s') \Big]^2 d t d s d t' d s' .
\end{split}
\]
Since $a_1^K(t,s)a_2^K(s,s') = a_1(t_i,t_k)a_2(t_{j},s_{l})$ on $I_{i,j}^K \times I_{k,l}^K$, it follows from Lipschitz continuity that
\[
\begin{split}
|a_1^K(t,t')a_2^K(s,s') - a_1(t,t')a_2(s,s')| &\leq L \sup_{(t,s,t',s') \in I_{i,j}^K \times I_{k,l}^K} \| (t,s,t',s') - (t_i,s_j,t_k,s_l) \|_2 \\
&\leq 4^{1/2} K^{-1} L \,,
\end{split}
\]
which implies the bound for the bias term. It remains to show that the variance term is $\mathcal{O}_P(N^{-1})$ uniformly in $K$.

Since $\Tr{}^\delta(A) > 0$ and $\delta_K = \lceil \delta K \rceil/K \searrow \delta$, due to continuity of kernel $a$ there exist $K_0 \in \N$ such that $\Tr{}^{\delta_K}(A) > 0$ for any $K\geq K_0$. Assume from now on that $K \geq K_0$.

Using that $A_1^K \ct A_2^K = \frac{\Tr{1}^{\delta_K}(C^K) \ct \Tr{2}^{\delta_K}(C^K)}{\Tr{}^{\delta_K}(C^K)}$ in our model, it follows from the triangle inequality that
\begin{equation}\label{eq:three_term_bound}
\begin{split}
\verti{\widehat{A}_1^K \ct \widehat{A}_2^K - A_1^K \ct A_2^K}_2 &= \verti{\frac{\Tr{1}^{\delta_K}(\widehat{C}_N^K) \ct \Tr{2}^{\delta_K}(\widehat{C}_N^K)}{ \Tr{}^{\delta_K}(\widehat{C}_N^K) } -  \frac{\Tr{1}^{\delta_K}(C^K) \ct \Tr{2}^{\delta_K}(C^K)}{ \Tr{}^{\delta_K}(C^K) }}_2 \\
&\leq \frac{\verti{\Tr{1}^{\delta_K}(\widehat{C}_N^K)}_2}{\left|\Tr{}^{\delta_K}(\widehat{C}_N^K)\right|} \verti{\Tr{2}^{\delta_K}(\widehat{C}_N^K - C^K)}_2 \\
&\quad+\frac{\verti{\Tr{2}^{\delta_K}(C^K)}_2}{\left|\Tr{}^{\delta_K}(\widehat{C}_N^K)\right|} \verti{\Tr{1}^{\delta_K}(\widehat{C}_N^K - C^K)}_2 \\
&\quad + \frac{\verti{A_1^K \ct A_2^K}_2}{\left|\Tr{}^{\delta_K}(\widehat{C}_N^K)\right|} \left|\Tr{}^{\delta_K}(\widehat{C}_N^K - C^K)\right| .
\end{split}
\end{equation}

Now we treat different terms separately. The numerators will be shown to be $\mathcal{O}_P(1)$, as well as $1/\left|\Tr{}^{\delta_K}(\widehat{C}_N^K)\right|$, while the remaining terms will be shown to be $\mathcal{O}_P(N^{-1/2})$; all these rates being uniform in $K$. To simplify the notation, we denote $\widetilde{k} := k + \delta_K K$ and $\widetilde{K} := (1- \delta_K) K$.

Firstly, we show that $\verti{\Tr{1}^{\delta_K}(\widehat{C}_N^K - C^K)}_2 = \mathcal{O}_P(N^{-1/2})$ uniformly in $K$. To that end, since $\mathbf{C}^K = \E ( \mathbf{X}^K \otimes \mathbf{X}^K )$, using Lemma \ref{lem:continuous_discrete_traces}, and denoting $d = \delta_K K$, we have
\[
\begin{split}
\E \verti{\Tr{1}^{\delta_K}(\widehat{C}_N^K - C^K)}_2^2 &= K^{-4} \E \left\| \Tr{1}^{d}(\widehat{\mathbf C}_N^K - \mathbf C^K)\right\|_F^2 \\
&= K^{-4} \sum_{i=1}^K \sum_{j=1}^K \E \left| \Tr{1}^{d}(\widehat{\mathbf C}_N^K - \mathbf C^K)[i,j] \right|^2 \\
&= K^{-4} \sum_{i=1}^K \sum_{j=1}^K \E \Bigg| \frac{1}{N} \sum_{n=1}^N \sum_{k=1}^{\widetilde{K}} \Big( \widetilde{\mathbf X}_n^K[i,k] \widetilde{\mathbf X}_n^K[j,\widetilde{k}] \\
&\hspace*{4cm}- \E \mathbf X_n^K[i,k] \mathbf X_n^K[i,\widetilde k] \Big) \Bigg|^2 \,.
\end{split}
\]

If we denote
\[
\begin{split}
Z_{n,i,j} &:= \sum_{k=1}^{\widetilde{K}} \Big( \widetilde{\mathbf X}_n^K[i,k] \widetilde{\mathbf X}_n^K[j,\widetilde{k}] - \E \mathbf X_n^K[i,k] \mathbf X_n^K[i,\widetilde k] \Big) \\
& = \sum_{k=1}^{\widetilde{K}} \Big( \mathbf X_n^K[i,k] \mathbf X_n^K[j,\widetilde{k}] - \E \mathbf X_n^K[i,k] \mathbf X_n^K[i,\widetilde k] \\
&\qquad\qquad+ \mathbf E_n^K[i,k] \mathbf X_n^K[j,\widetilde{k}] + \mathbf X_n^K[i,k] \mathbf E_n^K[j,\widetilde{k}] + \mathbf E_n^K[i,k] \mathbf E_n^K[j,\widetilde{k}] \Big) \,,
\end{split}
\]
we see that, for any $i,j=1,\ldots,K$, $\big\{ Z_{n,i,j} \big\}_{n=1}^N$ is a set of mean zero and i.i.d. random variables and thus
\[
E \verti{\Tr{1}^{\delta_K}(\widehat{C}_N^K - C^K)}_2^2 = \frac{1}{N} K^{-4} \sum_{i=1}^K \sum_{j=1}^K \E \left| Z^K_{\cdot,i,j} \right|^2
\]
which can be bounded, using the parallelogram law, by
\begin{equation}\label{eq:four_term_bound}
\begin{split}
&\frac{4}{N} K^{-4} \sum_{i=1}^K \sum_{j=1}^K \Bigg\{
\E \left| \sum_{k=1}^{\widetilde{K}}  \mathbf X^K[i,k] \mathbf X^K[j,\widetilde{k}] - \E \mathbf X^K[i,k] \mathbf X^K[i,\widetilde k]  \right|^2 \\
&\qquad\quad + \E \left| \sum_{k=1}^{\widetilde{K}}  \mathbf E^K[i,k] \mathbf X^K[j,\widetilde{k}]  \right|^2
 + \E \left| \sum_{k=1}^{\widetilde{K}}  \mathbf X^K[i,k] \mathbf E^K[j,\widetilde{k}]  \right|^2
 + \E \left| \sum_{k=1}^{\widetilde{K}}  \mathbf E^K[i,k] \mathbf E^K[j,\widetilde{k}]  \right|^2 \Bigg\} \,.
\end{split}
\end{equation}
The four terms in the parentheses will be treated separately.

For the first term, it follows from Lemma \ref{lem:trivial_inequality} that
\[
\E \left| \sum_{k=1}^{\widetilde{K}}  \mathbf X^K[i,k] \mathbf X^K[j,\widetilde{k}] - \E \mathbf X^K[i,k] \mathbf X^K[j,\widetilde k]  \right|^2 \leq \widetilde{K} \sum_{k=1}^{\widetilde{K}} \var\Big( \mathbf X^K[i,k] \mathbf X^K[j,\widetilde{k}] \Big) \leq S_1 \widetilde{K}^2 \,,
\]
where
\[
\begin{split}
\var\big( \mathbf X^K[i,k] \mathbf X^K[j,\widetilde{k}] \big) &= \var \big( X(t_i^K,s_k^K) X(t_j,s_{\widetilde{k}}) \big)  \\
&\leq \sup_{t,s,t',s' \in [0,1]} \var\big( X(t,s) X(t',s') \big) =: S_1 < \infty.
\end{split}
\]
Note that $S_1$ is finite, since $X$ has finite fourth moment and continuous sample paths. Also, $S_1$ is uniform in $K$.

For the second term, we have (denoting $\widetilde{l} = l + \delta_k K$)
\[
\begin{split}
\E \left| \sum_{k=1}^{\widetilde{K}}  \mathbf E^K[i,k] \mathbf X^K[j,\widetilde{k}]  \right|^2 &= \sum_{k=1}^{\widetilde{K}} \sum_{l=1}^{\widetilde{K}} \E \Big( \mathbf E^K[i,k] \mathbf X^K[j,\widetilde{k}] \mathbf E^K[i,l] \mathbf X^K[j,\widetilde{l}] \Big)  \\
& = \sum_{k=1}^{\widetilde{K}} \sum_{l=1}^{\widetilde{K}} \E \Big( \mathbf E^K[i,k] \mathbf E^K[i,l] \Big) \E \Big( \mathbf X^K[j,\widetilde{k}] \mathbf X^K[j,\widetilde{l}] \Big) \,.
\end{split}
\]
Since $\E \big( \mathbf E^K[i,k] \mathbf E^K[i,l] \big) = \sigma^2 \mathds{1}_{[k=l]}$, one of the sums vanishes, while $\E\big| \mathbf X^K[j,\widetilde{k}] \big|^2$ is bounded uniformly in $K$ by $S_2 := \sup_{t,s \in [0,1]} \E \big| X(t,s) \big|^2 \leq \infty$. Hence the second term is bounded by $\widetilde{K} S_2 \sigma^2$. The third term is dealt with similarly.

For the fourth and final term, we have
\[
\begin{split}
\E \left| \sum_{k=1}^{\widetilde{K}}  \mathbf E^K[i,k] \mathbf E^K[j,\widetilde{k}]  \right|^2 &= \sum_{k=1}^{\widetilde{K}} \sum_{l=1}^{\widetilde{K}} \E \Big( \mathbf E^K[i,k] \mathbf E^K[j,\widetilde{k}] \mathbf E^K[i,l] \mathbf E^K[j,\widetilde{l}]\Big) \\
&= \sum_{k=1}^{\widetilde{K}} \sum_{l=1}^{\widetilde{K}} \sigma^4 \mathds{1}_{[k=l]} = \widetilde{K} \sigma^4 \,.
\end{split}
\]

Upon collecting the bounds for the four terms and importing them back to bound \eqref{eq:four_term_bound}, we obtain
\begin{equation}\label{eq:bound_SPT}
\E \verti{\Tr{1}^{\delta_K}(\widehat{C}_N^K - C^K)}_2^2 \leq \frac{4}{N}\Big[ S_1 + S_2 K^{-1} \sigma^2 + K^{-1} \sigma^4 \Big] \,.
\end{equation}
This shows that if $\sigma^2=\mathcal{O}(\sqrt{K})$, $\verti{\Tr{1}^{\delta_K}(\widehat{C}_N^K - C^K)}_2 = \mathcal{O}_P(N^{-1/2})$ uniformly in $K$.

The term $\verti{\Tr{2}^{\delta_K}(\widehat{C}_N^K - C^K)}_2$ from bound \eqref{eq:three_term_bound} can be treated similarly. Now we focus on the final stand-alone term $\left|\Tr{}^{\delta_K}(\widehat{C}_N^K - C^K)\right|$:
\[
\begin{split}
\E \left|\Tr{}^\delta(\widehat{C}_N^K - C^K)\right|^2 &= K^{-4} \E \left| \Tr{}^\delta(\widehat{\mathbf C}_N^K - \mathbf C^K)\right|^2 \\
&= K^{-4} \E \left| \frac{1}{N} \sum_{n=1}^N \sum_{i=1}^{\widetilde{K}} \sum_{j=1}^{\widetilde{K}} \Big(\widetilde{\mathbf X}_n^K[i,j] \widetilde{\mathbf X}_n^K[\widetilde{i},\widetilde{j}] - \E \mathbf X_n^K[i,j] \mathbf X_n^K[\widetilde{i},\widetilde{j}] \Big)  \right|^2 \\
&= \frac{1}{N} K^{-4} \E \left| \sum_{i=1}^{\widetilde{K}} \sum_{j=1}^{\widetilde{K}} \Big(\widetilde{\mathbf X}^K[i,j] \widetilde{\mathbf X}^K[\widetilde{i},\widetilde{j}] - \E \mathbf X^K[i,j] \mathbf X^K[\widetilde{i},\widetilde{j}] \Big)  \right|^2 \\
\end{split}
\]
From the parallelogram law we have
\[
\begin{split}
\E \left|\Tr{}^\delta(\widehat{C}_N^K - C^K)\right|^2 &\leq \frac{4}{N} K^{-4} \Bigg\{ \E \left| \sum_{i=1}^{\widetilde{K}} \sum_{j=1}^{\widetilde{K}} \Big(\mathbf X^K[i,j] \mathbf X^K[\widetilde{i},\widetilde{j}] - \E \mathbf X^K[i,j] \mathbf X^K[\widetilde{i},\widetilde{j}] \Big)  \right|^2 \\
&\qquad\qquad\quad + \E \left| \sum_{i=1}^{\widetilde{K}} \sum_{j=1}^{\widetilde{K}} \mathbf E^K[i,j] \mathbf X^K[\widetilde{i},\widetilde{j}]  \right|^2 \\
&\qquad\qquad\quad + \E \left| \sum_{i=1}^{\widetilde{K}} \sum_{j=1}^{\widetilde{K}} \mathbf X^K[i,j] \mathbf E^K[\widetilde{i},\widetilde{j}]  \right|^2 \\
&\qquad\qquad\quad + \E \left| \sum_{i=1}^{\widetilde{K}} \sum_{j=1}^{\widetilde{K}} \mathbf E^K[i,j] \mathbf E^K[\widetilde{i},\widetilde{j}]  \right|^2 \Bigg\} \,.
\end{split}
\]
Using Lemma \ref{lem:trivial_inequality} to take the sums out of the expectation, the first term in the parentheses is again bounded by $K^4 S_1$. For the second term,
\[
\begin{split}
\E \left| \sum_{i=1}^{\widetilde{K}} \sum_{j=1}^{\widetilde{K}} \mathbf E^K[i,j] \mathbf X^K[\widetilde{i},\widetilde{j}]  \right|^2 &= \sum_{i,j,k,l=1}^{\widetilde{K}} \E \Big( \mathbf E^K[i,j] \mathbf X^K[\widetilde{i},\widetilde{j}] \mathbf E^K[k,l] \mathbf X^K[\widetilde{k},\widetilde{l}] \Big) \\
&= \sigma^2 \sum_{i=1}^{\widetilde{K}} \sum_{j=1}^{\widetilde{K}} \E \left| X^K[\widetilde{i},\widetilde{j}] \right|^2 \leq K^2 \sigma^2 S_2 \,.
\end{split}
\]
The third term can be treated similarly, while for the fourth and final term we have
\[
\begin{split}
\E \left| \sum_{i=1}^{\widetilde{K}} \sum_{j=1}^{\widetilde{K}} \mathbf E^K[i,j] \mathbf E^K[\widetilde{i},\widetilde{j}]  \right|^2  &= \sum_{i,j,k,l=1}^{\widetilde{K}} \E \Big( \mathbf E^K[i,j] \mathbf E^K[\widetilde{i},\widetilde{j}] \mathbf E^K[k,l] \mathbf E^K[\widetilde{k},\widetilde{l}] \Big) \\
&= \sum_{i=1}^{\widetilde{K}} \sum_{j=1}^{\widetilde{K}} \E \left|\mathbf E^K[i,j] \right|^2 \E \left| \mathbf E^K[\widetilde{i},\widetilde{j}] \right|^2
\leq K^2 \sigma^4 \,.
\end{split}
\]
Hence we obtain
\begin{equation}\label{eq:bound_ST}
\E \left|\Tr{}^\delta(\widehat{C}_N^K - C^K)\right|^2 \leq \frac{4}{N}\Big[ S_1 + S_2 K^{-2} \sigma^2 + K^{-2} \sigma^4 \Big]\,.
\end{equation}

Note the different powers of $K$ in \eqref{eq:bound_SPT} and \eqref{eq:bound_ST}. This reflects that the concentration of measurement error is weaker when averaging is performed over both time and space (when shifted tracing is used) in comparison to averaging only over either time or space (when shifted partial tracing is used).

Now let us focus on the numerators in \eqref{eq:three_term_bound}, for example:
\[
\verti{\Tr{1}^{\delta_K}(\widehat{C}_N^K)}_2 \leq \verti{\Tr{1}^{\delta_K}(C^K)}_2 + \verti{\Tr{1}^{\delta_K}(\widehat{C}_N^K - C^K)}_2,
\]
where the second term is $\mathcal{O}_P(N^{-1/2})$ uniformly in $K$, while the first term is clearly bounded by\\ $\sup_{t,s,t',s' \in [0,1]} c(t,s,t',s') < \infty$, hence $\verti{\Tr{1}^{\delta_K}(\widehat{C}_N^K)}_2$ is $\mathcal{O}_P(1)$ uniformly in $K$. Similarly for the other two numerator terms.

Finally, we consider the denominators in \eqref{eq:three_term_bound}. The reverse triangle inequality implies
\[
\left|\Tr{}^{\delta_K}(\widehat{C}_N^K)\right| \geq \left|\Tr{}^{\delta_K}(C^K)\right| - \left|\Tr{}^{\delta_K}(\widehat{C}_N^K - C^K)\right|,
\]
where the second term is again $\mathcal{O}_P(N^{-1/2})$ uniformly in $K$ as shown above, and the first term is bounded away from 0 uniformly in $K$ (for large enough $K$) due to continuity of the kernel $a$ of the separable part $A$ and the assumption $\Tr{}^{\delta_K}(A) > 0$, because $\Tr{}^{\delta_K}(A) = \Tr{}^{\delta_K}(C)$. Hence $1/\left|\Tr{}^{\delta_K}(\widehat{C}_N^K)\right|$ is $\mathcal{O}_P(1)$ uniformly in $K$.

The proof of the rates for the separable estimator is complete upon collecting the rates for the different terms in \eqref{eq:three_term_bound}.

The rate for the eigenvalues follows from the perturbation bounds \cite[Lemma 4.2]{bosq2012}:
\[
|\widehat{\lambda}_i^K \widehat{\rho}_j^K - \lambda_i \rho_j|^2 \leq \verti{\widehat{A}_1^K \ct \widehat{A}_2^K - A_1 \ct A_2}_2^2.
\]
To show the rates for the eigenvectors, we will use again the perturbation bounds \cite[Lemma 4.3]{bosq2012}:
\[
\| \widehat{e}_j^K - \sign(\langle \widehat{e}_j^K, e_j \rangle e_j) \|_2 \leq \alpha \| \widehat{A}_1^K - A_1 \|_2,
\]
where $\alpha$ is a constant depending on spacing between the eigenvalues. We cannot use this result directly, since we do not have consistency of $\widehat{A}_1^K$ (this is because of the scaling issues: $A_1 \ct A_2 = (\alpha A_1) \ct (A_2)/\alpha$ for any $\alpha$). Hence similar bounds always have to be used in the product space. However, this poses no issues due to Lemma \ref{lem:ct}. We have
\[
\begin{split}
\| \widehat{e}_j^K - \sign(\langle \widehat{e}_j^K, e_j \rangle e_j) \|_2 &= \| f_j \|_2 \| \widehat{e}_j^K - \sign(\langle \widehat{e}_j^K, e_j \rangle e_j) \|_2 \\
&= \| \widehat{e}_j^K \otimes f_j - \sign(\langle \widehat{e}_j^K, e_j \rangle) e_j \otimes f_j \|_2 \\
&\leq \| \widehat{e}_j^K \otimes \widehat{f}_j^K - \sign(\langle \widehat{e}_j^K, e_j \rangle) \sign(\langle \widehat{f}_j^K, f_j \rangle) e_j \otimes f_j \|_2.
\end{split}
\]
The previous inequality follows from the Cauchy-Schwartz inequality and the fact that the left-hand side of the inequality equal to $2 - 2 \sign(\langle \widehat{e}_j^K, e_j \rangle) \langle \widehat{e}_j^K, e_j \rangle e_j \rangle \langle f_j, f_j \rangle $ while the right-hand side is equal to $2 - 2 \sign(\langle \widehat{e}_j^K, e_j \rangle) \sign(\langle \widehat{f}_j^K, f_j \rangle) \langle \widehat{e}_j^K, e_j \rangle e_j \rangle \langle \widehat{f}_j^K, f_j \rangle$. Altogether, the rate for $\widehat{A}_1^K \ct \widehat{A}_2^K$ translates to the eigenvectors of $\widehat{A}_1^K$, and similarly for the eigenvectors of $\widehat{A}_2^K$.
\end{proof}

Regarding the eigenvalues, we cannot bound $|\widehat{\lambda}_j^K - \lambda_j|$ since our estimators $A_1$ and $A_2$ can be re-scaled versions of $A_1$ and $A_2$. However, our estimators of the product $A_1 \ct A_2$ are consistent, and we know from Lemma \ref{lem:ct} that eigenvectors of $A_1 \ct A_2$ are given as $\lambda_i \rho_j$, i.e. products of the eigenvalues of $A_1$ and $A_2$. Those products are still estimated with the same rates, similarly as the eigenvectors, using Lemma 4.2 of \cite{bosq2012}.

The proof of the theorem in the case of pixel-wise sampling scheme (S2) is in many regards similar, but some arguments are slightly more subtle.

\begin{proof}[Proof of Theorem \ref{thm:rates}, pixel-wise sampling scheme S2.]

We begin again the by the bias-variance decomposition and bound the bias term in the same manner. For the variance term, we use the triangle inequality treat all the terms in \eqref{eq:three_term_bound} separately. The fractions are also treated the same way as before and the conclusion of the proof will follow similarly, once it is established that $\verti{\Tr{1}^{\delta_K}(\widehat{C}_N^K - C^K)}_2$, $\verti{\Tr{2}^{\delta_K}(\widehat{C}_N^K - C^K)}_2$ and $\left|\Tr{}^{\delta_K}(\widehat{C}_N^K - C^K)\right|$ are all $\mathcal{O}_P(N^{-1/2})$ uniformly in $K$. Establishing these rates for the pointwise sampling scheme (S1) was the bulk of the previous proof, and now we will establish the same for the pixel-wise sampling scheme (S2).

We begin with $\verti{\Tr{1}^{\delta_K}(\widehat{C}_N^K - C^K)}_2$. Exactly as in the previous proof, we obtain the bound \eqref{eq:four_term_bound} here as well:
\[
\begin{split}
\E \verti{\Tr{1}^{\delta_K}(\widehat{C}_N^K - C^K)}_2^2 &\leq \frac{4}{N} \Bigg\{ K^{-4} \sum_{i=1}^K \sum_{j=1}^K \E \left| \sum_{k=1}^{\widetilde{K}} \mathbf X^K[i,k] \mathbf X^K[j,\widetilde{k}] \right|^2 \\
&\qquad\quad+ K^{-4} \sum_{i=1}^K \sum_{j=1}^K \E \left| \sum_{k=1}^{\widetilde{K}} \mathbf E^K[i,k] \mathbf X^K[j,\widetilde{k}] \right|^2\\
&\qquad\quad  + K^{-4} \sum_{i=1}^K \sum_{j=1}^K \E \left| \sum_{k=1}^{\widetilde{K}} \mathbf X^K[i,k] \mathbf E^K[j,\widetilde{k}] \right|^2 \\
&\qquad\quad+ K^{-4} \sum_{i=1}^K \sum_{j=1}^K \E \left| \sum_{k=1}^{\widetilde{K}} \mathbf E^K[i,k] \mathbf E^K[j,\widetilde{k}] \right|^2 \Bigg\} \\ &=: \frac{4}{N} \Bigg\{ (I) + (II) + (III) + (IV) \Bigg\} ,
\end{split}
\]
and again we treat the four terms in the parentheses (labeled by Roman numbers) separately.

For the first term, we drop the inner expectation only increasing the term and obtaining
\[
\begin{split}
(I) &= K^{-4} \sum_{i=1}^K \sum_{j=1}^K \E \left| \sum_{k=1}^{\widetilde{K}} \mathbf X^K[i,k] \mathbf X^K[j,\widetilde{k}] \right|^2 \\
&= K^{-4} \sum_{i=1}^K \sum_{j=1}^K \sum_{k=1}^{\widetilde{K}} \sum_{l=1}^{\widetilde{K}} \E \Big( \mathbf X^K[i,k] \mathbf X^K[j,\widetilde{k}] \mathbf X^K[i,l] \mathbf X^K[j,\widetilde{l}] \Big) \\
&=\sum_{i=1}^K \sum_{j=1}^K \sum_{k=1}^{\widetilde{K}} \sum_{l=1}^{\widetilde{K}} \E  \langle X, g_{i,k}^K \rangle \langle X, g_{j,\widetilde{k}}^K \rangle \langle X, g_{i,l}^K \rangle \langle X, g_{j,\widetilde{l}}^K \rangle \,,
\end{split}
\]
where we used that $\mathbf X^K[i,j] = K \langle X, g_{i,j}^K \rangle$ for the function $g_{i,j}$ defined in \eqref{eq:g_functions}. If we now denote $\Gamma = \E X \otimes X \otimes X \otimes X$, it follows from the outer product algebra (or can be verified explicitly using integral representations) that (recall that we denote $\widetilde{k} = k+\delta_K K$ and $\widetilde{l} = l+\delta_K K$)
\[
\begin{split}
\E  \langle X, g_{i,k}^K \rangle \langle X, g_{j,\widetilde{k}}^K \rangle \langle X, g_{i,l}^K \rangle \langle X, g_{j,\widetilde{l}}^K \rangle &= \E \langle X \otimes X \otimes X \otimes X, g_{i,k}^K \otimes g_{j,\widetilde{k}}^K \otimes g_{i,l}^K \otimes g_{j,\widetilde{l}}^K \rangle \\
&= \langle \Gamma, g_{i,k}^K \otimes g_{j,\widetilde{k}}^K \otimes g_{i,l}^K \otimes g_{j,\widetilde{l}}^K \rangle \\
&= \langle \Gamma(g_{i,k}^K \otimes g_{j,\widetilde{l}}^K),  g_{j,\widetilde{k}}^K \otimes g_{i,l}^K \rangle \,.
\end{split}
\]
Due to positive semi-definiteness of $\Gamma$, the last expression is bounded by
\[
\frac{1}{2} \Big[ \langle \Gamma(g_{i,k}^K \otimes g_{j,\widetilde{l}}^K), g_{i,k}^K \otimes g_{j,\widetilde{l}}^K \rangle + \langle \Gamma(g_{j,\widetilde{k}}^K \otimes g_{i,l}^K, g_{j,\widetilde{k}}^K \otimes g_{i,l}^K \rangle \Big]
\]
which gives us the bound
\[
(I) \leq \frac{1}{2} \sum_{i=1}^K \sum_{j=1}^K \sum_{k=1}^{(1-\delta_K)K} \sum_{l=1}^{(1-\delta_K)K} \Big[ \langle \Gamma(g_{i,k}^K \otimes g_{j,\widetilde{l}}^K), g_{i,k}^K \otimes g_{j,\widetilde{l}}^K \rangle + \langle \Gamma(g_{j,\widetilde{k}}^K \otimes g_{i,l}^K, g_{j,\widetilde{k}}^K \otimes g_{i,l}^K \rangle \Big] \,.
\]
Since $\Gamma$ is positive semi-definite, we can add terms into the bound to symmetrize it:
\[
(I) \leq \sum_{i=1}^K \sum_{j=1}^K \sum_{k=1}^K \sum_{l=1}^K \langle \Gamma(g_{i,k}^K \otimes g_{j,l}^K), g_{i,k}^K \otimes g_{j,l}^K \rangle \,.
\]
Finally, note that $\langle g_{i,j}^k, g_{k,l} \rangle = \mathds{1}_{[i=k,j=l]}$ for $i,j,k,l=1,\ldots,K$, hence $\{ g_{i,j}^K \}_{i,j=1}^K$ can be completed to an orthonormal basis of $\mathcal L^2[0,1]^2$ denoted as $\{ g_{i,j}^K \}_{i,j=1}^\infty$. We can add some more extra terms due to positive semi-definiteness of $\Gamma$ to obtain
\[
(I) \leq \sum_{i=1}^\infty \sum_{j=1}^\infty \sum_{k=1}^\infty \sum_{l=1}^\infty \langle \Gamma(g_{i,k}^K \otimes g_{j,l}^K), g_{i,k}^K \otimes g_{j,l}^K \rangle = \verti{\Gamma}_1 \,.
\]
Note that even though the orthonormal basis used changes with every $K$, the final equality holds for any orthonormal basis \cite[p. 114]{hsing2015}, and hence we obtain uniformity in $K$.

The strategy is similar for the remaining terms $(II)$, $(III)$ and $(IV)$. For the second one:
\begin{align*}
(II) &= K^{-4} \sum_{i=1}^K \sum_{j=1}^K \E \left| \sum_{k=1}^{\widetilde{K}} \mathbf E^K[i,k] \mathbf X^K[j,\widetilde{k}] \right|^2 = K^{-4} \sum_{i=1}^K \sum_{j=1}^K \sum_{k=1}^{\widetilde{K}} \E \big|\mathbf E^K[i,k] \big|^2 \E \big|\mathbf X^K[j,\widetilde{k}] \big|^2 \\
&= K^{-3} \sigma^2 \sum_{j=1}^K \sum_{k=1}^{\widetilde{K}} \E |\mathbf X^K[j,\widetilde{k}]|^2 = K^{-1}\sigma^2 \sum_{j=1}^K \sum_{k=1}^{\widetilde{K}} \E \langle X, g_{j,\widetilde{k}}^K \rangle^2 \\
&= K^{-1}\sigma^2 \sum_{j=1}^K \sum_{k=1}^{\widetilde{K}} \langle C(g_{j,\widetilde{k}}^K), g_{j,\widetilde{k}}^K \rangle^2 \leq K^{-1} \sigma^2 \verti{C}_1 \,.
\end{align*}
The third term can be treated exactly like the second one, and for the final term we have
\begin{align*}
(IV) &= K^{-4} \sum_{i=1}^K \sum_{j=1}^K \E \left| \sum_{k=1}^{\widetilde{K}} \mathbf E^K[i,k] \mathbf E^K[j,\widetilde{k}] \right|^2 \\
&= K^{-4} \sum_{i=1}^K \sum_{j=1}^K \sum_{k=1}^{\widetilde{K}} \sum_{l=1}^{\widetilde{K}} \E \Big( \mathbf E^K[i,k] \mathbf E^K[j,\widetilde{k}] \mathbf E^K[i,l] \mathbf E^K[j,\widetilde{l}] \Big) \\
&= K^{-4} \sum_{i=1}^K \sum_{j=1}^K \sum_{k=1}^{\widetilde{K}} \E \Big| \mathbf E^K[i,k] \Big|^2 \E \Big| \mathbf E^K[j,\widetilde{k}] \Big| \leq K^{-1} \sigma^4 \,,
\end{align*}
Piecing things together, we have
\[
\E \verti{\Tr{1}^{\delta_K}(\widehat{C}_N^K - C^K)}_2^2 \leq \frac{4}{N} \Big[ \verti{\Gamma}_1 + 2 K^{-1} \sigma^2 \verti{C}_1 + K^{-1} \sigma^4 \Big] \,.
\]

Thus we have shown that $\verti{\Tr{1}^{\delta_K}(\widehat{C}_N^K - C^K)}_2 = \mathcal{O}_P(N^{-1/2})$ uniformly in $K$, since $\sigma^2 = \O(\sqrt{K})$. It can be shown in an analogous way that $\verti{\Tr{2}^{\delta_K}(\widehat{C}_N^K - C^K)}_2 = \O_P(N^{-1/2})$ uniformly in $K$, and it remains to show the same for $\left|\Tr{}^{\delta_K}(\widehat{C}_N^K - C^K)\right|$.

Similarly to before we obtain the following bound:
\begin{align*}
\E \left|\Tr{}^{\delta_K}(\widehat{C}_N^K - C^K)\right|^2 &\leq \frac{4}{N} \Bigg\{ K^{-4} \E\left| \sum_{i=1}^{\widetilde{K}} \sum_{j=1}^{\widetilde{K}} \mathbf X^K[i,j] \mathbf X^K[\widetilde{i},\widetilde{j}] \right|^2 \\
&\qquad \quad+ K^{-4} \E\left| \sum_{i=1}^{\widetilde{K}} \sum_{j=1}^{\widetilde{K}} \mathbf E^K[i,j] \mathbf X^K[\widetilde{i},\widetilde{j}] \right|^2 \\
&\qquad \quad + K^{-4} \E\left| \sum_{i=1}^{\widetilde{K}} \sum_{j=1}^{\widetilde{K}} \mathbf X^K[i,j] \mathbf E^K[\widetilde{i},\widetilde{j}] \right|^2 \\
&\qquad \quad+ K^{-4} \E\left| \sum_{i=1}^{\widetilde{K}} \sum_{j=1}^{\widetilde{K}} \mathbf E^K[i,j] \mathbf E^K[\widetilde{i},\widetilde{j}] \right|^2 \Bigg\} 
\\ &=: \frac{4}{N} \Bigg\{ (I) + (II) + (III) + (IV) \Bigg\} \,,
\end{align*}
in which we will treat again the four terms separately.

For the first term:
\[
\begin{split}
(I) &= K^{-4} \sum_{i,j,k,l=1}^{\widetilde{K}} \E \Big( \mathbf X^K[i,j] \mathbf X^K[\widetilde{i},\widetilde{j}] \mathbf X^K[k,l] \mathbf X^K[\widetilde{k},\widetilde{l}] \Big) \\
&= \sum_{i,j,k,l=1}^{\widetilde{K}} \E \Big( \langle X, g_{i,j}^K \rangle \langle X, g_{\widetilde{i},\widetilde{j}}^K \rangle \langle X, g_{k,l}^K \rangle \langle X, g_{\widetilde{k},\widetilde{l}}^K \rangle \Big) \\
& = \sum_{i,j,k,l=1}^{\widetilde{K}} \E \langle X \otimes X \otimes X \otimes X , g_{i,j}^K \otimes g_{\widetilde{i},\widetilde{j}}^K \otimes g_{k,l}^K \otimes g_{\widetilde{k},\widetilde{l}}^K \rangle  \\
&= \sum_{i,j,k,l=1}^{\widetilde{K}} \langle \Gamma( g_{i,j}^K \otimes  g_{k,l}^K) , g_{\widetilde{i},\widetilde{j}}^K \otimes g_{\widetilde{k},\widetilde{l}}^K \rangle \leq \verti{\Gamma}_1 \,.
\end{split}
\]

For the second term,
\[
\begin{split}
(II) &= K^{-4} \sum_{i,j,k,l=1}^{\widetilde{K}} \E \Big( \mathbf X^K[i,j] \mathbf E^K[\widetilde{i},\widetilde{j}] \mathbf X^K[k,l] \mathbf E^K[\widetilde{k},\widetilde{l}] \Big) \\
&= K^{-4} \sum_{i,j,k,l=1}^{\widetilde{K}} \E \Big( \mathbf X^K[i,j] \mathbf X^K[k,l] \ \Big) \E \Big( \mathbf E^K[\widetilde{i},\widetilde{j}] \mathbf E^K[\widetilde{k},\widetilde{l}] \Big)
\end{split}
\]
and since $\E \big( \mathbf E^K[\widetilde{i},\widetilde{j}] \mathbf E^K[\widetilde{k},\widetilde{l}] \big) = \sigma^2 \mathds{1}_{[i=k,j=l]}$, we have
\[
(II) = \sigma^2 K^{-4} \sum_{i,j=1}^{\widetilde{K}} \E \Big| \mathbf X^K[i,j] \Big|^2 = \sigma^2 K^{-2} \sum_{i,j=1}^{\widetilde{K}} \E \langle X , g_{i,j}^K \rangle^2 \leq \sigma^2 K^{-2} \verti{C}_1 \,.
\]

The third term is bounded similarly, and for the final term:
\[
\begin{split}
(IV) &= K^{-4} \sum_{i,j,k,l=1}^{\widetilde{K}} \E \Big( \mathbf E^K[i,j] \mathbf E^K[\widetilde{i},\widetilde{j}] \mathbf E^K[k,l] \mathbf E^K[\widetilde{k},\widetilde{l}] \Big) \\
&= K^{-4} \sum_{i=1}^{\widetilde{K}} \sum_{j=1}^{\widetilde{K}} \E \big| \mathbf E^K[i,j] \big|^2 \E \big| \mathbf E^K[\widetilde{i},\widetilde{j}] \big|^2 \leq K^{-2} \sigma^4
\end{split}
\]

In summary, we have obtained the following bound:
\[
\E \left|\Tr{}^{\delta_K}(\widehat{C}_N^K - C^K)\right|^2 = \frac{4}{N} \Big[ \verti{\Gamma}_1 + 2 K^{-2} \sigma^2 \verti{C}_1 + K^{-2} \sigma^4 \Big] \,.
\]

The proof for the eigenvalues and eigenvectors remains the same as with sampling scheme S1.
\end{proof}

\subsubsection{Uniform Rates}\label{app:G5}

Next, we provide uniform rates of convergence in the following proposition. Note that compared to Theorem \ref{thm:rates}, noise variance is allowed to grow with the grid size at a slower rate, and Lipschitz assumption is put on the sample path. The latter is for the CLT to work 

\begin{theorem}\label{thm:rates_uniform}
Let $X_1, \ldots, X_N$ be i.i.d. copies of $X \in \mathcal{L}^2[0,1]^2$, which has (w.l.o.g.~mean zero and) covariance given by \eqref{eq:model}, where the the separable part $A := A_1 \ct A_2$ has kernel $a(t,s,t',s')$, which is Lipschitz continuous on $[0,1]^4$ with Lipshitz constant $L>0$.  Let $\E \| X \|^4 < \infty$ and $\delta \in [0,1)$ be such that $B$ from \eqref{eq:model} is banded by $\delta$ and $\Tr{}^\delta(A) \neq 0$. Let the samples come from \eqref{eq:discrete_noisy_observations} via measurement scheme (S1) or (S2) with $\var(\mathbf{E}_n^K[i,j]) \leq \sigma^2 = \mathcal{O}(\sqrt{K})$. Then we have
\[
\sup_{t,s,t',s' \in [0,1]}|\widehat{a}_1^K(t,t')\widehat{a}_2^K(s,s') - a_1(t,t') a_2(s,s')|  = \mathcal{O}_P(N^{-1/2}) + 2 K^{-1} L,
\]
where the $\mathcal{O}_P(N^{-1})$ term is uniform in $K$, for all $K \geq K_0$ for a certain $K_0 \in \N$.
\end{theorem}

\begin{proof}
The proof is similar to the one of Theorem \ref{thm:rates}. To save space, we will use the notation $\| \cdot \|_\infty$ for the uniform norm, i.e. $\| C \|_\infty := \sup_{t,s,t',s'} | c(t,s,t',s') |$. This is not to be confused with the operator norm of $C$ denoted as $\verti{C}_\infty$.

We begin with the triangle inequality separating the bias and the variance:
\[
\| \widehat{A}^K - A \|_\infty \leq \| \widehat{A}^K - A^K \|_\infty + \| A^K - A \|_\infty,
\]
and we bound the bias first.

Under (S1), we have
\[
\begin{split}
\| A^K - A \|_\infty &= \sup_{i,j,k,l=1}^K \sup_{(t,s,t',s') \in I_{i,j}^K \times I_{k,l}^K} \left| a_1^K(t_i,t_k) a_2^K(s_j,s_l) - a_1(t,t') a_2(s,s') \right| \\
&\leq \sup_{i,j,k,l=1}^K 2 L K^{-1} = 2 L K^{-1},
\end{split}
\]
where we used the Lipschitz property of $A$. On the other hand, under (S2), we have
\[
\begin{split}
&\| A^K - A \|_\infty =\\
&= \sup_{i,j,k,l=1}^K \sup_{(t,s,t',s') \in I_{i,j}^K \times I_{k,l}^K} \left| \frac{1}{| I_{i,j}^K |} \frac{1}{|I_{k,l}^K|} \int_{I_{i,j}^K \times I_{k,l}^K} \big[ a_1(u,v) a_2(x,y) -  a_1(t,s)a_2(t',s') \big] d u d v d x d y  \right|\\
&\leq \sup_{i,j,k,l=1}^K \sup_{(t,s,t',s') \in I_{i,j}^K \times I_{k,l}^K} K^4 \int_{I_{i,j}^K \times I_{k,l}^K} \big| a_1(u,v) a_2(x,y) - a_1(t,t') a_2(s,s') \big| d u d v d x d y \\
&\leq \sup_{i,j,k,l=1}^K \sup_{(t,s,t',s') \in I_{i,j}^K \times I_{k,l}^K} K^4 \int_{I_{i,j}^K \times I_{k,l}^K} 2LK^{-1} \leq 2LK^{-1}.
\end{split}
\]

Similarly to \eqref{eq:three_term_bound} in the proof of Theorem \ref{thm:rates}, we obtain
\[
\begin{split}
\| \widehat{A}^K - A^K \|_\infty &\leq \frac{\left\|\Tr{1}^{\delta_K}(\widehat{C}_N^K)\right\|_\infty}{\left|\Tr{}^{\delta_K}(\widehat{C}_N^K)\right|} \left\|\Tr{2}^{\delta_K}(\widehat{C}_N^K - C^K)\right\|_\infty \\
&\quad+\frac{\left\|\Tr{2}^{\delta_K}(C^K)\right\|_\infty}{\left|\Tr{}^{\delta_K}(\widehat{C}_N^K)\right|} \left\|\Tr{1}^{\delta_K}(\widehat{C}_N^K - C^K)\right\|_\infty \\
&\quad + \frac{\left\|A_1^K \ct A_2^K\right\|_\infty}{\left|\Tr{}^{\delta_K}(\widehat{C}_N^K)\right|} \left|\Tr{}^{\delta_K}(\widehat{C}_N^K - C^K)\right| ,
\end{split}
\]
and we will again show that the numerators and denominators are $\O_P(1)$, while the remaining terms are $\O_P(N^{-1/2})$. In fact, the term that has to be treated is $\left\|\Tr{1}^{\delta_K}(\widehat{C}_N^K - C^K)\right\|_\infty$. Once we show that this term is $\O_P(N^{-1/2})$, exactly the same arguments like in the proof of Theorem \ref{thm:rates} can be used to conclude.

We calculate
\[
\begin{split}
\E \verti{\Tr{1}^{\delta_K}(\widehat{C}_N^K - C^K)}_\infty^2 &= K^{-2} \E \left\| \Tr{1}^{d_K}(\widehat{\mathbf C}_N^K - \mathbf C^K)\right\|_\infty^2 \\
&= K^{-2} \sup_{i,j=1}^K \E \left| \Tr{1}^{d_K}(\widehat{\mathbf C}_N^K - \mathbf C^K)[i,j] \right|^2 \\
&= K^{-2} \sup_{i,j=1}^K \E \Bigg| \frac{1}{N} \sum_{n=1}^N \underbrace{\sum_{k=1}^{\widetilde{K}} \Big( \widetilde{\mathbf X}_n^K[i,k] \widetilde{\mathbf X}_n^K[j,\widetilde{k}] - \E \mathbf X_n^K[i,k] \mathbf X_n^K[i,\widetilde k] \Big)}_{=: Z_{n,ij}} \Bigg|^2 \,.
\end{split}
\]
Again, for a fixed $i,j$, $Z_{n,ij}$ is a set of mean-zero i.i.d. random variables and hence
\[
\E \verti{\Tr{1}^{\delta_K}(\widehat{C}_N^K - C^K)}_\infty^2 = N^{-1} K^{-2} \sup_{i,j=1}^K \E \left| Z_{n,ij} \right|^2,
\]
so it suffices to show that $K^{-2} \E \left| Z_{n,ij} \right|^2$ is uniformly bounded.

Under (S1), we proceed similarly as in \eqref{eq:four_term_bound}:
\[
\begin{split}
K^{-2} \E \left| Z_{n,ij} \right|^2 \leq K^{-2} &\Bigg\{
\E \left| \sum_{k=1}^{\widetilde{K}}  \mathbf X^K[i,k] \mathbf X^K[j,\widetilde{k}] - \E \mathbf X^K[i,k] \mathbf X^K[i,\widetilde k]  \right|^2 \\
& + \E \left| \sum_{k=1}^{\widetilde{K}}  \mathbf E^K[i,k] \mathbf X^K[j,\widetilde{k}]  \right|^2
 + \E \left| \sum_{k=1}^{\widetilde{K}}  \mathbf X^K[i,k] \mathbf E^K[j,\widetilde{k}]  \right|^2\\
 &+ \E \left| \sum_{k=1}^{\widetilde{K}}  \mathbf E^K[i,k] \mathbf E^K[j,\widetilde{k}]  \right|^2 \Bigg\}.
\end{split}
\]
The first term in the parentheses is bounded again by $S_1 K^2$, the second term is bounded by $K S_2 \sigma^2$, and the third term by $K \sigma^4$. Collecting the bounds together, we obtain under (S1) that
\[
\E \verti{\Tr{1}^{\delta_K}(\widehat{C}_N^K - C^K)}_\infty^2 \leq \frac{4}{N}\big[ S_1 + S_2 K^{-1} \sigma^2 + K^{-1} \sigma^4 \big],
\]
from which the claim of follows.

Under (S2), the proof is an equivalent modification to the proof of Theorem \ref{thm:rates}.
\end{proof}

While the rate for the eigenvalues remain the same as in Theorem \ref{thm:rates} (simply because eigenvalues are just numbers), uniform rates for the eigenfunctions are a bit trickier. The simple perturbations bound cannot be used anymore, and it does not seem possible to separate the effect of the grid size from the effect of the sample size. But if we assume e.g. that $K \asymp \sqrt{N}$, the corresponding rate holds for the eigenfunctions as well.

Again, let $\widehat{A}_1^K = \sum_{j \in \N} \widehat{\lambda}^K_j \widehat{e}_j^K \otimes \widehat{e}_j^K$, $\widehat{A}_2^K = \sum_{j \in \N} \widehat{\rho}^K_j \widehat{f}_j^K \otimes \widehat{f}_j^K$, $A_2 = \sum_{j \in \N} \lambda_j e_j \otimes e_j$, and $A_2 = \sum_{j \in \N} \rho_j f_j \otimes f_j$ be eigendecompositions. We show below that for $j$ such that the corresponding eigensubspace is one-dimensional it holds
\[
\| \widehat{e}_j^K - \sign(\langle \widehat{e}_j^K,e_j \rangle) e_j \|_\infty = \O_P(N^{-1/2}) \quad \& \quad \| \widehat{f}_j^K - \sign(\langle \widehat{f}_j^K,f_j \rangle) f_j \|_\infty = \O_P(N^{-1/2}).
\]

For simplicity, we assume that the signs are estimated correctly. Firstly, note that we have from the triangle inequality followed by the Cauchy-Schwartz inequality that
\[
\begin{split}
\Big| \widehat{\alpha}_j^K \widehat{\beta}_1^K \widehat{e}_j^K &- \alpha_j \beta_1 e_j \Big| = \\
&= 
\Bigg| \int \int \int \widehat{a}^K(t,s,t',s') \widehat{e}_j^K(t') \widehat{f}_1^K(s) \widehat{f}_1^K(s') d t d s d t' d s' \\
&\qquad - \int \int \int a(t,s,t',s') e_j(t') f_1(s) f_1(s') d t d s d t' d s' \Bigg| \\
&\leq \Bigg| \int \int \int \big[ \widehat{a}^K(t,s,t',s') - a(t,s,t',s')\big] e_j(t') f_1(s) f_1(s') d t d s d t' d s' \Bigg| \\
&\quad + \Bigg| \int \int \int a(t,s,t',s') \big[ \widehat{e}_j^K(t') \widehat{f}_1^K(s) \widehat{f}_1^K(s') - e_j(t') f_1(s) f_1(s') \big] d t d s d t' d s' \Bigg|\\
&\leq \| \widehat{A}^K - A \|_\infty \cdot 1 + \| A \|_\infty \| \widehat{e}_j^K \otimes \widehat{f}_1 \otimes \widehat{f}_1 - e_j \otimes f_1 \otimes f_1  \|_2.
\end{split}
\]
The first term on the previous line is bounded by the previous theorem, while the second term is bounded from the triangle inequality and the $\mathcal{L}^2$ rate for the product eigenfunctions given Theorem \ref{thm:rates} (the rate for the product eigenfunctions can be found in the proof of Theorem \ref{thm:rates}):
\[
\| \widehat{e}_j^K \otimes \widehat{f}_1 \otimes \widehat{f}_1 - e_j \otimes f_1 \otimes f_1  \|_2 \leq \|\widehat{f}_1 - f \|_2 + \| \widehat{e}_j^K \otimes \widehat{f}_1 - e_j \otimes f_1 \| = \mathcal{O}_P(N^{-1/2}) + 2 K^{-1} L.
\]

Now, from the triangle inequality, we have for any $t$ that
\[
\widehat{\alpha}_j^K \widehat{\beta}_1^K \Big| \widehat{e}_j^K(t) - e_j(t) \Big| \leq \Big| \widehat{\alpha}_j^K \widehat{\beta}_1^K \widehat{e}_j^K - \alpha_j \beta_1 e_j \Big| + \Big| \widehat{\alpha}_j^K \widehat{\beta}_1^K - \alpha_j \beta_1 \Big| \| e_j \|_\infty .
\]
Since $e_j$ is continuous and $\widehat{\alpha}_j^K \widehat{\beta}_1^K$ converges to $alpha_j \beta_1 \neq 0$, we obtain the rates for $\widehat{e}_j^K$. The rates for $\widehat{f}_j^K$ are obtained similarly.

\subsubsection{Discrete Rates under Adaptive Bandwidth}\label{app:G6}

And finally, we discuss what happens to the rates in Theorem \ref{thm:rates} when the bandwidth is chosen adaptively. There are several difficulties that need to be addressed in this case.

Firstly, the candidate values for $\delta$ can now depend on the grid size as described in the main paper, so we should denote the set of candidate values as $\Delta^K$. However, let us assume for simplicity that the candidate values do not depend on the grid size for $K \geq K_0$, i.e. starting from some critical resolutions. If this is not true, one needs to take a similar care like in the proof of Theorem \ref{thm:rates}. A natural question arises whether it is reasonable to have a fixed set of candidate values $\Delta$, should we not allow the number of candidate bandwidth values increase with increasing grid size $K$? The answer is negative simply because there is a whole range of equally good candidate values (large enough to eliminate the banded part) among which to pick. This range does not depend on $K$, we only need $K$ large enough such that at least one candidate discrete bandwidth falls inside this range. At the same time, $\Delta$ needs to contain this suitable candidate. However, this is always satisfied for a finite grid size $K$ and a finite cardinality of $\Delta$. For example, when when $\Tr{}^\delta(C) \neq 0$ for all $\delta \in (0,1)$ and the true bandwidth $\delta^\star$ is smaller than 0.5, then for an equidistant grid of size $K \geq 2$ it is enough to choose $\Delta = \{ 1/3, 2/3 \}$. Of course, in practice, $\delta^\star$ is unknown, usually much smaller, and we would like to approximate it more closely, so we choose a larger set of candidate values $\Delta$. However, it is clear that the cardinality of $\Delta$ should not depend on the grid size $K$.

Secondly, and more importantly, while Theorem \ref{thm:rates} establishes that our estimation methodology is robust against noise, this is not the case for the bandwidth selection procedure of Section \ref{sec:delta}. Here, we will change the bandwidth selection procedure to one that is robust against noise. We should, however, note that the new bandwidth selection procedure should rarely be used in practice. This goes back to whether we see a banded part of the model as a nuisance or as a signal to be estimated. For example, in the case of the mortality data analysis, one can either use (the discrete version of) the bandwidth selection procedure of Section \ref{sec:delta} to choose $\delta_D$ among $\{ 0, 1, 2, \ldots \}$ for the discrete bandwidths, and when $\delta_D=1$ is chosen, one takes $B$ to be a diagonal structure corresponding to heteroscedastic white noise, i.e. one directly models the noise structure. Alternatively, we could use the bandwidth selection procedure below to choose $\delta_D$ among $\{ 1, 2, \ldots \}$, and when $\delta_D=1$ would be chosen, we would obtain the same estimators for the separable part, but we would not estimate the banded part of the model; there would be no banded part in this case. Moreover, we should point out that in our mortality data analysis, we actually have a reason to believe that $\delta_D=1$ should be used, and then we test the validity of the separable-plus-diagonal model, so the bandwidth choice on this data set should rather be taken as an illustration of the bandwidth selection procedure. On the other hand, the development below can be taken simply as a complementary evidence that we truly obtain the correct estimators of the separable part, even under discrete noisy measurements and when the bandwidth is unknown.

For $K \in \N$, $\mathbf F, \mathbf G \in \R^{K \times K \times K \times K}$, and $F^K, G^K \in \mathcal{S}_2(L^2[0,1]^2)$ the piece-wise constant continuations of $\mathbf F$ and $\mathbf{G}$, respectively, we define $\| F^K \|_\star$ via
\[
\verti{F^K}_\star^2 = \verti{ F^K }_2^2 - \frac{1}{K^2}\| \diag(\mathbf F) \|_2^2.
\]
We also define $\langle \cdot , \cdot \rangle_\star$ as
\[
\langle F^K, G^K \rangle_\star = \langle F^K, G^K \rangle - \frac{1}{K^2}\langle \diag(\mathbf F) , \diag(\mathbf G) \rangle .
\]
Finally, recall that $\widehat{X}_n^K$ are the discrete noisy samples (or  rather piece-wise constant continuation thereof), and define
\[
\Xi^K(\delta) := \verti{C^K(\delta) - C^K}_\star^2 \quad \& \quad \widehat{\Xi}^K(\delta) := \verti{\widehat{C}^K(\delta)}_\star^2 - \frac{2}{N} \sum_{n=1}^N \langle \widetilde{X}_n^K, \widehat{C}^K_{-n}(\delta) \widetilde{X}_n^K \rangle_\star,
\]
and
\begin{equation}\label{eq:delta_choice_discrete}
\widehat{\delta} := \argmin_{\delta \in \Delta} \widehat{\Xi}^K(\delta) \quad \& \quad \delta_\star := \argmin_{\delta \in \Delta} \Xi^K(\delta).
\end{equation}
With these definitions, we are trying to bypass the effect of noise on the bandwidth selection procedure, to obtain an adaptive version of Theorem \ref{thm:rates}. Since $\| \cdot \|_\star$ is clearly a semi-norm and $\langle \cdot , \cdot \rangle_\star$ is the corresponding semi-inner-product \cite{conway2019}, we will be able to combine the continuous-domain proof of Theorem \ref{thm:rates_adaptive} with the discrete-domain proof of Theorem \ref{thm:rates} to obtain the following result.

\begin{theorem}\label{thm:rates_adaptive_discrete}
Let $X_1, \ldots, X_N$ be i.i.d. copies of $X \in \mathcal{L}^2[0,1]^2$, which has (w.l.o.g.~mean zero and) covariance given by \eqref{eq:model}, where the the separable part $A := A_1 \ct A_2$ has kernel $a(t,s,t',s')$, which is Lipschitz continuous on $[0,1]^4$ with Lipshitz constant $L>0$. Let $\E \| X \|^4 < \infty$ and $\delta^\star \in [0,1)$ be such that $B$ from \eqref{eq:model} is banded by $\delta^\star$. Let $\Delta$ be such that $\Tr{}^{\delta}(A) \neq 0$ for all $\delta \in \Delta$ of which at least one is larger than $\delta^\star$, and let $\widehat{\delta}$ be chosen from $\Delta$ as in \eqref{eq:delta_choice_discrete}. Let the samples come from \eqref{eq:discrete_noisy_observations} via measurement scheme (S1) or (S2) with $\var(\mathbf{E}_n^K[i,j]) \leq \sigma^2$. Then we have
\begin{equation}\label{eq:rates_discrete}
\verti{\widehat{A}_1^K(\widehat{\delta}) \ct \widehat{A}_2^K(\widehat{\delta}) - A_1 \ct A_2}_2^2  = \mathcal{O}_P(N^{-1}) + 2 K^{-2} L^2,
\end{equation}
where the $\mathcal{O}_P(N^{-1})$ term is uniform in $K$, for all $K \geq K_0$ for a certain $K_0 \in \N$.
\end{theorem}

\begin{proof}
We begin by proving several claims. All of the four claims below hold uniformly in $K$ for any $\delta$ such that $\Tr{}^{\delta}(A) \neq 0$, and are proven sequentially.
\begin{description}
\item[Claim 1:] $\verti{\widehat{C}^K_N - C^K}_\star^2 = \O_P(N^{-1})$

Similarly to the proof of Theorem \ref{thm:rates}, we calculate
\[
\E \verti{\widehat{C}^K_N - C^K}_\star^2 = \frac{1}{K^4} \sum_{(i,j) \neq (k,l)} \E \Big| \frac{1}{N} \sum_{n=1}^N \big( \underbrace{\widetilde{\mathbf X}_n^K[i,j] \widetilde{\mathbf X}_n^K[k,l] - \E \mathbf X_n^K[i,j] \mathbf X_n^K[k,l]}_{=: Z_{n,ijkl}}  \big) \Big|^2 .
\]
For a fixed $i,j,k,l$, $Z_{n,ijkl}$ are zero-mean (this is the reason why we need to remove the diagonal from the norm) i.i.d. random variables and thus
\[
\E \verti{\widehat{C}^K_N - C^K}_\star^2 =\frac{1}{N} \frac{1}{K^4} \sum_{(i,j) \neq (k,l)} \E \left| Z_{n,ijkl} \right|^2,
\]
which is from the parallelogram law equal to
\[
\begin{split}
&\frac{4}{N} \frac{1}{K^4} \sum_{(i,j) \neq (k,l)} \Bigg\{
\E \left| \mathbf X^K[i,j] \mathbf X^K[k,l] - \E \mathbf X^K[i,j] \mathbf X^K[k,l]  \right|^2 \\
&\qquad\quad + \E \left| \mathbf E^K[i,j] \mathbf X^K[k,l]  \right|^2
 + \E \left|  \mathbf X^K[i,j] \mathbf E^K[k,l]  \right|^2
 + \E \left|  \mathbf E^K[i,j] \mathbf E^K[k,l]  \right|^2 \Bigg\} \,.
\end{split}
\]
The first term in the parentheses is bounded by $S_1$, the second and third are bounded by $\sigma^2 S_2$, and the final term is bounded by $\sigma^4$, which yields the claim.

\item[Claim 2:] $\verti{\widehat{A}^K(\delta) - A^K(\delta)}_2^2 = \O_P(N^{-1})$

Here we use the linearization argument \eqref{eq:linearization} and proceed exactly like in Theorem \ref{thm:rates}.

\item[Claim 3:] $\verti{\widehat{C}^K(\delta) - C^K(\delta)}_\star^2 = \O_P(N^{-1})$

We have
\[
\begin{split}
\verti{\widehat{C}^K(\delta) - C^K(\delta)}_\star &\leq \verti{\widehat{A}^K(\delta) - A^K(\delta)}_\star + \verti{\topavg \left(\widehat{C}_N^K - \widehat{A}^K(\delta) - C^K + A^K(\delta) \right)}_\star \\
&\leq \verti{\widehat{A}^K(\delta) - A^K(\delta)}_\star + \verti{\widehat{C}_N^K - \widehat{A}^K(\delta) - C^K + A^K(\delta)}_\star \\
&\leq 2 \verti{\widehat{A}^K(\delta) - A^K(\delta)}_\star + \verti{\widehat{C}_N^K - C^K}_\star,
\end{split}
\]
where we utilized the triangle inequality in the first and last inequality, while the second inequality follows from $\topavg(\cdot)$ being a linear projection. Now, the second term is bounded by Claim 2, while the first term is bounded by Claim 1 and the fact that
\[
\verti{\widehat{A}^K(\delta) - A^K(\delta)}_\star \leq \verti{\widehat{A}^K(\delta) - A^K(\delta)}_2 .
\]

\item[Claim 4:] $\widehat{\Xi}^K(\delta) = \Xi^K(\delta) - \verti{C^K}_\star^2 + \O_P(N^{-1/2})$

We first work with a biased version of the empirical objective $\widehat{\Xi}^K$, i.e.
\[
\widetilde{\Xi}^K(\delta) = \verti{\widehat{C}^K(\delta)}_\star^2 - \frac{2}{N} \sum_{n=1}^N \langle X_n^K, \widehat{C}^K(\delta) X_n^K \rangle_\star,
\]
and show Claim 4 with $\widehat{\Xi}^K$ replaced by $\widetilde{\Xi}^K$. For this, we bound similarly to the proof of Proposition \ref{prop:obj_consistency}:
\[
\begin{split}
\left| \widetilde{\Xi}^K(\delta) + \verti{C^K}_\star^2 - \Xi^K(\delta) \right| &\leq 2 \left| \langle \widehat{C}^K(\delta), \widehat{C}_N^K - C \rangle_\star \right| \\
&\kern5ex + 
\left| \verti{\widehat{C}^K(\delta) - C^K }_\star - \verti{C^K(\delta) - C^K }_\star \right|
\end{split}
\]
The Cauchy-Schwartz inequality still holds for the semi-inner-product \cite{conway2019}, which allows us to bound the first term using Claim 1. For the second term, note that the mean value theorem can still be used, since the Fr\'{e}chet derivative is a linear operation and the semi-norm is \emph{consistent} with the semi-inner-product. Hence using the mean value theorem, and the Cauchy-Schwartz inequality, we have
\[
\begin{split}
\left| \verti{\widehat{C}^K(\delta) - C^K }_\star - \verti{\widehat{C}^K(\delta) - C^K }_\star \right| &= 2 \langle \Gamma - C^K, \widehat{C}^K(\delta) - C^K(\delta) \rangle_\star \\&\leq 2 \verti{ \Gamma - C^K}_\star \verti{ \widehat{C}^K(\delta) - C^K(\delta) }_\star.
\end{split}
\]
Hence the bound follows from Claim 3.

It remains to show that the introduced bias is asymptotically negligible, i.e.~to bound $|\widehat{\Xi}^K(\delta) - \widetilde{\Xi}^K(\delta)|$. For this, we use triangle and Cauchy-Schwartz inequality:
\[
\begin{split}
\left| \frac{1}{N} \sum_{n=1}^N \langle \widehat{C}(\delta) - \widehat{C}_{-n}(\delta), \widetilde{X}_n \otimes \widetilde{X}_n \rangle_\star \right| &\leq \frac{1}{N} \sum_{n=1}^N \left| \langle \widehat{C}(\delta) - \widehat{C}_{-n}(\delta), \widetilde{X}_n \otimes \widetilde{X}_n \rangle_\star \right| \\
&\leq \frac{1}{N} \sum_{n=1}^N \verti{\widehat{C}(\delta) - \widehat{C}_{-n}(\delta)}_\star \| \widetilde{X}_n \otimes \widetilde{X}_n \|_\star.
\end{split}
\]
Now, since $\verti{\cdot}_\star \leq \verti{\cdot}_2$, the remainder of the proof is exactly the same as the end of the proof of Proposition \ref{prop:obj_consistency}.
\end{description}

Now we can prove the Theorem itself. Using the parallelogram law, we have
\[
\verti{\widehat{A}^K(\widehat{\delta}) - A}_2^2 \leq 4 \left[ \verti{\widehat{A}^K(\widehat{\delta}) - A^K(\widehat{\delta})}_2^2 + \verti{A^K(\widehat{\delta}) - A^K}_2^2 + \verti{A^K - A}_2^2 \right].
\]
The first term in the brackets is bounded by Claim 2, while the last term in the brackets correspond to the bias and can be treated the same as in the proof of Theorem \ref{thm:rates}. It remains to show that the middle term in the brackets is equal to zero for all sufficiently large $N$. But this the same way as the first paragraph of the proof of Theorem \ref{thm:rates_adaptive}.
\end{proof}

The perturbation bounds \cite{bosq2012} yield again the adaptive rates for the eigenvalues and eigenfunctions, just as in the proof of Theorem \ref{thm:rates}.

\subsection{Simulation Study: Details and Additional Results}\label{app:H}

As described in the main paper, we generate data for the simulation study as a superposition of two independent processes, one with a separable covariance $\mathbf A_1 \ct \mathbf A_2$ and other with a banded covariance $\mathbf B$. For the separable part, we set in the main paper both $\mathbf A_1$ and $\mathbf A_2$ as rank-7 covariances with linearly decaying eigenvalues and shifted Legendre polynomials, resulting in a covariance depicted in Figure \ref{fig:setup_separable} (left). Here in the appendices, we show additional results, where $\mathbf A_1$ and $\mathbf A_2$ are set as the covariance of Wiener process depicted in Figure \ref{fig:setup_separable} (right). As will be explained later, the Wiener case is simpler than the Legendre case, because the Wiener covariance decays slower away from the diagonal.

\begin{figure}[!t]
   \centering
   \begin{tabular}{ccc}
   \includegraphics[width=0.4\textwidth]{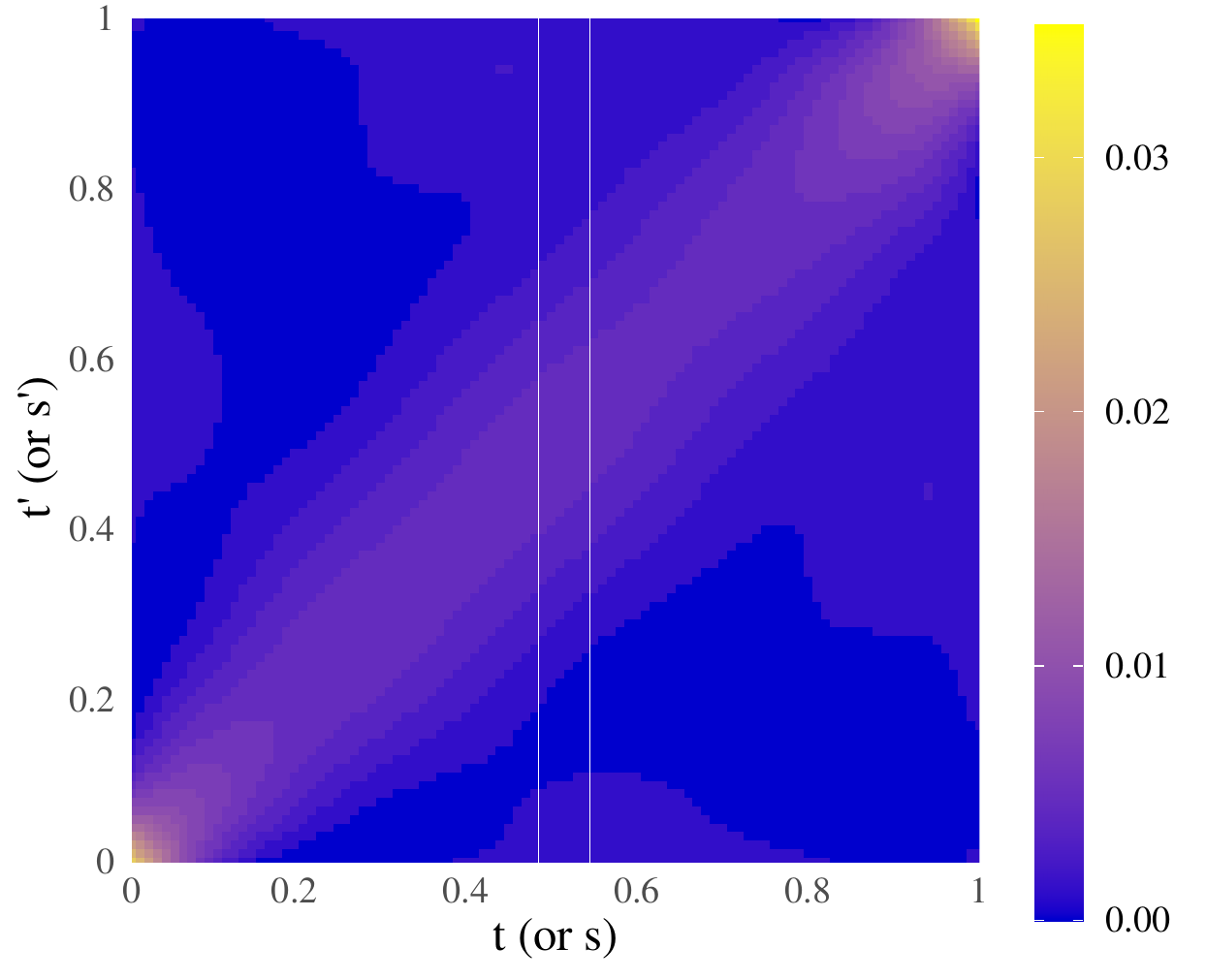} &
   \includegraphics[width=0.4\textwidth]{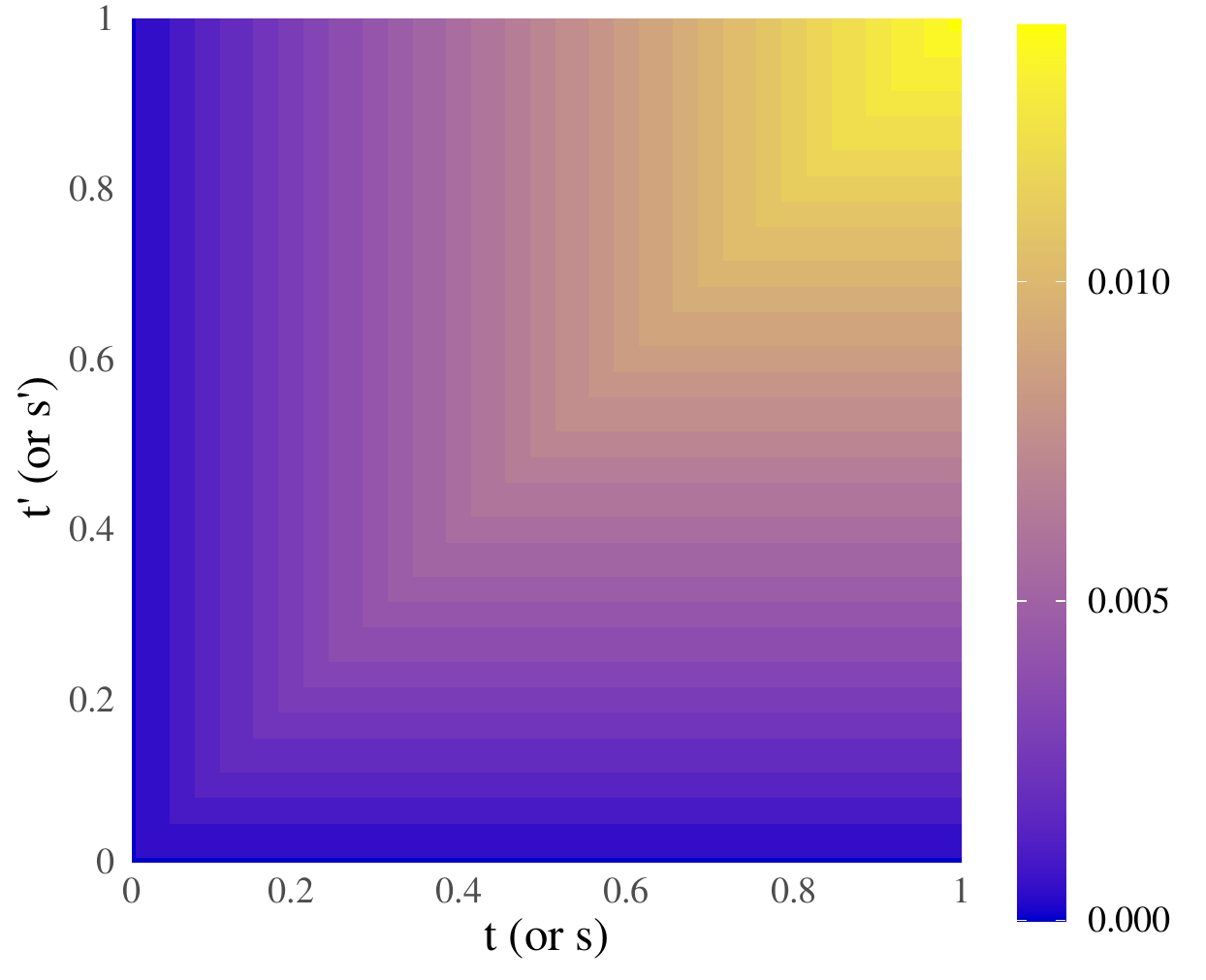}
   \end{tabular}  
   \caption{The two choices for the separable constituents of the separable-plus-banded model: the Legendre covariance (\emph{left}) and the Wiener covariance (\emph{right}).}
    \label{fig:setup_separable} 
\end{figure}

\begin{figure}[!b]
   \centering
   \begin{tabular}{ccc}
   \includegraphics[width=0.4\textwidth]{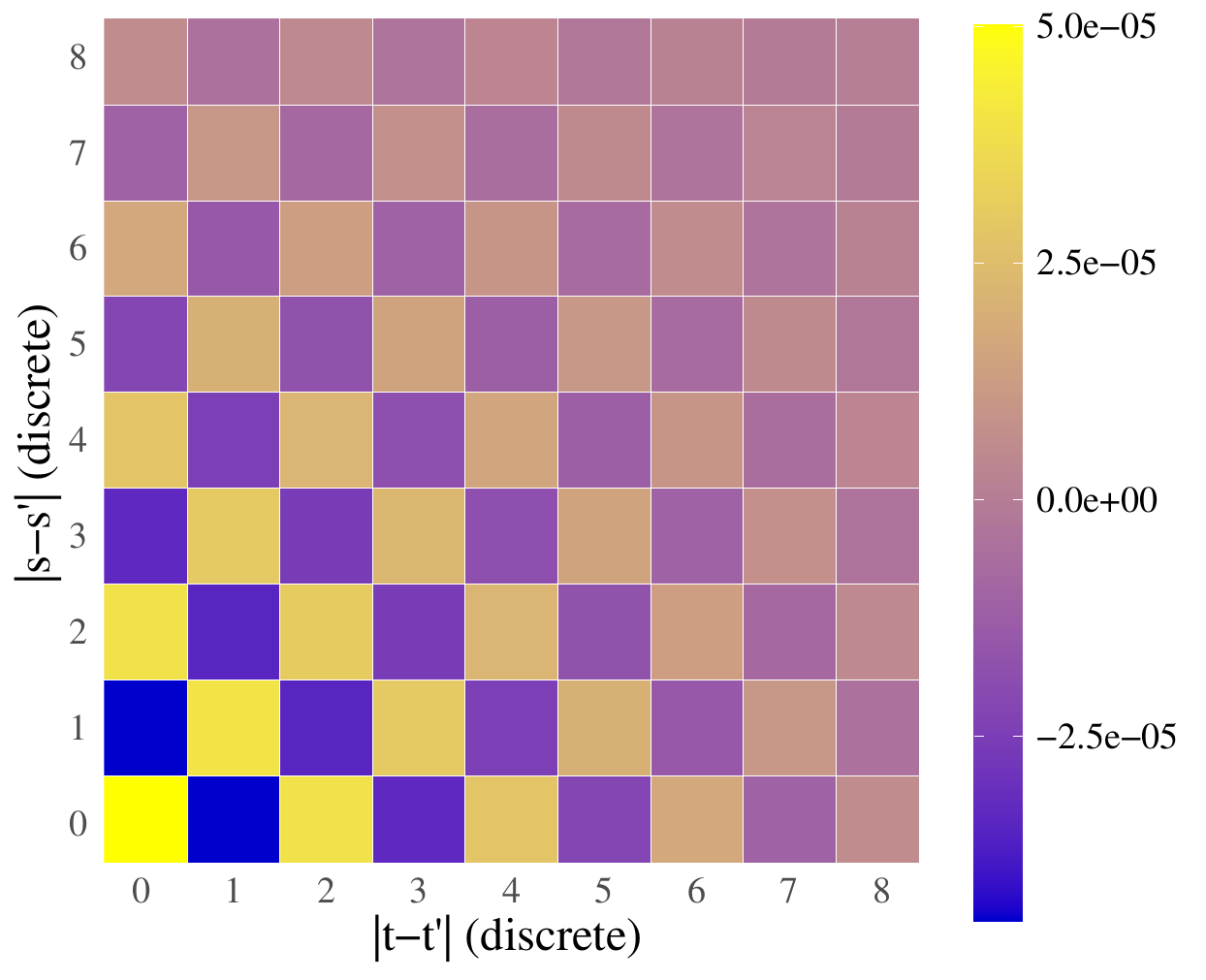} &
   \includegraphics[width=0.4\textwidth]{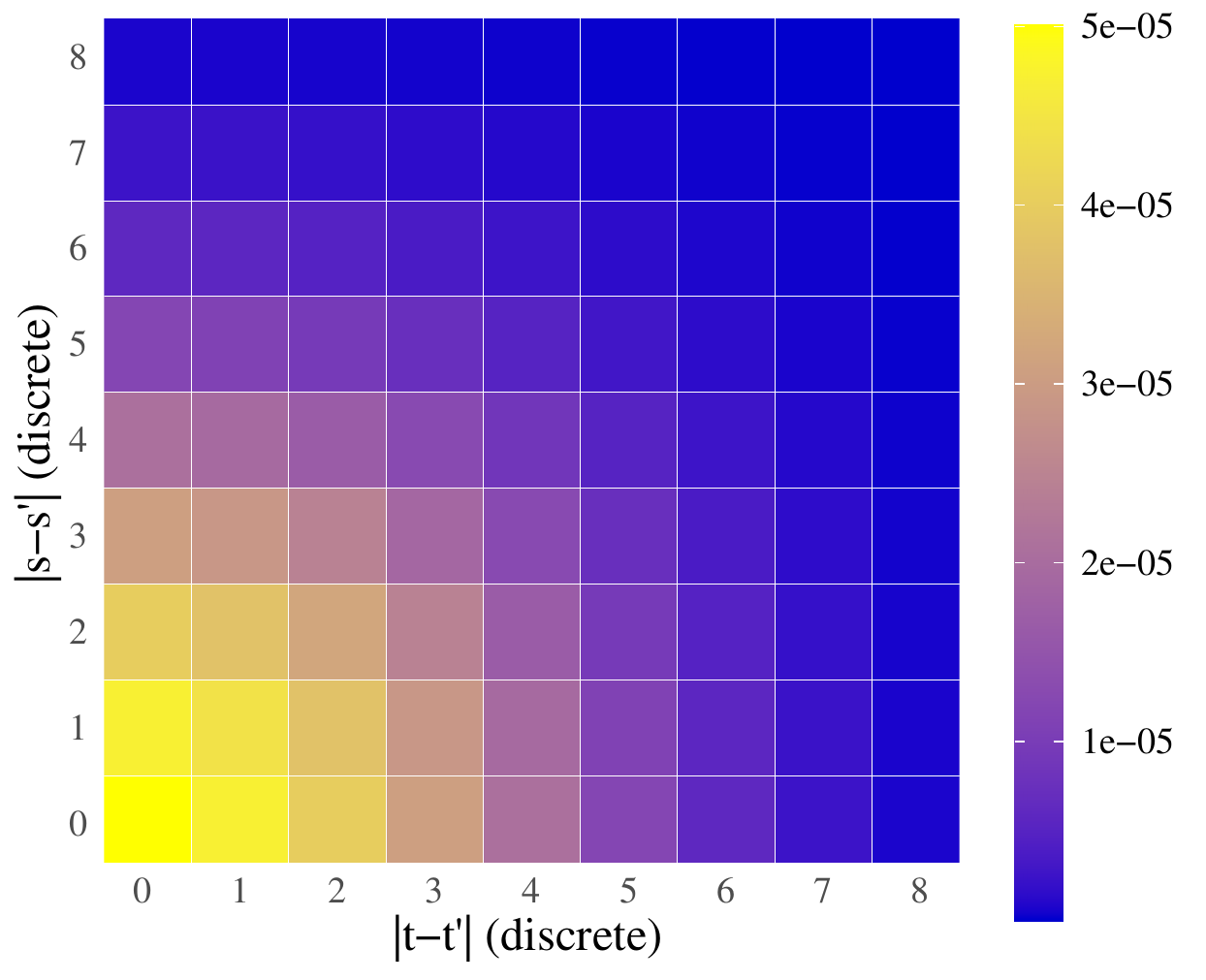}
   \end{tabular}  
   \caption{The two choices for symbol of the banded part of the separable-plus-banded model: the signed case (\emph{left}) and the Epanechnikov case (\emph{right}).}
    \label{fig:setup_banded} 
\end{figure}

The banded process $\mathbf W \in \R^{K \times K}$ is created by space-time averaging of white noise entries. For a fixed odd bandwidth $d = 2p + 1 \in \N$ and grid size $K$, let $\epsilon_{k,l}$, $k,l = 1 - d, \ldots, K + d$ be i.i.d.~$\mathcal{N}(0,1)$ entries. Then we set $\mathbf W[i,j] = \sum_{k=i-p}^{i+p} \sum_{l=j-p}^{j+p} q_{k,l} \epsilon_{k,l}$ for every $i,j=1,\ldots,K$. Here, $\mathbf Q = \big(q_{k,l} \big) \in \R^{d \times d}$ is the averaging filter. Regardless of how the filter is chosen, the resulting covariance of $\mathbf W$ will be stationary and banded by $d$. We choose the filter as either $q_{k,l} = (-1)^{|k-l|}$, leading to $\mathbf B$ with its symbol depicted in Figure \ref{fig:setup_banded} (left), or $q_{k,l} = \frac{9}{16} \left(1 - \frac{|k|}{p+1} \right) \left(1- \frac{|l|}{p+1}\right)$, i.e. the outer product of two Epanechnikov kernels, leading to $\mathbf B$ with its symbol depicted in Figure \ref{fig:setup_banded} (right). Again, the Epanechnikov case will turn out to be easier compared to the other choice of the filter (called the \emph{signed} case).

We have shown results for the Legendre-signed case (i.e. Legendre covariance as the separable part of the process and the signed covariance for the banded part of the process) in Figure \ref{fig:buf} of the main paper. We show results for the remaining scenarios (Legendre-Epanechnikov at the top, Wiener-signed in the middle, and Wiener-Epanechnikov at the bottom) in Figure \ref{fig:buf2}. All the scenarios exhibit qualitatively similar behavior, which is described in the main paper, and the quantitative differences can be attributed to different shapes of the underlying covariances.

\begin{figure}[!t]
   \advance\leftskip-0.3cm
   \begin{tabular}{ccc}
   \includegraphics[width=0.32\textwidth]{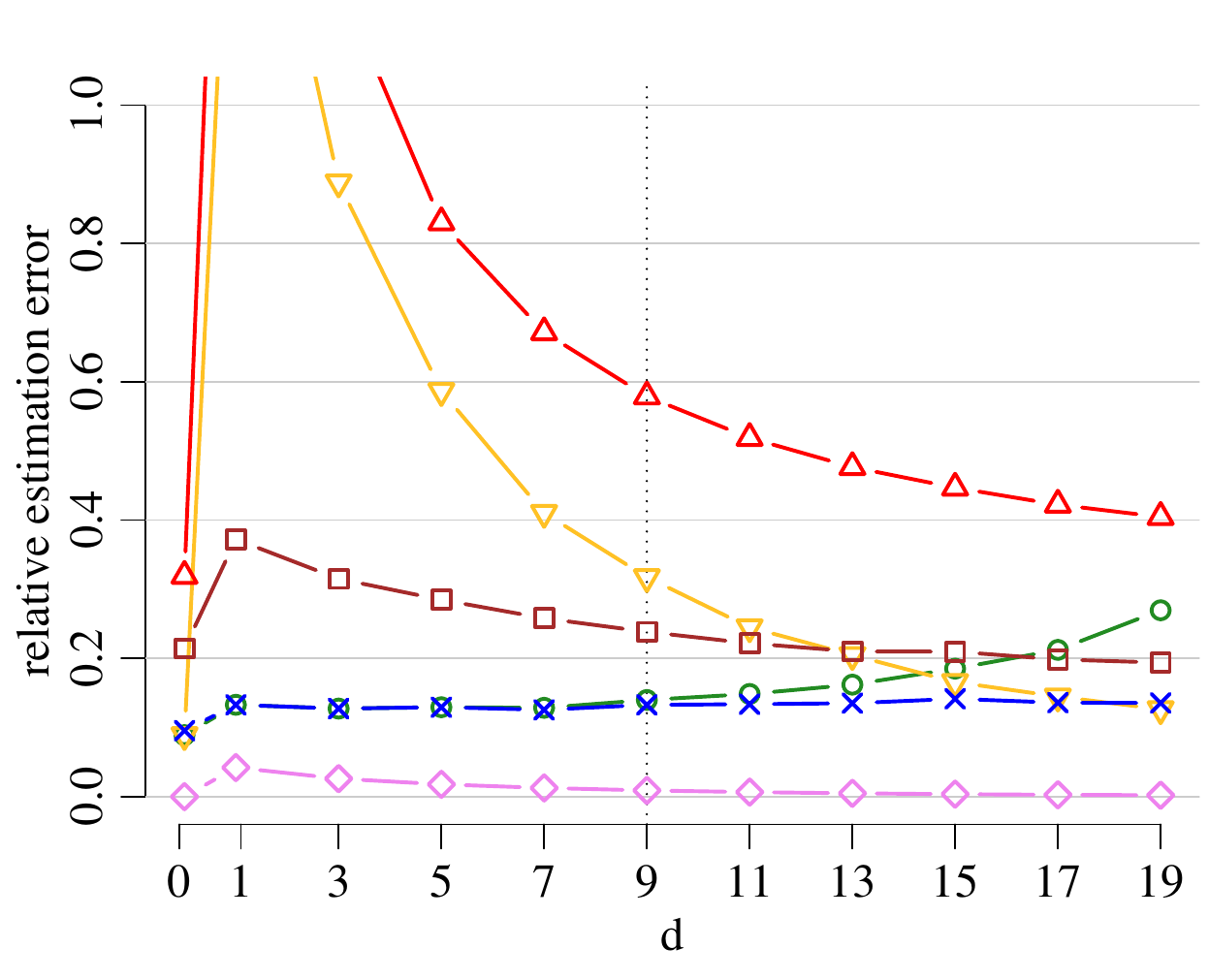} &
   \includegraphics[width=0.32\textwidth]{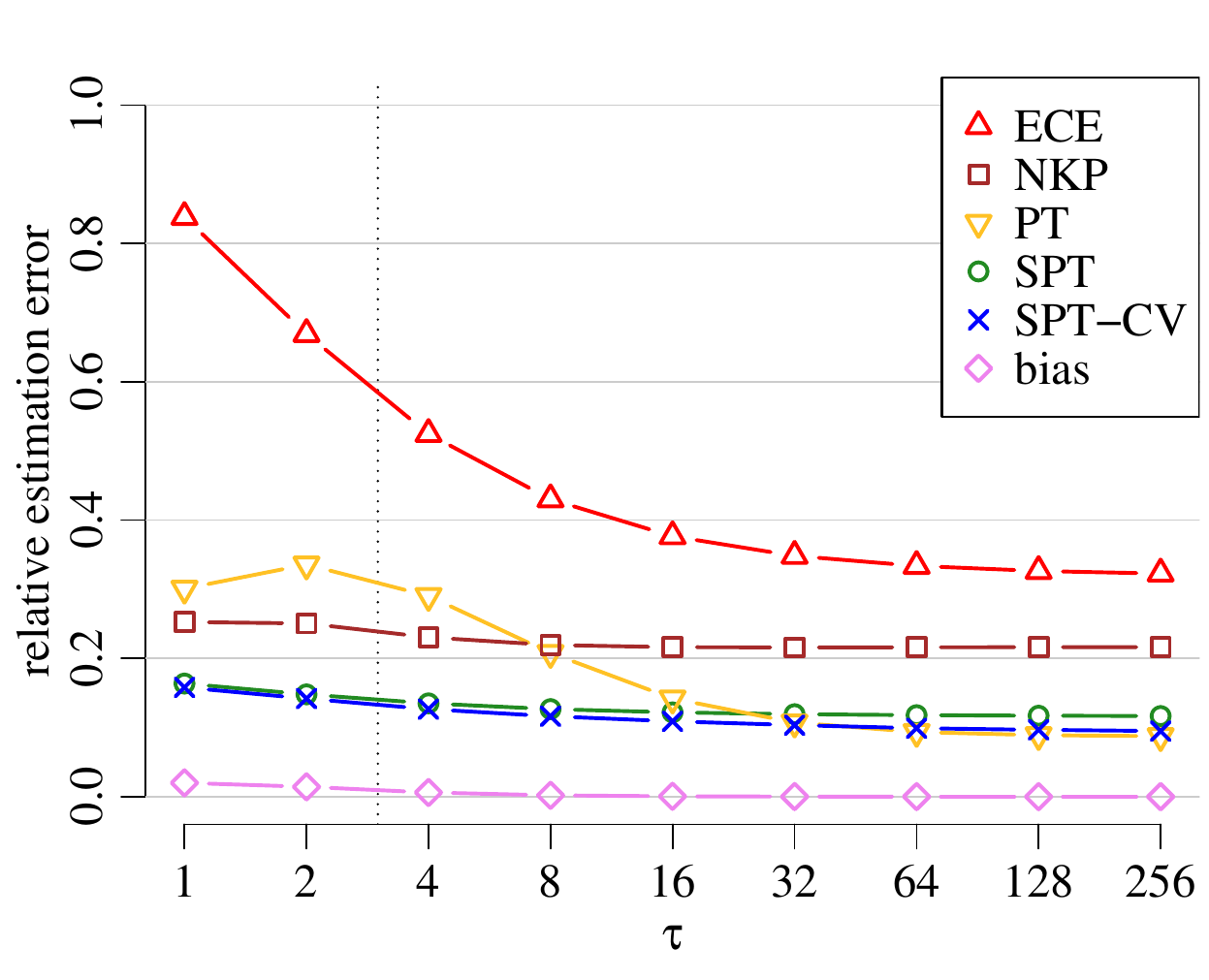} &
   \includegraphics[width=0.32\textwidth]{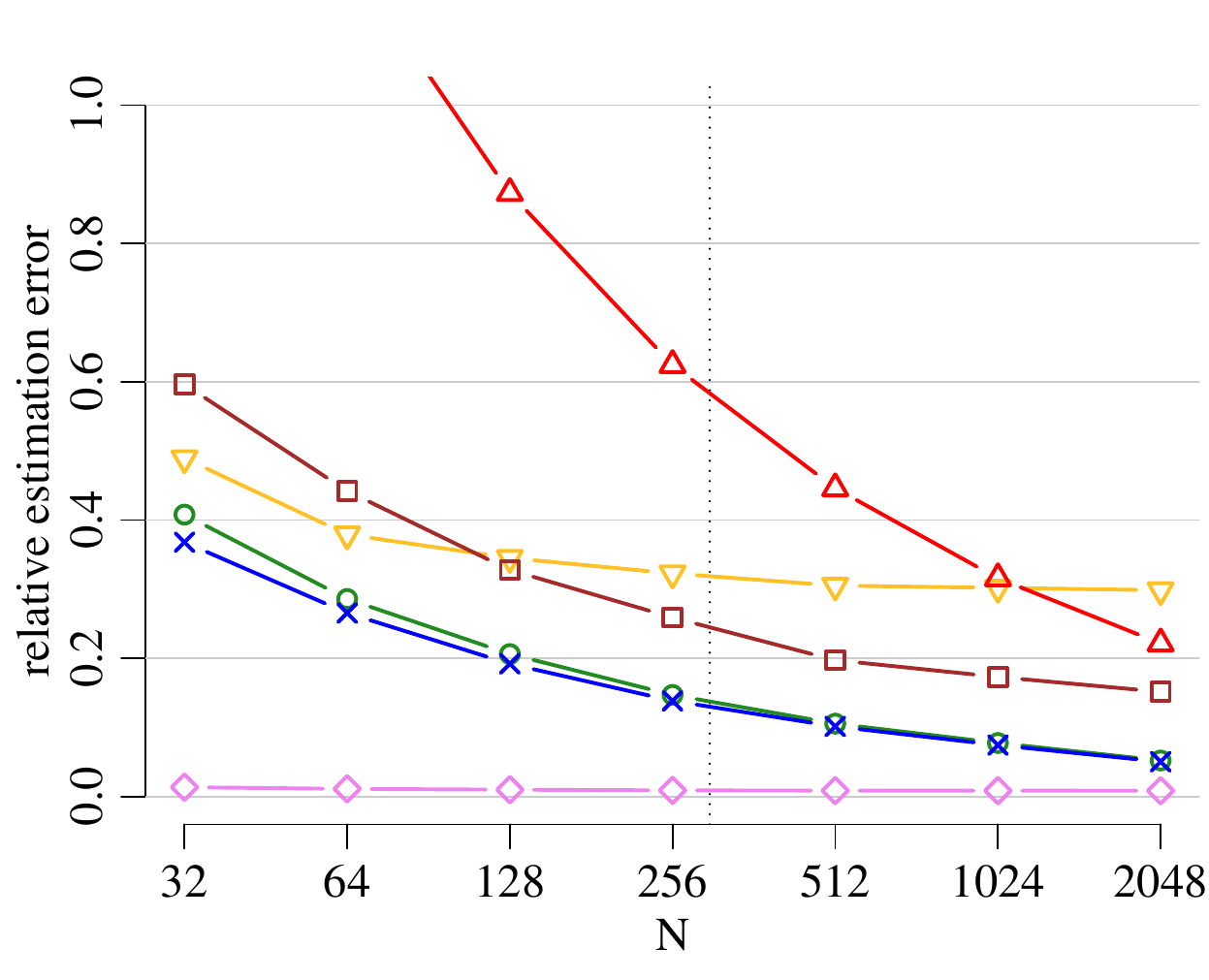} \\
   \includegraphics[width=0.32\textwidth]{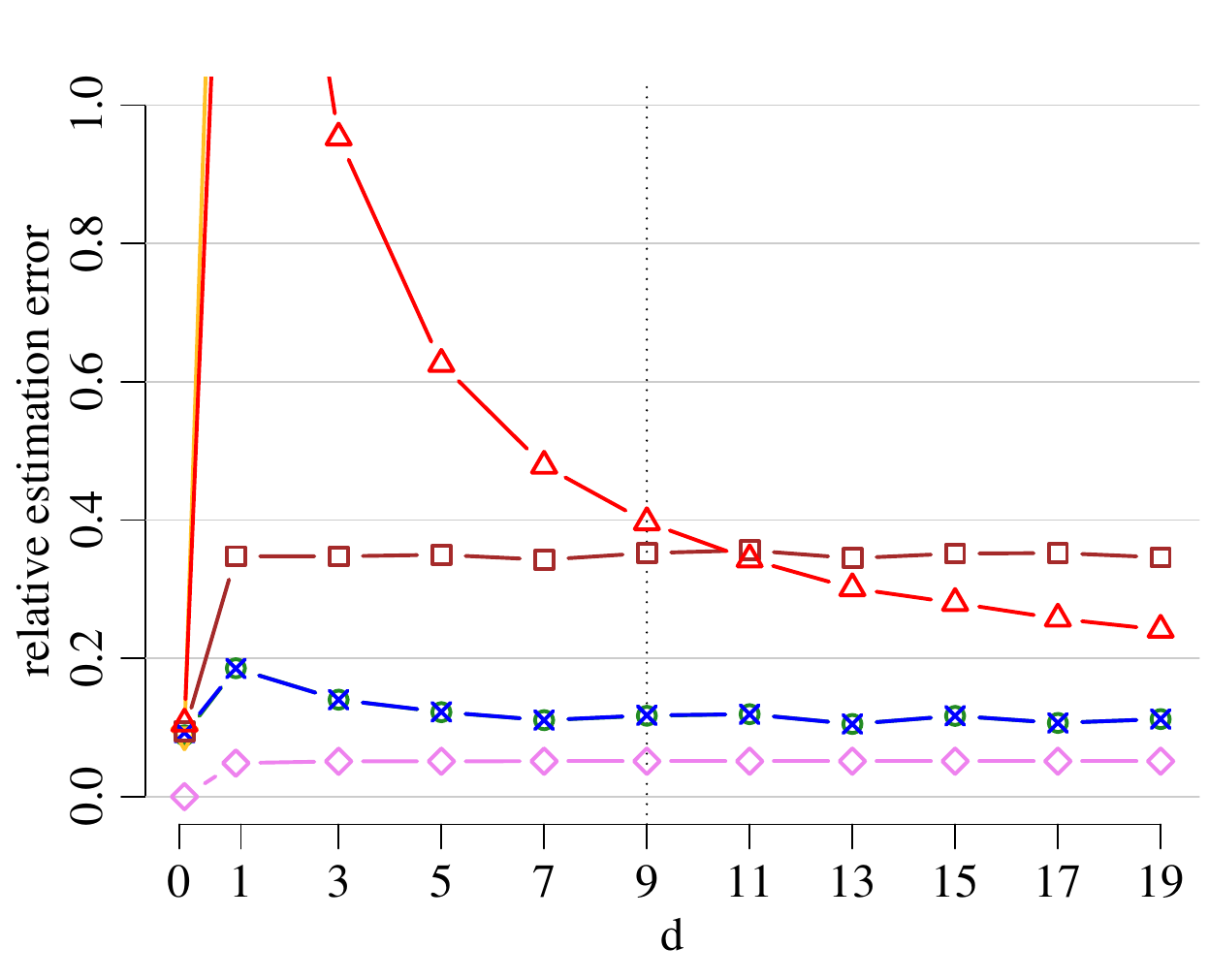} &
   \includegraphics[width=0.32\textwidth]{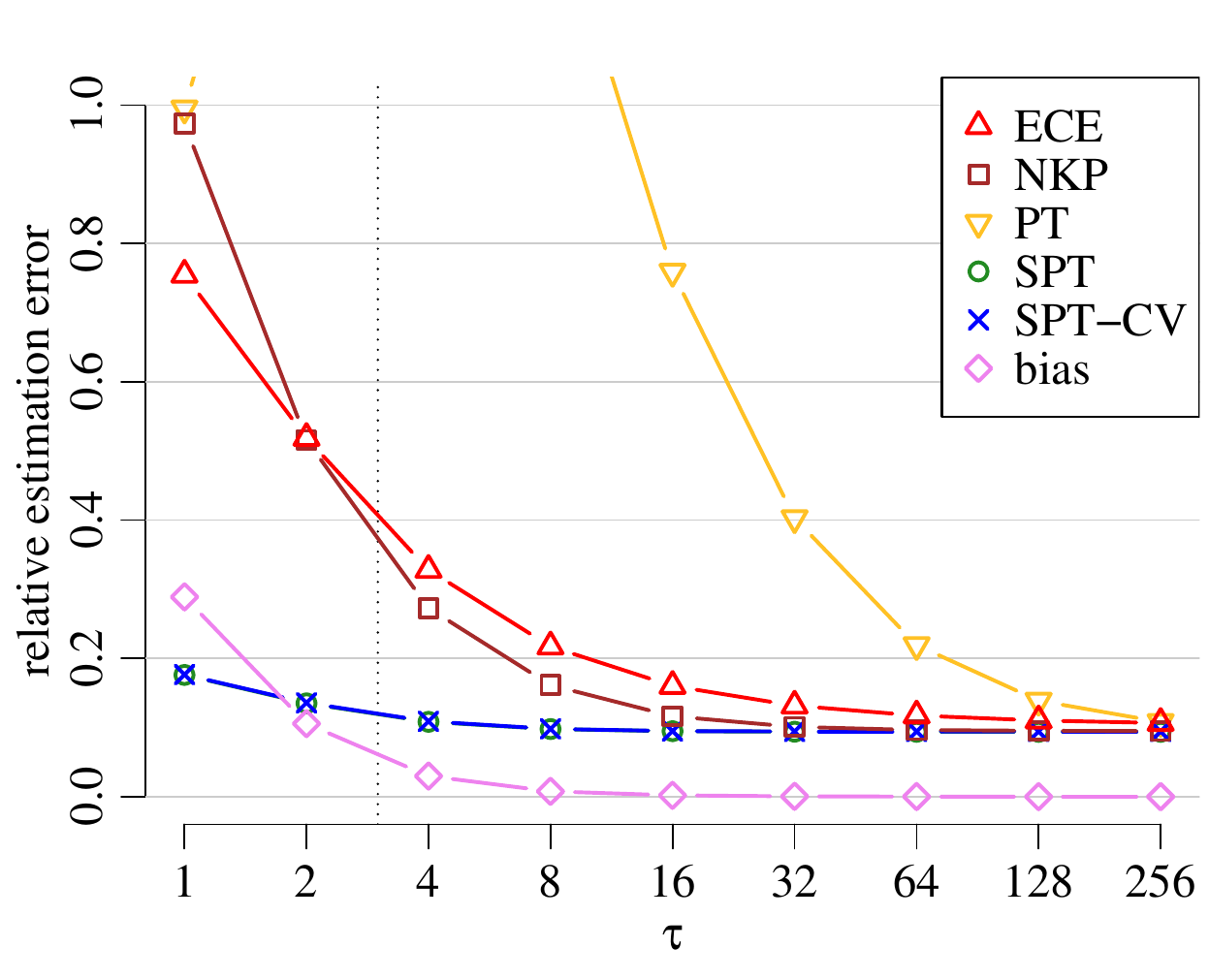} &
   \includegraphics[width=0.32\textwidth]{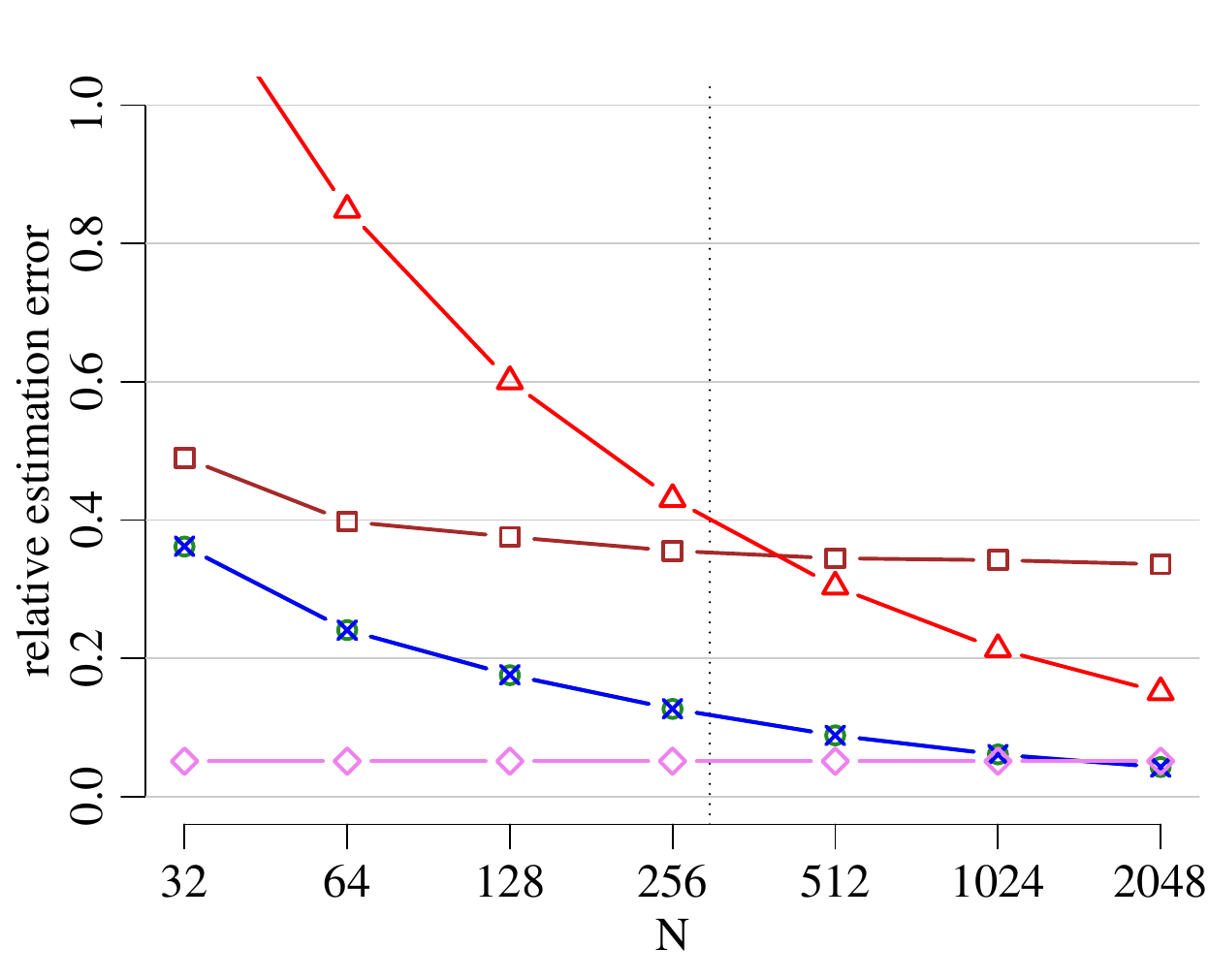} \\
   \includegraphics[width=0.32\textwidth]{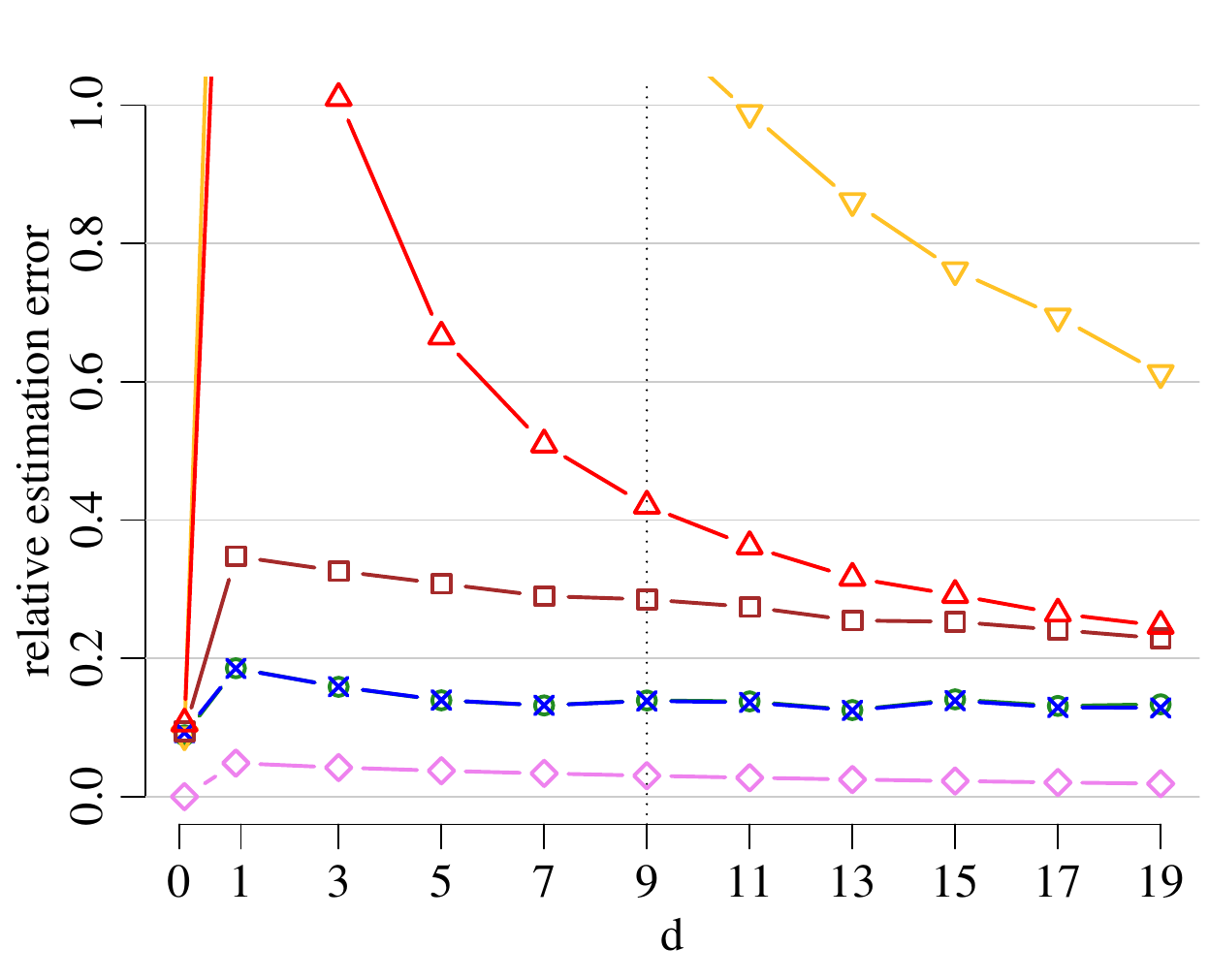} &
   \includegraphics[width=0.32\textwidth]{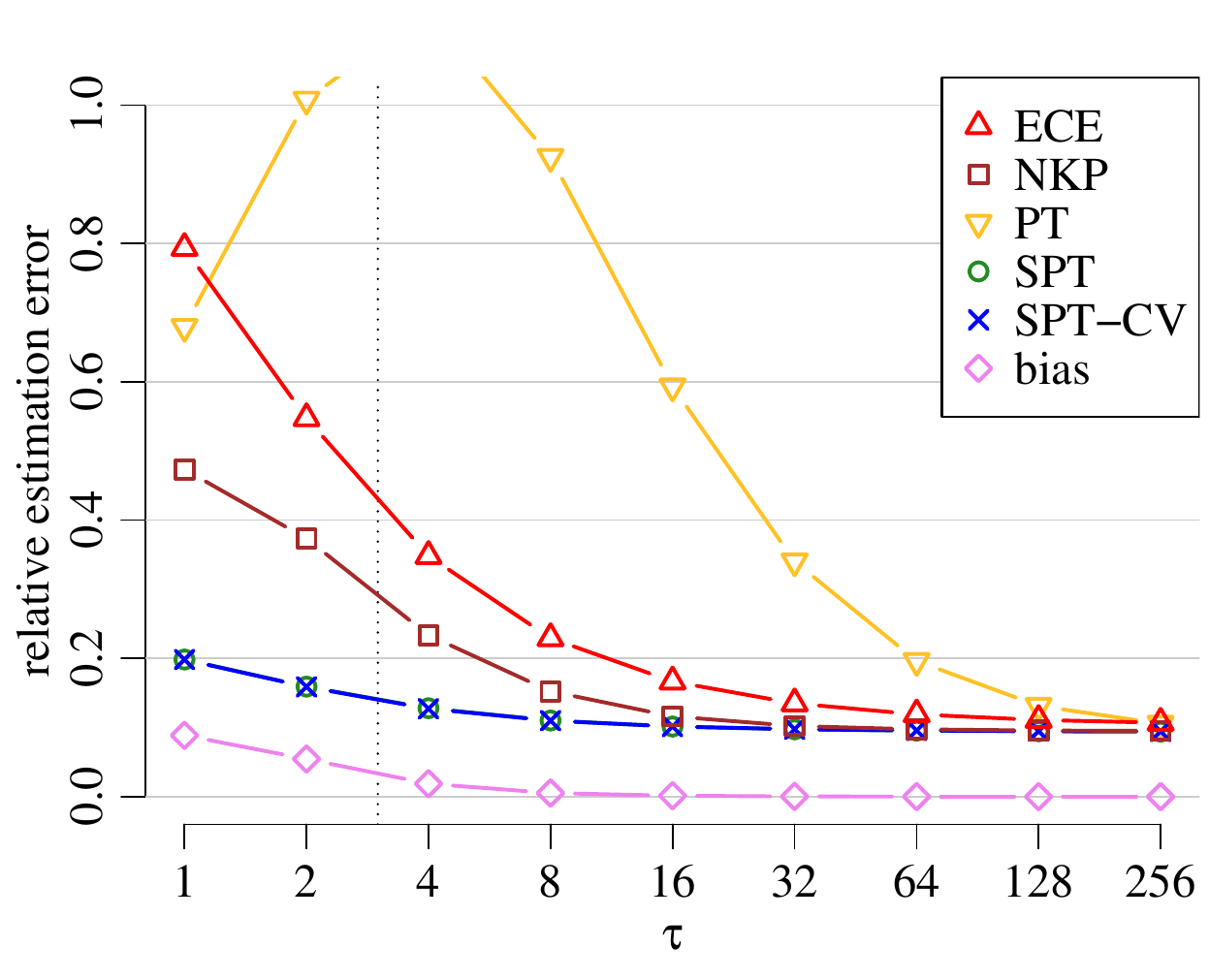} &
   \includegraphics[width=0.32\textwidth]{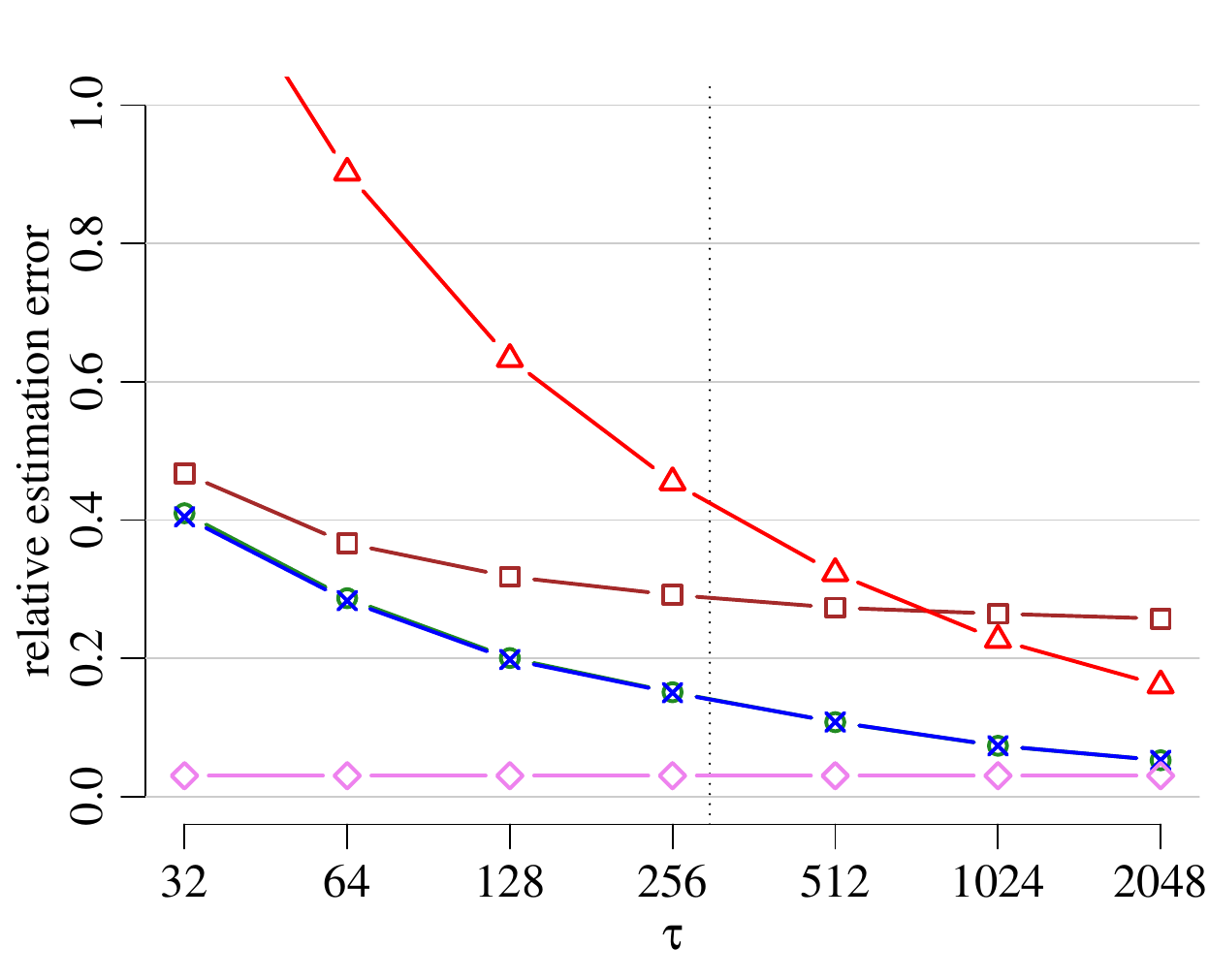} 
   \end{tabular}  
   \caption{Analogous to Figure \ref{fig:buf} of the main paper, i.e. relative estimation errors with varying bandwidth $d$ (\emph{left column}), signal-to-noise ratio $\tau$ (\emph{middle column}), or sample size $N$ (\emph{right column}). The \emph{top row} corresponds to the Legendre-Epanechnikov scenario, the \emph{middle row} depicts the Brownian-signed scenario, and the  \emph{bottom row} captures the Brownian-Epanechnikov case. The black dotted vertical lines show where the active parameter is fixed for the remaining two plots (i.e. $d=9$, $\tau=3$, and $N=300$), i.e. all the plots are roughly the same at the black dotted vertical cuts.}
    \label{fig:buf2} 
\end{figure}

Firstly, note that the Wiener case seems to be slightly easier than the Legendre case, despite the fact the former covariance is non-differentiable while the latter is analytic. Higher-order smoothness of the covariances does not play a role -- this should be expected since our theoretical development does not make any such assumptions. On the other hand, difficulty of the problem is governed by how fast $A_1$ and $A_2$ decay away from the diagonal. This is reflected in the theory by the assumption of non-zero shifted traces $\Tr{}^d(A_1)$ and $\Tr{}^d(A_2)$. While the shifted traces being non-zero suffices for asymptotic purposes, the finite-sample performance of our methodology is poor if the shifted traces past the true bandwidth are very small. This is quite natural: if $\mathbf B$ covers almost all the mass of the separable component, the latter cannot be estimated reliably.
This happens more easily in the Legendre case, because the Legendre covariances are more diagonally concentrated than the Wiener covariance.

Secondly, the signed choice for the banded part seems to make the problem harder than the Epanechnikov choice. While this cannot be visualized, one can imagine that the shape of $\mathbf B$ mimics the shape of the separable part better with the Epanechnikov choice. One can also observe this by looking at the ``bias'' error curves, notice that the bias is generally smaller in the Epanechnikov case compared to the signed case. Hence in the Epanechnikov case, the covariance can be approximated with the assumption of separability much better, making the problem easier.

Finally, we noted in the main paper that choosing a smaller bandwidth can sometimes be beneficial. This happens mainly when relatively large true bandwidth leads to only mild amount of non-separability, as happens in the right part of the top-left plot in Figure \ref{fig:buf2}. Focusing specifically at $d=15$ in the top-left plot in Figure \ref{fig:buf2}, we see an instance where performance of the proposed methodology with adaptively chosen bandwidth outperforms both the separable model and the separable-plus-banded model with the true (oracle) choice of the bandwidth. This is because, as suggested by Figure \ref{fig:setup_banded} (right), the effective bandwidth is in this case smaller than the true bandwidth (because the symbol of $\mathbf B$ decays fast away from the diagonal), while not completely ignoring the banded part is still beneficial.

\subsection{Goodness-of-fit Testing}\label{app:I}

In this section, we develop a testing procedure to check validity of the separable-plus-banded model, generalizing the bootstrap separability test of \cite{aston2017}.

We begin by reviewing the seminal test of \cite{aston2017}. For a covariance $C \in \mathcal{S}_2(\mathcal{H})$ with $\mathcal{H} = \mathcal{H}_1 \otimes \mathcal{H}_2$, a separable proxy is given by
\[
C_1 \ct C_2 = \frac{\Tr{1}(C) \ct \Tr{2}(C)}{\Tr{}(C)},
\]
where $\Tr{1}$ and $\Tr{2}$ are partial traces, i.e. shifted partial traces with the zero shifts. Plugging in the empirical covariance estimator, we obtain
\[
\widehat{C}_1 \ct \widehat{C}_2 = \frac{\Tr{1}(\widehat{C}_N) \ct \Tr{2}(\widehat{C}_N)}{\Tr{}(\widehat{C}_N)}
\]
a separable estimator of the covariance. Testing for separability is now based on the following operator
\begin{equation}\label{eq:distance_to_sep}
D_N = \widehat{C}_N  -\widehat{C}_1 \ct \widehat{C}_2.
\end{equation}
Under the hypothesis of separability, the norm of $D_N$ (a distance to separability) converges to zero as $N \to \infty$.

While a test can be based directly on the asymptotic distribution of $D_N$ (given as a special case of Theorem 1 of the main paper), such a test would require full calculation of the empirical covariance, and even worse calculation of the asymptotic variance, which is an eight-dimensional structure. Hence \cite{aston2017} propose to test separability only on a subspace of $\mathcal{S}_2(\mathcal{H})$ and use bootstrap to avoid calculation of the asymptotic variance. Namely, let $\mathcal{U} = \mathrm{span}\{ u_1, \ldots, u_m \} \subset \mathcal{H}$, then $\mathcal{U} \otimes \mathcal{U}$ determines a subspace of $\mathcal{S}_2(\mathcal{H})$ via isometry. Let $T_{\mathcal{U} \otimes \mathcal{U}}$ denote the orthogonal projection to the subspace $\mathcal{U} \otimes \mathcal{U}$. Then we have
\begin{equation}\label{eq:test_stat}
\verti{T_{\mathcal{U} \otimes \mathcal{U}} D_N}_2^2 = \sum_{r=1}^m \langle u_r , D_N u_r \rangle^2.
\end{equation}
The most natural choice of $\mathcal{U}$ is given by the eigenfunctions of the separable estimator. There are two reasons for this. Firstly, since we constrain the subspace on which separability will be tested, it makes sense to focus on the subspace on which further analysis (e.g. PCA) will likely be performed. Secondly, using the separable eigenfunctions allows for a fast calculation of the test statistic. Let $\widehat{C}_1 = \sum \widehat{\lambda}_j \widehat{e}_j \otimes \widehat{e}_j$ and $\widehat{C}_2 = \sum \widehat{\gamma}_j \widehat{f}_j \otimes \widehat{f}_j$ be the eigendecompositions and let $\mathcal{U} = \mathrm{span}\{ \widehat{e}_i \otimes \widehat{f}_j; i=1,\ldots,I, j=1,\ldots,J\}$. Then we have
\[
\verti{T_{\mathcal{U} \otimes \mathcal{U}} D_N}_2^2 = \sum_{i=1}^I \sum_{j=1}^J \langle \widehat{e}_i \otimes \widehat{f}_j , D_N \; \widehat{e}_i \otimes \widehat{f}_j \rangle^2 = \sum_{i=1}^I \sum_{j=1}^J  \left(\frac{1}{N} \sum_{n=1}^N \langle X_n, \widehat{e}_i \otimes \widehat{f}_j \rangle^2 - \widehat{\lambda}_j \widehat{\gamma}_j \right)^2 .
\]
As for the bootstrap, Aston et al. \cite{aston2017} proposed to approximate the distribution of $\verti{T_{\mathcal{U} \otimes \mathcal{U}} D_N}_2^2$ by
\begin{equation}\label{eq:boot_stat_sep}
\verti{T_{\mathcal{U} \otimes \mathcal{U}} (D_N - D_N^\star) }_2^2,
\end{equation}
where $D_N^\star$ is the distance-to-separability operator calculated based on a bootstrap sample $\{ X_1^\star, \ldots, X_N^\star \}$ drawn from the set $\{ X_1, \ldots, X_N \}$ with replacement. The reason for using bootstrap statistic \eqref{eq:boot_stat_sep} instead of simply $\verti{T_{\mathcal{U} \otimes \mathcal{U}} (D_N^\star) }_2^2$ is that the latter would approximate the distribution of $\verti{T_{\mathcal{U} \otimes \mathcal{U}} D_N}_2^2$ under the true $C$, i.e. not necessarily under the null, which is that the $C$ is separable.

It is natural to modify the separability test described above by changing the definition of the distance-to-separability \eqref{eq:distance_to_sep}. From now on, let
\begin{equation}\label{eq:distance}
D_N = \widehat{C}_N  - \widehat{A}_1 \ct \widehat{A}_2 - \widehat{B},
\end{equation}
where $\widehat{A}_1$, $\widehat{A}_2$ and $\widehat{B}$ are the estimators proposed in the main paper \eqref{eq:estimator_A} and \eqref{eq:estimator_B}. The test statistic is still taken as \eqref{eq:test_stat}, where $\mathcal{U}$ is still given by the eigenfunctions of the separable part $\widehat{A}_1 \ct \widehat{A}_2$. We still approximate the distribution of \eqref{eq:distance} by the bootstrap statistic
\begin{equation}\label{eq:boot_stat}
\verti{T_{\mathcal{U} \otimes \mathcal{U}} (D_N - D_N^\star) }_2^2,
\end{equation}
where $D_N^\star$ is calculated like $D_N$ in \eqref{eq:distance} from a bootstrapped sample. Drawing e.g. $10^3$ bootstrap samples, the bootstrapped $p$-value is given by
\[
\frac{1}{10^3+1} \sum_{m=1}^{10^3} \mathds{1}_{\left[ \verti{T_{\mathcal{U} \otimes \mathcal{U}} (D_N - D_{N,m}^\star) }_2^2 > \verti{T_{\mathcal{U} \otimes \mathcal{U}} D_N}_2^2 \right]},
\]
where $D_{N,m}$ is calculated like in \eqref{eq:distance} from the $m$-th bootstrapped sample. The statistic can be evaluated efficiently using the linear structure:
\[
\verti{T_{\mathcal{U} \otimes \mathcal{U}} D_N}_2^2 = \sum_{i=1}^I \sum_{j=1}^J  \left(\frac{1}{N} \sum_{n=1}^N \langle X_n, \widehat{e}_i \otimes \widehat{f}_j \rangle^2 - \widehat{\lambda}_j \widehat{\gamma}_j - \langle \widehat{e}_i \otimes \widehat{f}_j, \widehat{B} \; \widehat{e}_i \otimes \widehat{f}_j \rangle \right)^2 ,
\]
where $\widehat{B} \widehat{e}_i \otimes \widehat{f}_j$ is calculated using the fast Fourier transform. To implement this test, we modified the codes available in the \textsf{covsep} package \cite{tavakoli2016}.

The procedure outlined above allows for goodness-of-fit testing of the separable-plus-banded model using the ideas of \cite{aston2017}. It can be vaguely though of as testing whether separable model holds outside of a band.

\bibliographystyle{imsart-nameyear}
\bibliography{biblio}

\end{document}